\documentclass[11pt,letterpaper]{article}

\usepackage{times,natbib}
\usepackage{mysetup,xspace}
\usepackage{mathrsfs}
\usepackage{verbatim}

\usepackage{algorithm,caption}
\usepackage[noend]{algorithmic}
\newcommand{\myCaption}[1]{\caption*{{\bf Algorithm} #1}}

\newcommand{\ie}{{\em i.e.,~\xspace}}
\newcommand{\eg}{{\em e.g.,~\xspace}}

\usepackage[suppress]{color-edits}
\addauthor{ab}{blue}   
\addauthor{bk}{magenta} 
\addauthor{as}{red}     

\usepackage{hyperref}

\FULLPAGE

\newcommand{\preDiscr}{preadjusted discretization\xspace}

\title{Bandits with Knapsacks%
\footnote{
An extended abstract of this paper \citep{BwK-focs13} was published in \emph{IEEE FOCS 2013}.
\newline \indent
This paper has undergone several rounds of revision since the original ``full version" has been published on {\tt arxiv.org} in May 2013. Presentation has been thoroughly revised throughout, from low-level edits to technical intuition and details to high-level restructuring of the paper. Some of the results have been improved: a stronger regret bound in one of the main results (Theorem~\ref{thm:Balance}), and a more general example of \preDiscr for dynamic pricing (Theorem~\ref{thm:DynProcurement-discretization}). Introduction discusses a significant amount of follow-up work, and is up-to-date regarding open questions.
\newline \indent Parts of this research have been done while A.~Badanidiyuru was a research intern at Microsoft Research and a graduate student at Cornell University, and while R.~Kleinberg was a Consulting Researcher at Microsoft Research.
\newline \indent A.~Badanidiyuru was partially supported by NSF grant IIS-0905467. R.~Kleinberg was partially supported by NSF grants CCF-0643934, IIS-0905467 and AF-0910940, a Microsoft Research New Faculty Fellowship, and a Google Research Grant.}}

\author{Ashwinkumar Badanidiyuru
\footnote{Google Research, Mountain View CA, USA.
Email: {\tt ashwinkumarbv@gmail.com}. }
\and Robert Kleinberg
\footnote{Department of Computer Science, Cornell University, Ithaca NY, USA. Email: {\tt rdk@cs.cornell.edu}.}
\and Aleksandrs Slivkins
\footnote{Microsoft Research, New York NY, USA. Email: {\tt slivkins@microsoft.com}.}
}
\date{May 2013\\This revision: September 2017}

\begin{document}

\maketitle

\newcommand{\OptPolicy}{optimal dynamic policy\xspace}
\newcommand{\PreDiscr}{Preadjusted discretization\xspace}

\renewcommand{\F}{\mathbf{\Delta}} 

\renewcommand{\eqref}[1]{Equation~(\ref{#1})}
\renewcommand{\E} {\operatornamewithlimits{\ensuremath{\mathbb{E}}}} 
\renewcommand{\Re}{\mathbb{R}}
\newcommand{\D}{\mathcal{D}}
\newcommand{\mM}{\mathcal{M}} 
\newcommand{\mF}{\mathcal{F}} 

\newcommand{\AlgFont}[1]{\ensuremath{\mathtt{#1}}}
\newcommand{\ALG}{\AlgFont{ALG}\xspace} 
\newcommand{\BwK}{\AlgFont{BwK}\xspace} 
\newcommand{\pdbwk}{\AlgFont{PrimalDualBwK}\xspace}
\newcommand{\kMAB}{\AlgFont{BalancedExploration}\xspace}

\newcommand{\arms}{X}
\newcommand{\domain}{\mathcal{M}_{\mathtt{feas}}} 
\newcommand{\OPT}{\mathtt{OPT}}
\newcommand{\LPOPT}[1][\mathtt{LP}]{\OPT_{#1}}
\newcommand{\Rew}{\mathtt{REW}} 
\newcommand{\LP}{\mathtt{LP}}
\newcommand{\UCB}{\mathtt{UCB}}
\newcommand{\LCB}{\mathtt{LCB}}



\newcommand{\SalesRate}{F}  


\newcommand{\empir}[1]{\widehat{#1}} 
\newcommand{\aveD}{\bar{\D}}    
\newcommand{\rad}{\mathtt{rad}} 
\newcommand{\chernoffC}{C_{\mathtt{rad}}} 

\newcommand{\fbt}{\frac{B}{T}} 
\newcommand{\tfbt}{\tfrac{B}{T}}

\newcommand{\Dpx}{\D_{p,x}}
\newcommand{\Dqx}{\D_{q,x}}

\newcommand{\empirdtx}[2]{\empir{\D}_{#1,#2}}
\newcommand{\avedtx}[2]{\aveD_{#1,#2}}



\newcommand{\narms}{{m}}
\newcommand{\armset}{{X}}
\newcommand{\arm}{{x}}
\newcommand{\nrsc}{{d}}
\newcommand{\rscset}{{\mathscr R}}
\newcommand{\constraint}{{\mathscr P}}
\newcommand{\optrwd}{{\OPT}}
\newcommand{\optlp}{{\LPOPT}}
\newcommand{\optucb}{{\Rew_{\mathtt{UCB}}}}
\newcommand{\algrwd}{{\Rew}}
\newcommand{\budg}{{B}}
\newcommand{\horizon}{{T}}

\newcommand{\splx}[1]{{\mathbf{\Delta}[#1]}}
\newcommand{\vctr}[1]{{\mathbf{#1}}}
\newcommand{\sbv}{{\vctr{e}}} 
\newcommand{\stime}{{\tau}}
\newcommand{\trans}{{\intercal}}
\newcommand{\ones}{{\mathbf{1}}}
\newcommand{\diag}[1]{{\operatorname{Diag} \{ #1 \}}}

\newenvironment{lparray}%
{\begin{array}{l@{\hspace{8mm}}l@{\hspace{8mm}}l}}%
{\end{array}}
\newlength{\lplb}
\setlength{\lplb}{3mm}

\newcommand{\npul}{N}
\newcommand{\vnpul}{{\vctr{\npul}}}
\newcommand{\crlz}[1]{{C}_{#1}}
\newcommand{\cemp}[1]{{\overline{C}}_{#1}}
\newcommand{\crad}[1]{{W}_{#1}}
\newcommand{\clcb}[1]{{L}_{#1}} 
\newcommand{\cdif}[1]{{E}_{#1}}
\newcommand{\rrlz}[1]{{r}_{#1}}
\newcommand{\remp}[1]{{\overline{r}_{#1}}}
\newcommand{\rrad}[1]{{w}_{#1}}
\newcommand{\rucb}[1]{{u}_{#1}} 
\newcommand{\rdif}[1]{{\delta}_{#1}}
\newcommand{\aemp}[1]{{\overline{a}_{#1}}}
\newcommand{\arad}[1]{{w}_{#1}}
\newcommand{\asum}{{A}}
\newcommand{\vrad}{{\overrightarrow{\rad}}}

\newcommand{\expthree}{{\texttt{Exp3}}}
\newcommand{\hedge}{{\texttt{{Hedge}}}}
\newcommand{\ucbone}{{\texttt{{UCB1}}}}

\newcommand{\x}{{z}}
\newcommand{\adv}[1]{\ensuremath{\mathcal{I}_{#1}}\xspace}
\newcommand{\alladv}{\mathcal{A}}
\newcommand{\reg}{\ensuremath{\mathtt{REG}}}
\newcommand{\Null}{\ensuremath{\mathtt{null}}}
\newcommand{\ucblp}{\ensuremath{\mathfrak{M}}} 


\newcommand{\Otilde}{\widetilde{O}}

\begin{abstract}

Multi-armed bandit problems are the predominant theoretical model of exploration-exploitation tradeoffs in learning, and they have countless applications ranging from medical trials, to communication networks, to Web search and advertising. In many of these application domains the learner may be constrained by one or more supply (or budget) limits, in addition to the customary limitation on the time horizon. The literature lacks a general model encompassing these sorts of problems. We introduce such a model, called \emph{bandits with knapsacks}, that combines bandit learning with aspects of stochastic integer programming. In particular, a bandit algorithm needs to solve a stochastic version of the well-known \emph{knapsack problem}, which is concerned with packing items into a limited-size knapsack. A distinctive feature of our problem, in comparison to the existing regret-minimization literature, is that the optimal policy for a given latent distribution may significantly outperform the policy that plays the optimal fixed arm. Consequently, achieving sublinear regret in the bandits-with-knapsacks problem is significantly more challenging than in conventional bandit problems.

We present two algorithms whose reward is close to the information-theoretic optimum: one is based on a novel ``balanced exploration'' paradigm, while the other is a primal-dual algorithm that uses multiplicative updates. Further, we prove that the regret achieved by both algorithms is optimal up to polylogarithmic factors. We illustrate the generality of the problem by presenting applications in a number of different domains, including electronic commerce, routing, and scheduling. As one example of a concrete application, we consider the problem of dynamic posted pricing with limited supply and obtain the first algorithm whose regret, with respect to the \OptPolicy, is sublinear in the supply.

\end{abstract}

\newpage
\setcounter{tocdepth}{2}
\tableofcontents

\newpage

\section{Introduction}
\label{sec:intro}

For more than fifty years, the multi-armed bandit problem (henceforth, \emph{MAB}) has
been the predominant theoretical model for sequential decision
problems that embody the tension between exploration and
exploitation, ``the conflict between taking
actions which yield immediate reward and taking actions
whose benefit (\eg acquiring information or preparing the ground) will come only later,'' to quote Whittle's
apt summary~\citep{whittle80}. Owing to the universal
nature of this conflict, it is not surprising that
MAB algorithms have found diverse applications
ranging from medical trials, to
communication networks, to Web search and advertising.

A common feature in many of these application domains is the presence of one or more limited-supply resources that are consumed during the decision process. For example, scientists experimenting with alternative medical treatments may be limited not only by the number of patients participating in the study but also by the cost of materials used in the treatments. A website experimenting with displaying advertisements is constrained not only by the number of users who visit the site but by the advertisers' budgets. A retailer engaging in price experimentation faces inventory limits along with a limited number of consumers. The literature on MAB problems lacks a general model that encompasses these sorts of decision problems with supply limits. Our paper contributes such a model, called \emph{bandits with knapsacks} (henceforth \BwK), in which a bandit algorithm needs to solve a stochastic, multi-dimensional version of the well-known \emph{knapsack problem}. We present algorithms whose regret (normalized by the payoff of the optimal policy) converges to zero as the resource budget and the optimal payoff tend to infinity. In fact, we prove that this convergence takes place at the information-theoretically optimal rate.

\subsection{Our model: bandits with knapsacks (\BwK)}
\label{sec:problem-description}

\xhdr{Problem definition.} A learner has a fixed set of potential actions, a.k.a. \emph{arms}, denoted
by $\armset$ and called \emph{action space}. (In our main results,
$\armset$ will be finite, but we will also consider extensions with an infinite set of arms, see Section~\ref{sec:apps} and Section~\ref{sec:discretization}.)
There are $d$ resources being consumed by the learner.
Over a sequence of time steps,
the learner chooses an arm and observes two things:
a \emph{reward} and a \emph{resource consumption vector}.
Rewards are scalar-valued, whereas resource consumption
vectors are $\nrsc$-dimensional:
the $i$-th component represents consumption of resource $i$.
For each resource $i$ there is a pre-specified
\emph{budget} $B_i$ representing the maximum amount
that may be consumed, in total.
The process stops at the first time $\stime$ when the total
consumption of some resource exceeds its budget. The objective is to maximize
the total reward received before time $\stime$.

We assume that the environment does not change over time. Formally,
the observations for a fixed arm $\arm$ in each time step
(\ie the reward and resource consumption vector) are independent samples
from a fixed joint distribution
 on $[0,1] \times [0,1]^d$,
called the \emph{latent distribution} for arm $\arm$.

There is a known, finite time horizon $T$. We model it as one of the resources, one unit of which is deterministically consumed in each decision period, and the budget is $T$.

\xhdr{Notable examples.}
The conventional MAB problem, with a finite time horizon $T$, naturally fits into this framework. A more interesting example is the \emph{dynamic pricing} problem faced by a retailer selling $B$ items to a population of $T$ unit-demand consumers who arrive sequentially. Modeling this as a \BwK problem, rounds correspond to consumers, and arms correspond to the possible prices which may be offered to a consumer. Reward is the revenue from a sale, if any.  Resource consumption vectors express the number of items sold and consumers seen, respectively. Thus, if a price $p$ is offered and accepted, the reward is $p$ and the resource consumption is {\tiny $\begin{bmatrix} 1 \\ 1 \end{bmatrix}$}. If the offer is declined, the reward is $0$ and the resource consumption is {\tiny $\begin{bmatrix} 0 \\ 1 \end{bmatrix}$}.


\asedit{A “dual” problem of dynamic pricing is \emph{dynamic procurement}, where the algorithm is “dynamically buying” rather than “dynamically selling”. The reward refers to the number of bought items, and the budget constraint $B$ now applies to the amount spent (which is why the two problems are not merely identical up to sign reversal).  If a price $p$ is offered and accepted, the reward is $1$ and the resource consumption is
    {\tiny $\begin{bmatrix} p \\ 1 \end{bmatrix}$}.
If the offer is declined, the reward is $0$ and the resource consumption is
    {\tiny $\begin{bmatrix} 0 \\ 1  \end{bmatrix}$}.
This problem is also relevant to the domain of crowdsourcing: the items “bought” then correspond to microtasks ordered on a crowdsourcing platform such as Amazon Mechanical Turk.

Another simple example concerns \emph{dynamic ad allocation} for pay-per-click ads with unknown click probabilities. There is one advertiser with several ads and budget $B$ across all ads, and $T$ users to show the ads to. The ad platform allocates one ad to a new user in each round. Whenever a given ad $x$ is chosen and clicked on, the advertiser pays a known amount $\pi_x$. To model this as a \BwK problem, arms correspond to ads, rewards are the advertiser's payments, and resource consumption refers to the amount spent by the advertiser and the number of users seen. Thus, if ad $x$ is chosen and clicked, the reward is $\pi_x$ and the resource consumption is
    {\tiny $\begin{bmatrix} \pi_x \\ 1 \end{bmatrix}$};
otherwise, the reward is $0$ and the resource consumption is
    {\tiny $\begin{bmatrix} 0 \\ 1  \end{bmatrix}$}.

All three examples can be easily generalized to multiple resource constraints: resp., to selling multiple products, procuring different types of goods, and allocating ads from multiple advertisers.
} 

\xhdr{Benchmark and regret.}
The performance of an algorithm will be measured by its
\emph{regret}: the worst case, over all possible tuples of
latent distributions, of the
difference between $\OPT$ and the algorithm's expected total reward. Here $\OPT$ is the
expected total reward of the benchmark: an \emph{\OptPolicy}, an algorithm that maximizes expected total reward given foreknowledge of the latent distributions.

In a conventional MAB problem,
the \OptPolicy is to play a fixed arm, namely the one with the highest expected
reward. In the \BwK problem, the \OptPolicy is more complex, as
the choice of an arm in a given round depends on the remaining supply of each resource.
In fact, we doubt there is a polynomial-time algorithm to compute the
\OptPolicy given the latent distributions; similar problems in optimal control have long been
known to be PSPACE-hard~\citep{PapaTsit99}.

It is easy to see that the \OptPolicy may significantly out-perform the best fixed arm. To take a simple example, consider a problem instance with $d$ resources and $d$ arms such that pulling arm $i$ deterministically produces a reward of $1$, consumes one unit of resource $i$, and does not consume any other resources. We are given an initial endowment of $\budg$ units of each resource. Any policy that plays a fixed arm $i$ in each round is limited to a total reward of $\budg$ before running out of its budget of resource $i$. Whereas an algorithm that alternates arms in a round-robin fashion achieves reward $d\budg$: $d$ times larger. Similar, but somewhat more involved examples can be found for application domains of interest, see Appendix~\ref{ap:example}. Interestingly, in all these examples it suffices to consider a \emph{time-invariant mixture} of arms, \ie a policy that samples in each period from a fixed probability distribution over arms regardless of the remaining resource supplies. In particular, in the simple example above it suffices to consider a uniform distribution.

\ascomment{The following xhdr is moved from elsewhere.}

\xhdr{Alternative definitions.}
More generally we could model the budget constraints as a downward-closed polytope $\constraint \subset \R_+^{\nrsc}$ such that the process stops when the sum of resource consumption vectors is no longer in $\constraint$. However, our assumption that $\constraint$ is a box constraint is virtually without loss of generality. If $\constraint$ is instead specified by a system of inequalities $\{A x \preceq b\}$, we can redefine the resource consumption vectors to be $Ax$ instead of $x$ and then the budget constraint is the box constraint defined by the vector $b$. The only potential downside of this transformation is that it increases the dimension of the resource vector space, when the constraint matrix $A$ has more rows than columns. However, one of our algorithms has regret depending only logarithmically on $\nrsc$, so this increase typically has only a mild effect on regret.

Our stopping condition halts the algorithm as soon as any budget is exceeded. Alternatively, we could restrict the algorithm to actions that cannot possibly violate any constraint if chosen in the current round, and stop if there is no such action. This alternative is essentially equivalent to the original version: each budget constraint changes by at most one, which does not affect our regret bounds in any significant way.

\subsection{Main results}

We seek regret bounds that are sublinear in $\optrwd$, whereas in analyzing MAB algorithms one typically expresses regret bounds as a sublinear  function of the time horizon $T$.   This is because a regret guarantee of the form $o(T)$ may be unacceptably weak for the \BwK problem because supply limits prevent the \OptPolicy from achieving a reward close to $T$. An illustrative example is the dynamic pricing problem with supply $B \ll T$: the seller can only sell $B$ items, each at a price of at most 1, so bounding the regret by any number greater than $B$ is worthless. To achieve sublinear regret, the algorithm must be able to explore each arm a significant number of times without exhausting its resource budgets. Accordingly, we parameterize our regret bound by $\budg=\min_i \budg_i$, the smallest budget constraint.%

\xhdr{Algorithms.}
We present an algorithm, called \pdbwk, whose regret is sublinear in
$\optrwd$ as both $\optrwd$ and $\budg$ tend to infinity.
More precisely, denoting the number of arms by $\narms$,
our algorithm's regret is
\begin{align}\label{eq:intro-regret}
\Otilde\left( \sqrt{\narms\, \optrwd} + \optrwd \sqrt{\narms/\budg}
\;\;\right),
\end{align}
where the $\Otilde()$ notation hides logarithmic factors.
Note that without resource constraints, \ie setting $B=T$, we recover regret
    $\Otilde(\sqrt{\narms \OPT})$,
which is optimal up to $\log$ factors~\citep{bandits-exp3}.  In fact, we prove a slightly stronger regret bound which has an optimal scaling property: if all budget constraints, including the time horizon, are increased by the factor of $\alpha$, then the regret bound scales as $\sqrt{\alpha}$.%
\footnote{The square-root scaling is optimal even for the basic MAB problem, as proved in~\citet{bandits-exp3}.}
The algorithm is computationally efficient, in a strong sense: with machine word size of $\log T$ bits or more, the per-round running time is $O(\narms \nrsc)$.  Moreover, if each arm $j$ consumes only $\nrsc_j$ resources that are known in advance, then the per-round running time is
    $O(\narms + \nrsc+\sum_j \nrsc_j)$.

We also present another algorithm, called \kMAB, whose regret bound is the same up to logarithmic factors for $d=O(1)$. \asedit{The regret bounds for the two algorithms are incomparable: while \pdbwk achieves a better dependence on $d$, \kMAB performs better in some special cases, see Appendix~\ref{app:Balance-wins} for a simple example. While \pdbwk is very computationally efficient, the specification of \kMAB involves a mathematically well-defined optimization step for which we do not provide a specific implementation, see Remark~\ref{rem:Balance-implementation} fur further discussion.
} 

\xhdr{Lower bound.} We provide a matching lower bound: we prove that the regret bound~\refeq{eq:intro-regret} is optimal up to polylogarithmic factors; moreover, this holds \emph{for any given tuple of parameters}.  Specifically, we show that for any given tuple $(\narms,\budg,\OPT)$, any algorithm for \BwK must incur regret
\begin{align}\label{eq:LB-intro}
\Omega\left( \min\left( \OPT,\;
    \OPT \sqrt{\narms/\budg}
    +\sqrt{\narms\,\OPT}\right)\right),
\end{align}
in the worst-case over all instances of \BwK with these $(\narms,\budg,\OPT)$. We also show that this dependence on the \emph{smallest} budget constraint is inevitable in the worst case.

\xhdr{Applications and special cases.}
\asedit{We derive corollaries for the three examples outlined in Section~\ref{sec:problem-description}:
\begin{itemize}
\item We obtain regret $\Otilde(B^{2/3})$ for the basic version of dynamic pricing. This is optimal for each $(B,T)$ pair \citep{DynPricing-ec12}. Prior work \citep{DynPricing-ec12,Wang-OR14} achieved $\Otilde(B^{2/3})$ regret w.r.t. the best fixed price, and $\Otilde(\sqrt{B})$ regret assuming ``regularity".%
    \footnote{``Regularity" is a standard (but limiting) condition which states that the mapping from prices to expected rewards is concave.} The former result is much weaker than ours, see Appendix~\ref{ap:example} for a simple example, and the latter result is incomparable.

\item We obtain regret $\Otilde(T/B^{1/4})$ for the basic version of dynamic procurement. Prior work \citep{DynProcurement-ec12} achieves a constant-factor approximation to the optimum with a prohibitively large constant (at least in the tens of thousands), so our result is a big improvement unless
        $\OPT\gg T/B^{1/4}$.

\item We obtain regret $\Otilde(\sqrt{B})$ for the basic version of dynamic ad allocation. This is optimal when $B=T$ (\ie when the budget constraint is void), by the basic $\sqrt{T}$ lower bound for MAB.
\end{itemize}

Our model admits numerous generalizations of these three examples, as well as applications to several other domains. To emphasize the generality of our contributions, we systematically discuss applications and corollaries in Section~\ref{sec:apps}.
Pointers to prior work on special cases of \BwK can be found in Section~\ref{sec:related-work}.
} 

\subsection{Challenges and techniques}
\label{sec:intro-techniques}

\ascomment{New subsection: assembled content from elsewhere.}

\xhdr{Challenges.}
As with all MAB problems, a central issue in \BwK is the tradeoff between exploration and exploitation. A na\"ive way to resolve this tradeoff is to separate exploration and exploitation: before the algorithm starts, the rounds are partitioned into ``exploration rounds" and ``exploitation rounds", so that the arms chosen in the former does not depend on the feedback, and the feedback from the latter is discarded.%
\footnote{While the intuition behind this definition has been well-known for some time, the precise definition is due to \citet{MechMAB-ec09,DevanurK09}.}
For example, an algorithm may pick an arm uniformly at random for a pre-defined number of rounds, then choose the best arm given the observations so far, and stick to this arm from then on. However, it tends to be much more efficient to combine exploration and exploitation by adapting the exploration schedule to observations. Typically in such algorithms all but the first few rounds serve both exploration and exploitation. Thus, one immediate challenge is to implement this approach in the context of \BwK.

The \BwK problem is significantly more difficult to solve than conventional MAB problems for the following three reasons. First, in order to estimate the performance of a given time-invariant policy, one needs to estimate the expected \emph{total} reward of this policy, rather than the per-round expected reward (because the latter does not account for resource constraints). Second, since exploration consumes resources other than time, the negative effect of exploration is not limited to the rounds in which it is performed. Since resource consumption is stochastic, this negative effect is not known in advance, and can only be estimated over time. Finally, and perhaps most importantly, the \OptPolicy can significantly outperform the best fixed arm, as mentioned above. In order to compete with the \OptPolicy, an algorithm needs, essentially, to search over mixtures of arms rather than over arms themselves, which is a much larger search space. In particular, our algorithms improve over the performance of the best fixed arm, whereas algorithms for explore-exploit learning problems typically do not.%
\footnote{A few notable exceptions are in \citep{bandits-exp3,BanditSurveys-colt13,BesbesZeevi-or12,DynProcurement-ec12}. Of these, \citet{BesbesZeevi-or12} and \citet{DynProcurement-ec12} are on special cases of \BwK, and are discussed later.
}

\xhdr{Our algorithms.} Algorithm \kMAB explicitly optimizes over mixtures of arms, based on a simple idea: balanced exploration inside confidence bounds. The design principle underlying many confidence-bound based algorithms for stochastic MAB, including the famous $\ucbone$ algorithm~\citep{bandits-ucb1} and our algorithm \pdbwk, is generally, ``Exploit as much as possible, but use confidence bounds that are wide enough to encourage some exploration.'' The design principle in \kMAB, in contrast, could be summarized as, ``Explore as much as possible, but use confidence bounds that are narrow enough to eliminate obviously suboptimal alternatives.'' Our algorithm balances exploration across arms, exploring \emph{each arm} as much as possible given the confidence bounds. More specifically, there are designated rounds when the algorithm picks a mixture that approximately maximizes the probability of choosing this arm, among the mixtures that are not obviously suboptimal given the current confidence bounds.

\OMIT{The exploration of this arm is \emph{balanced} -- not too little and not too much -- as the probability of selecting this arm is sandwiched between the maximum (among the mixtures that are not obviously suboptimal) and half of that value.}

\OMIT{ 
Intriguingly, the algorithm
matches the logarithmic regret bound of $\ucbone$ for stochastic
MAB (up to constant factors) and
achieves qualitatively similar bounds to \pdbwk for \BwK, despite being based on a design
principle that is the polar opposite of those algorithms.
We believe that the ``balanced exploration'' principle
underlying \kMAB is of independent interest.
} 

Algorithm \pdbwk is a primal-dual algorithm based
on the multiplicative weights update method. It
maintains a vector of ``resource costs'' that is
adjusted using multiplicative updates. In every
period it estimates each arm's expected reward and
expected resource consumption, using upper confidence
bounds for the former and lower confidence bounds for
the latter; then it plays the most ``cost-effective''
arm, namely the one with the highest ratio of estimated
resource consumption to estimated resource cost, using
the current cost vector. Although confidence bounds
and multiplicative updates are the bread and butter of
online learning theory, we consider this way of combining
the two techniques to be quite novel. In particular,
previous multiplicative-update algorithms in online
learning theory --- such as the $\expthree$ algorithm
for MAB~\citep{bandits-exp3} or the
weighted majority~\citep{LittWarm94} and $\hedge$~\citep{FS97}
algorithms for learning from expert
advice --- applied multiplicative updates to the
probabilities of choosing different arms (or experts).
Our application of multiplicative updates to the
dual variables of the LP relaxation of \BwK is
conceptually quite a different usage of this technique.

\asedit{Having alternative techniques to solve the same problem is generally useful in a rich problem space such as MAB. Indeed, one often needs to apply techniques beyond the original models for which they were designed, perhaps combining them with techniques that handle other facets of the problem. When pursuing such extensions, some alternatives may be more suitable than others, in particular because they are more compatible with the other techniques. We already see examples of that in the follow-up work: \citet{AgrawalDevanur-ec14} and \citet{cBwK-colt14} use some of the techniques from \kMAB and \pdbwk, resp., see Section~\ref{sec:further} for more details.}

\xhdr{LP-relaxation.} In order to compare our algorithms to $\OPT$, we compare both to a more tractable benchmark given by time-invariant mixtures of arms. More precisely, we define a linear programming relaxation for the expected total reward achieved by a time-invariant mixture of arms, and prove that the optimal value $\LPOPT$ achieved by this LP-relaxation is an upper bound for $\OPT$. Therefore it suffices to relate our algorithms to the time-invariant mixture of arms that achieves $\LPOPT$, and bound their regret with respect to $\LPOPT$.

\xhdr{Lower bounds.}
The lower bound \refeq{eq:LB-intro} is based on a simple example in which all arms have reward $1$ and 0-1 consumption of a single resource, and one arm has slightly smaller expected resource consumption than the rest. To analyze this example, we apply the KL-divergence technique from the MAB lower bound in \citet{bandits-exp3}. Some technical difficulties arise,  compared to the derivation in \citet{bandits-exp3}, because the arms are different in terms of the expected consumption rather than expected reward, and because we need to match the desired value for $\OPT$.

\xhdr{Discretization.}
In some applications, such as dynamic pricing and dynamic procurement, the action space $X$ is very large or infinite, so our main algorithmic result is not immediately applicable. However, the action space has some structure that our algorithms can leverage: \eg a price is just a number in some fixed interval. To handle such applications, we discretize the action space: we apply a \BwK algorithm with a restricted, finite action space $S\subset X$, where $S$ is chosen in advance. Immediately, we obtain a bound on regret with respect to the \OptPolicy restricted to $S$. \asedit{Further, we select $S$ so as to balance the tradeoff between $|S|$ and the \emph{discretization error}: the decrease in the performance benchmark due to restricting the action space to $S$. We call this approach \emph{\preDiscr}. While it has been used in prior work, the key step of bounding the discretization error is now considerably more difficult, as one needs to take into account resource constraints and argue about mixtures of arms rather than individual arms.

We bound discretization error for subset $S$ which satisfies certain axioms, and apply this result to handle dynamic pricing with a single product and dynamic procurement with a single budget constraint. While the former application is straightforward, the latter takes some work and uses a non-standard mesh of prices. Bounding the discretization error for more than one resource constraints (other than time) appears to be much more challenging; we only achieve this for a special case.
} 

\OMIT{ 
Algorithm \kMAB is \emph{domain-aware}, in the sense that it explicitly optimizes over all latent distributions that are feasible for a given ``BwK domain", such as dynamic pricing with limited supply. We do not provide a computationally efficient implementation for this optimization. In fact, it is not clear what model of computation would be appropriate to characterize access to the domain knowledge. Efficient implementation of \kMAB may be possible for some specific BwK domains, although we do not pursue that direction in this paper. (Recall that \pdbwk, on the other hand, is very computationally efficient.)
} 

\OMIT{ 
The regret bounds for the two algorithms are incomparable: while \pdbwk achieves a better dependence on $d$, \kMAB performs better in some special cases by virtue of being domain-aware. We illustrate this on a version of the simple deterministic example from Section~\ref{sec:problem-description} (see Appendix~\ref{app:Balance-wins}).
} 

\subsection{Follow-up work and open questions}
\label{sec:further}

Since the \BwK problem provides a novel general problem formulation in online learning, it lends itself to a rich set of research questions in a similar way as the stochastic MAB problem did following \citet{Lai-Robbins-85} and \citet{bandits-ucb1}. Some of these questions were researched in the follow-up work.

\xhdr{Follow-up work.}
Following the conference publication of this paper \citep{BwK-focs13}, there have been several developments directly \asedit{inspired by} \BwK.

\citet{AgrawalDevanur-ec14} extend \BwK from hard resource constraints and additive rewards to a more general model that allows penalties and diminishing returns. In particular, the time-averaged outcome vector $\bar{v}$ is constrained to lie in an arbitrary given convex set, and the total reward can be an arbitrary concave, Lipschitz-continuous function of $\bar{v}$. They provide several algorithms for this model whose regret scales optimally as a function of the time horizon. \asedit{Remarkably, these algorithms specialize to \emph{three} new algorithms for \BwK, based on different ideas. One of these new \BwK algorithms follows the ``optimism under uncertainty" approach from \citep{bandits-ucb1} (with an additional trick of rescaling the resource constraints). Despite the apparent simplicity, it is shown to satisfy our main regret bound \refeq{eq:intro-regret}.}

\citet{cBwK-colt14} extend \BwK to \emph{contextual bandits}: a bandit model where in each round the ``context" is revealed (\eg a user profile), then the algorithm selects an arm, and the resulting outcome (in our case, reward and resource consumption) depends on both the chosen arm and the context. \citet{cBwK-colt14} merge \BwK and \emph{contextual bandits with policy sets} \citep{Langford-nips07}, a well-established, very general model for contextual bandits. They achieve regret that scales optimally in terms of the time horizon and the number of policies (resp., square-root and logarithmic). Akin to \kMAB, their algorithm is not computationally efficient.

Both \citet{AgrawalDevanur-ec14} and \citet{cBwK-colt14} take advantage of various techniques developed in this paper. First, both papers use (a generalization of) linear relaxations from Section~\ref{sec:lp}. In fact, the two claims in Section~\ref{sec:lp} are directly used in~\citet{cBwK-colt14} to derive the corresponding statements for the contextual version. Second, \citet{cBwK-colt14} build on the design and analysis of \kMAB, and merging them with a technique from prior work on contextual bandits \citep{policy_elim}. Third, the analysis of one of the algorithms in \citet{AgrawalDevanur-ec14} relies on the bound on error terms (Lemma~\ref{lem:vector-conf}) from our analysis of \pdbwk. Fourth, the analysis of discretization errors in \citet{cBwK-colt14} uses a technique from Section~\ref{sec:discretization}.

Two recent developments, \citet{CBwK-colt16} and \citet{CBwK-nips16}, concern the contextual version of \BwK. \citet{CBwK-colt16} consider a common generalization of the extended \BwK model in \citep{AgrawalDevanur-ec14} and the contextual \BwK model in  \citep{cBwK-colt14}. In particular, for the latter model they achieve the same regret as \citet{cBwK-colt14}, but with a computationally efficient algorithm, resolving the main open question in that paper. On a technical level, their work combines ideas from \citep{AgrawalDevanur-ec14} and a recent break-through in contextual bandits \citep{monster-icml14}. \citet{CBwK-nips16} extend the model in \citep{AgrawalDevanur-ec14} to contextual bandits with a linear dependence on contexts (\eg see \cite{Reyzin-aistats11-linear}), achieving an algorithm with optimal dependence on the time horizon and the dimensionality of contexts.%
\footnote{\citet{AgrawalDevanur-ec14} prove a similar result for a special case when contexts do not change over time. They also claimed an extension to time-varying contexts, which has subsequently been retracted (see Footnote 1 in \citet{CBwK-nips16}).}

\xhdr{Open questions (current status).}
While the general regret bound in \eqref{eq:intro-regret} is optimal up to logarithmic factors,
better algorithms may be possible for various special cases. To rule out a domain-specific result that improves upon the general regret bound, one would need to prove a lower bound which, unlike the one in \eqref{eq:LB-intro}, is specific to that domain. Currently domain-specific lower bounds are known only for the basic $K$-armed bandit problem and for dynamic pricing.

For problems with infinite multi-dimensional action spaces, such as dynamic pricing with multiple products and dynamic procurement with multiple budgets, we are limited by the lack of a general approach to upper-bound the discretization error and choose the \preDiscr in a principled way.  A similar issue arises in the contextual extension of \BwK studied in \citet{cBwK-colt14} and \citet{CBwK-colt16}, even for a single resource constraint. To obtain regret bounds that do not depend on a specific choice of \preDiscr, one may need to go beyond \preDiscr.

The study of multi-armed bandit problems with
large strategy sets
has been a very fruitful line of investigation.
It seems likely that some of the techniques introduced
here could be wedded with the techniques from
that literature. In particular, it would be intriguing
to try combining our primal-dual algorithm \pdbwk with
confidence-ellipsoid algorithms for stochastic linear optimization
(\eg see \citet{DaniHK-colt08}), or enhancing the
\kMAB algorithm with the technique of adaptively
refined discretization, as in the zooming
algorithm of \citet{LipschitzMAB-stoc08}.

It is tempting to ask about a version of \BwK in which the rewards and resource consumptions are chosen by an adversary. Achieving sublinear regret bounds for this version appears hopeless even for the  fixed-arm benchmark. In order to make progress in the positive direction, one may require  a more subtle notion of benchmark and/or restrictions on the power of the adversary.

\OMIT{ 
Finally, some open questions concern the contextual extension of \BwK studied in \citet{cBwK-colt14} and \citet{CBwK-colt16}. First, it is not clear how to bound discretization error and choose \preDiscr for contextual dynamic pricing/procurement.  Second, the regret bounds in~\cite{cBwK-colt14} and \citet{CBwK-colt16}, while optimal in the worst case, may be significantly sub-optimal for some application domains and/or for some (classes of) tuples $(\narms,\budg,\OPT)$.
} 

\subsection{Related work}
\label{sec:related-work}

The study of prior-free algorithms for stochastic
MAB problems was initiated by \citet{Lai-Robbins-85} and \citet{bandits-ucb1}. Subsequent work
supplied algorithms for stochastic MAB
problems in which the set of arms can be infinite and the
payoff function is linear,
concave,
or Lipschitz-continuous;
see a recent survey \citep{Bubeck-survey12} for more background.
Confidence bound techniques have been an integral part
of this line of work, and they remain integral to ours.

As explained earlier, stochastic MAB problems
constitute a very special
case of bandits with knapsacks, in which there is only one type
of resource and it is consumed deterministically at rate 1.
Several papers have considered the
natural generalization in which there is a single resource (other than time), with
deterministic consumption, but different arms consume the resource
at different rates. \citet{GuhaM-icalp09}
gave a constant-factor
approximation algorithm for the Bayesian case of this problem,
which was later generalized by \citet{GuptaKMR-focs11}
to settings in which the arms' reward processes need not be martingales.
\citet{TranThanh-aaai10,TranThanh-aaai12} presented prior-free algorithms for this problem; the best such algorithm achieves a regret guarantee qualitatively similar to
that of the $\ucbone$ algorithm.

\asedit{ 
Several recent papers study models that, in hindsight, can be cast as special cases of \BwK:
\begin{itemize}
\item The two papers \citep{TranThanh-aaai10,TranThanh-aaai12} mentioned above and \citet{Qin-aaai13} consider models with a single resource and unlimited time.

\item Dynamic pricing with limited supply has been studied in \citep{BZ09,DynPricing-ec12,BesbesZeevi-or12,Wang-OR14}.%
    \footnote{The earlier papers~\citep{Blum03,Bobby-focs03} focus on the special case of unlimited supply. While we only cited papers that pursue regret-minimizing formulation of dynamic pricing, Bayesian and parametric formulations versions have a rich literature in Operations Research and Economics, see \citet{Boer-survey15} for a literature review.}

\item The basic version of dynamic procurement (as per Section~\ref{sec:problem-description}) has been studied in  \citep{DynProcurement-ec12,Krause-www13}.%
    \footnote{The regret bound in \citep{Krause-www13} is against the best-fixed-price benchmark, which may be much smaller than $\OPT$, see Appendix~\ref{ap:example} for a simple example. Benchmarks aside, one cannot directly compare our regret bound and theirs, because they do not derive a worst-case regret bound. \citep{Krause-www13} is simultaneous work w.r.t. our conference publication.}
    More background on the connection to crowdsourcing can be found in the survey \citet{Crowdsourcing-PositionPaper13}.

    \item Dynamic ad allocation (without budget constraints) and various extensions thereof that incorporate user/webpage context have received a considerable attention, starting with \citep{yahoo-bandits07,yahoo-bandits-icml07,Langford-nips07}. In fact, the connection to pay-per-click advertising has been one of the main drivers for the recent surge of interest in MAB.

    \item \citep{AminK-uai12,TranThanh-uai14} study repeated bidding on a budget, and \citet{RepeatedAuctions-soda13} study adjusting a repeated auction (albeit without inventory constraints); see Section~\ref{sec:apps} for more details on these special cases.

    \item Perhaps the earliest paper on resource consumption in MAB is \citet{Gyorgy-ijcai07}. They consider a contextual bandit model where the only resource is time, consumed at different rate depending on the context and the chosen arm. The restriction to a single context is a special case of \BwK.

\end{itemize}

\PreDiscr has been used in prior work on MAB on metric spaces (\eg \citep{Bobby-nips04,Hazan-colt07,LipschitzMAB-stoc08,Pal-Bandits-aistats10}) and dynamic pricing
(\eg \citep{Bobby-focs03,Blum03,BZ09,DynPricing-ec12}). However, bounding the discretization error in \BwK is much more difficult.
} 

\asedit{Our \kMAB algorithm extends the ``active arms elimination" algorithm \citep{EvenDar-colt02} for the stochastic MAB problem, where one iterates over arms that are not obviously suboptimal given the current confidence bounds . The novelty is that our algorithm chooses over \emph{mixtures} of arms, and the choice is ``balanced" across arms. ``Policy elimination" algorithm of \citet{policy_elim} extends ``active arms elimination" in a different direction: to contextual bandits. Like \kMAB, policy elimination algorithm makes a ``balanced" choice among objects that are more complicated than arms, and this choice is not computationally efficient; however, the technical details are very different.}

While \BwK is primarily an online learning problem, it also has
elements of a stochastic packing problem. The literature on
prior-free algorithms for stochastic packing has flourished
in recent years, starting
with prior-free algorithms for the stochastic AdWords
problem \citep{DevanurH-ec09}, and continuing with a series
of papers extending these results from AdWords to more
general stochastic packing
integer programs while also achieving stronger performance
guarantees \citep{AgrawalWY-OR14,DevanurJSW-ec11,FeldmanHKMS-esa10,MolinaroR-icalp12}. A running theme of these papers (and also of the primal-dual
algorithm in this paper) is the idea of estimating of an optimal dual vector
from samples, then using this dual to guide subsequent primal decisions.
Particularly relevant to our work is the algorithm
of \citet{DevanurJSW-ec11}, in which the dual vector
is adjusted using multiplicative updates, as we do in our algorithm.
However, unlike the \BwK problem, the stochastic packing
problems considered in prior work are not learning problems:
they are full information problems in which the costs and
rewards of decisions in the past and present are fully known.
The only uncertainty is about the future.) As such, designing
algorithms for \BwK requires a substantial departure from past
work on stochastic packing. Our primal-dual algorithm depends upon
a hybrid of confidence-bound techniques from online learning
and primal-dual techniques from the literature on solving packing
LPs; combining them requires
entirely new techniques for bounding the magnitude of the error
terms that arise in the analysis. Moreover, our
\kMAB algorithm manages to
achieve strong regret guarantees without even computing
a dual solution.

\section{Preliminaries}
\label{sec:prelims}

{\noindent \bf \BwK: problem formulation.} There is a fixed and known, finite set of $\narms$ \emph{arms} (possible actions), denoted $\arms$. There are $\nrsc$ resources being consumed. The time proceeds in $T$ rounds, where $T$ is a finite, known time horizon. In each round $t$, an algorithm picks an arm $x_t\in \arms$, receives reward $r_t\in [0,1]$, and consumes some amount $c_{t,i} \in [0,1]$ of each resource $i$. The values $r_t$ and $c_{t,i}$ are revealed to the algorithm after the round. There is a hard constraint $B_i\in \Re_+$ on the consumption of each resource $i$; we call it a \emph{budget} for resource $i$.
The algorithm stops at the earliest time $\stime$ when one or more
budget constraint is violated; its total reward is equal to the sum
of the rewards in all rounds strictly preceding $\stime$.
The goal of the algorithm is to maximize the expected total reward.

The vector $(r_t; c_{t,1}, c_{t,2} \LDOTS c_{t,d})\in [0,1]^{d+1}$ is called the \emph{outcome vector} for round $t$. We assume \emph{stochastic outcomes}: if an algorithm picks arm $x$, the outcome vector is chosen independently from some fixed distribution $\pi_x$ over $[0,1]^{d+1}$.
The distributions $\pi_x$, $x\in \arms$ are not known to the algorithm. The tuple $(\pi_x:\, x\in X)$ comprises all latent information in the problem instance. A particular \BwK setting (such as ``dynamic pricing with limited supply'') is defined by the set of all feasible tuples $(\pi_x:\, x\in X)$. This set, called the \emph{BwK domain}, is known to the algorithm.

We compare the performance of our algorithms to the expected total reward of the \OptPolicy given all the latent information, which we denote by $\OPT$. (Note that $\OPT$ depends on the latent information, and therefore is a latent quantity itself.) \emph{Regret} is defined as $\OPT$ minus the expected total reward of the algorithm.

\xhdr{W.l.o.g. assumptions.}
For technical convenience, we make several assumptions that are w.l.o.g.

We express the time horizon as a resource constraint: we model time as a specific resource, say resource $1$, such that every arm deterministically consumes $B_1/T$ units of this resource whenever it is picked.
W.l.o.g., $B_i \leq T$ for every resource $i$.

We assume there exists an arm, called the \emph{null arm} which yields no reward and no consumption of any resource other than time. Equivalently, an algorithm is allowed to spend a unit of time without doing anything. Any algorithm $\ALG$ that uses the null arm can be transformed, without loss in expected total reward, to an algorithm $\ALG'$ that does not use the null arm. Indeed, in each round $\ALG'$ runs $\ALG$ until it selects a non-null arm $x$ or halts. In the former case, $\ALG'$ selects $x$ and returns the observe feedback to $\ALG$. After $\ALG$ halts, $\ALG'$ selects arms arbitrarily.

We say that the budgets are \emph{uniform} if $B_i=B$ for each resource $i$. Any \BwK instance can be reduced to one with uniform budgets by dividing all consumption values for every resource $i$ by $B_i/B$, where $B=\min_i B_i$. (That is tantamount to changing the units in which we measure consumption of resource $i$.) Our technical results are for \BwK with uniform budgets. We will assume uniform budgets $B$ from here on.

\xhdr{Useful notation.}
Let $\mu_x = \E[\pi_x] \in [0,1]^{d+1}$ be the expected outcome vector for each arm $x$, and denote $\mu = (\mu_x:\, x\in X)$. We call $\mu$ the \emph{latent structure} of a problem instance. The BwK domain induces a set of feasible latent structures, which we denote $\domain$.

For notational convenience, we will write
    $\mu_x = \left(\; r(x,\mu);\; c_1(x,\mu) \LDOTS c_d(x,\mu) \; \right)$.
Also, we will write the expected consumption as a vector
    $c(x,\mu) = \left(\; c_1(x,\mu) \LDOTS c_d(x,\mu) \; \right)$.

If $\D$ is a distribution over arms, let
    $r(\D,\mu) = \sum_{x\in \arms} \D(x)\, r(x,\mu)$
and
    $c(\D,\mu) = \sum_{x\in \arms} \D(x)\,c(x,\mu)$
be, respectively, the expected reward and expected resource consumption in a single round if an arm is sampled from distribution $\D$. Let $\Rew(\D,\mu)$ denote the expected total reward of the time-invariant
policy that uses distribution $\D$.

\OMIT{ 
When discussing our primal-dual algorithm, it will be useful
to represent the latent values and the algorithm's decisions
as matrices and vectors. For this purpose, we will number the
arms as $\arm_1,\ldots,\arm_{\narms}$ and let $r \in \R^{\narms}$
denote the vector whose $j^{\mathrm{th}}$ component is $r(\arm_j,\mu)$.
Similarly we will
let $C \in \R^{\nrsc \times \narms}$ denote the matrix whose
$(i,j)^{\mathrm{th}}$ entry is $c_i(\arm_j,\mu)$.
} 


\xhdr{High-probability events.}
We will use the following expression, which we call the \emph{confidence radius}.
\begin{align}\label{eq:conf-rad-defn}
 \rad(\nu,N) =
        \sqrt{\frac{\chernoffC\, \nu }{N}} + \frac{\chernoffC}{N}.
\end{align}
Here $\chernoffC=\Theta(\log (d\,T|X|))$ is a parameter which we will fix later; we will keep it implicit in the notation. The meaning of \eqref{eq:conf-rad-defn} and $\chernoffC$ is explained by the following tail inequality from \citep{LipschitzMAB-stoc08,DynPricing-ec12}.%
\footnote{Specifically, this follows from Lemma 4.9 in the full version of \citet{LipschitzMAB-stoc08}, and Theorem 4.8 and Theorem 4.10 in the full version of \citet{DynPricing-ec12} (both full versions can be found on arxiv.org).}

\begin{theorem}[\cite{LipschitzMAB-stoc08,DynPricing-ec12}]
\label{thm:conf-rad}
Consider some distribution with values in $[0,1]$ and expectation $\nu$. Let $\empir{\nu}$ be the average of $N$ independent samples from this distribution. Then
\begin{align}\label{eq:thm:conf-rad}
\Pr\left[ \; |\nu - \empir{\nu}|
    \leq \rad(\empir{\nu},N)
    \leq 3\,\rad(\nu,N) \; \right] \geq 1-e^{-\Omega(\chernoffC)},
    \quad \text{for each $\chernoffC>0$}.
\end{align}
More generally, \eqref{eq:thm:conf-rad} holds if $X_1, \ldots, X_N \in [0,1]$ are random variables,
    $\empir{\nu} = \tfrac{1}{N} \sum_{t=1}^N X_t$
is the sample average, and
    $\nu = \tfrac{1}{N} \sum_{t=1}^N\; \E[X_t \,|\, X_1,\,\ldots, X_{t-1}]$.
\end{theorem}

If the expectation $\nu$ is a latent quantity, \eqref{eq:thm:conf-rad} allows us to estimate $\nu$ by a high-confidence interval
\begin{align}\label{eq:nu-conf-int}
    \nu \in [\empir{\nu} - \rad(\empir{\nu},N),\; \empir{\nu} + \rad(\empir{\nu},N)],
\end{align}
whose endpoints are observable (known to the algorithm). This estimate is on par with the one provided by Azuma-Hoeffding inequality (up to constant factors), but is much sharper for small $\nu$.%
\footnote{Essentially, Azuma-Hoeffding inequality states that
$|\nu - \empir{\nu}| \leq O(\sqrt{\chernoffC/N})$,
whereas by Theorem~\ref{thm:conf-rad} for small $\nu$ it holds with high probability that
    $\rad(\empir{\nu},N)\sim \chernoffC/N$.}

It is sometimes useful to argue about \emph{any} $\nu$ which lies in the high-confidence interval \refeq{eq:nu-conf-int}, not just the latent $\nu = \E[\empir{\nu}]$. We use the following claim which is implicit in \citet{LipschitzMAB-stoc08}.

\begin{claim}[\cite{LipschitzMAB-stoc08}]\label{cl:rad-to-rad}
For any $\nu,\empir{\nu}\in [0,1]$, \eqref{eq:nu-conf-int} implies that
    $\rad(\empir{\nu},N) \leq 3\,\rad(\nu,N)$.
\end{claim}

\section{LP relaxation for policy value}
\label{sec:lp}

$\OPT$ --- the expected reward of the \OptPolicy given foreknowledge of the
distribution of outcome vectors --- is typically difficult to characterize
exactly. In fact, even for a time-invariant policy,
it is difficult to give an exact
expression for the expected reward due to the dependence of the reward
on the random stopping time when the resource budget is exhausted.
To approximate these quantities, we consider the fractional relaxation of
\BwK in which the number of rounds in which a given arm is selected (and also the total number of rounds) can be
fractional, and the reward and resource consumption per unit time are
deterministically equal to the corresponding expected values in the original
instance of \BwK.

The following linear program constitutes our fractional relaxation of the \OptPolicy.
\begin{align}\label{lp:primal}
\begin{array}{lrcll}
	\max         & \sum_{x\in X}\, \xi_x\, r(x,\mu) & &            & \text{in $\xi_x \in \R$, for each $x\in X$} \\
	\text{s.t.}  & \sum_{x\in X} \xi_x\, c_i(x,\mu)  & \leq & B  & \text{for each resource $i$} \\
			     & \xi_x     & \geq & 0                            & \text{for each arm $x$}.
\end{array}
\tag{\texttt{LP-primal}}
\end{align}

\noindent The variables $\xi_x$ represent the fractional relaxation for the number of rounds in which a given arm $x$ is selected. This is a bounded LP
(because $\sum_x\, \xi_x\, r(x,\mu) \leq  \sum_x\, \xi_x \leq T$).
The optimal value of this LP is denoted by $\LPOPT$. We will also use the dual LP, shown below.
\begin{align}\label{lp:dual}
\begin{array}{lrcll}
	\min         & B\,\sum_{i}\, \eta_i & &                    & \text{in $\eta_i \in \R$, for each resource $i$} \\
	\text{s.t.}  & \sum_{i} \eta_i\, c_i(x,\mu)  &\geq& r(x,\mu)  & \text{for each arm $x\in X$} \\
			     & \eta_i     & \geq & 0                           & \text{for each resource $i$}.
\end{array}
\tag{\texttt{LP-dual}}
\end{align}
The dual variables $\eta_i$ can be interpreted as a unit cost for the corresponding resource $i$.

\OMIT{ 
\begin{figure}[h]
\centering
\begin{minipage}[t]{0.25\linewidth}\centering
\begin{equation}
\begin{lparray}
\max & r^\trans \xi \\[\lplb]
\mbox{s.t.} & C \xi \preceq B \ones \\[\lplb]
& \xi \succeq 0
\end{lparray}
\tag{P}
\label{lp:primal}
\end{equation}
\end{minipage}
\hspace{0.9cm}
\begin{minipage}[t]{0.25\linewidth}\centering
\begin{equation}
\begin{lparray}
\min & B \ones^\trans \eta \\[\lplb]
\mbox{s.t.} & C^\trans \eta \succeq r \\[\lplb]
& \eta \succeq 0
\end{lparray}
\tag{D}
\label{lp:dual}
\end{equation}
\end{minipage}
\end{figure}
}

\begin{lemma} \label{lem:lp-relax}
$\LPOPT$ is an upper
bound on the value of the \OptPolicy: $\LPOPT \geq \OPT$.
\end{lemma}

One way to prove this lemma is to define $\xi_x$ to be the expected
number of times arm $x$ is played by the \OptPolicy,
and argue that the vector $(\xi_x, x\in X)$ is primal-feasible and that
    $\sum_x \xi_x\, r(x,\mu)$
is the expected reward of the \OptPolicy. We instead present a
simpler proof using \refeq{lp:dual} and a martingale argument. 
\asedit{A similar lemma (but for a technically different setting of online stochastic packing problems) was proved in \citet{DevanurJSW-ec11}.}

\begin{proof}[Proof of Lemma~\ref{lem:lp-relax}]
Let
    $\eta^* = (\eta^*_1 \LDOTS \eta^*_d)$
denote an optimal solution to~\refeq{lp:dual}.
Interpret each $\eta^*_i$ as a unit cost for the corresponding resource $i$. By
strong LP duality, we have
$B\,\sum_{i}\, \eta^*_i = \LPOPT$.
Dual feasibility implies that for each
arm $x$, the expected cost of resources consumed when $x$
is pulled exceeds the expected reward produced. Thus, if we let
$Z_t$ denote the sum of rewards gained in rounds $1,\ldots,t$ of the \OptPolicy,
plus the cost of the remaining resource endowment after round $t$,
then the stochastic process $Z_0,Z_1,\ldots,Z_T$ is a supermartingale.
Let $\stime$ be the stopping time of the algorithm, i.e. the total number of rounds.
Note that
    $Z_0 = B\, \sum_{i}\, \eta^*_i = \LPOPT$,
and
$Z_{\stime-1}$ equals the algorithm's total payoff, plus the cost of
the remaining (non-negative) resource supply at the start of
round $\stime$.
By Doob's optional stopping theorem, $Z_0 \geq \E[Z_{\stime-1}]$
and the lemma is proved.
\end{proof}

\begin{remark}\label{rem:LP-intuition}
\asedit{Implicit in this proof is a simple, but powerful observation that for any algorithm,
\[ \LPOPT-\Rew \geq 
    \textstyle \E \left[ \sum_t r(x_t,\mu) - c(x_t,\mu)\cdot \eta^* \right].
\]
Each summand on the right-hand side is non-negative, and equals 0 if and only if the arm $x_t$ lies in the support of the primal solution.
We use this observation to motivate the design of our primal-dual algorithm.}
\end{remark}

\begin{remark}\label{rem:regret-LPOPT}
For each of the two main algorithms, we prove a regret bound of the form
\begin{align}\label{eq:regret-LPOPT}
\LPOPT - \Rew \leq f(\LPOPT),
\end{align}
where $\Rew$ is the expected total reward of the algorithm, and $f()$ depends only on parameters $(\budg,\narms,\nrsc)$. This regret bound has an optimal scaling property, highlighted in the Introduction:  if all budget constraints, including the time horizon, are increased by the factor of $\alpha$, then the regret bound $f(\LPOPT)$ scales as $\sqrt{\alpha}$.

Regret bound \refeq{eq:regret-LPOPT} implies the claimed regret bounds relative to $\OPT$ because
\begin{align}\label{eq:regret-LPOPT-OPT}
\Rew \geq \LPOPT - f(\LPOPT) \geq \OPT-f(\OPT),
\end{align}
where the second inequality follows trivially because
    $g(x) = \max(x-f(x),0)$
is a non-decreasing function of $x$ for $x\geq 0$, and $\LPOPT\geq \OPT$.
\end{remark}

Let us apply a similar LP-relaxation to a time-invariant policy that uses distribution $\D$ over arms. We approximate the expected total reward of this policy in a similar way: we define a linear program in which the only variable $t$ represents the expected stopping time of the algorithm.
\begin{align}\label{eq:LP-D}
\begin{array}{lrcll}
    \max         & t\, r(\D,\mu) &     & & \text{in $t \in \R$} \\
	\text{s.t.}  & t\, c_i(\D,\mu)    & \leq & B & \text{for each resource $i$} \\
			             & t     & \geq & 0.
\end{array}
\tag{\texttt{LP-distr}}
\end{align}
The optimal value to~\refeq{eq:LP-D}, which we call the \emph{LP-value} of $\D$, is
\begin{align}\label{eq:LP-value}
\LP(\D,\mu)
    = r(\D,\mu)\; \min_i\; \frac{B}{c_i(\D,\mu)}.
\end{align}
Observe that $t$ is feasible for \refeq{eq:LP-D} if and only if
$\xi = t \D$ is feasible for~\refeq{lp:primal}. Therefore
    $$\LPOPT = \sup_\D \LP(\D,\mu).$$
This supremum is attained by any distribution
    $\D^* = \xi/ \, \|\xi\|_1$
such that $\xi = (\xi_x: x\in X)$ is an optimal solution to~\refeq{lp:primal}.
A distribution
    $\D^*\in \argmax_\D \LP(\D,\mu)$
is called \emph{LP-optimal} for $\mu$.

\begin{claim}
\label{cl:LP-properties}
For any latent structure $\mu$, there exists a distribution $\D$ over arms which is LP-optimal for $\mu$ and moreover satisfies the following three properties:
\begin{OneLiners}
\item[(a)] $c_i(\D,\mu) \leq B/T$ for each resource $i$.
\item[(b)] $\D$ has a support of size at most $d$.
\item[(c)] If $\D$ has a support of size exactly $2$ then for some resource $i$ we have $c_i(\D,\mu) = B/T$.
\end{OneLiners}
(Such distribution $\D$ will be called {\bf\em LP-perfect} for $\mu$.)
\end{claim}

\begin{proof}
Fix the latent structure $\mu$. It is a well-known fact that for any linear program there exists an optimal solution whose support has size that is exactly equal to the number of constraints that are tight for this solution. Take any such optimal solution $\xi = (\xi_x: x\in X)$ for \refeq{lp:primal}, and take the corresponding LP-optimal distribution
    $\D = \xi/ \|\xi\|_1$.
Since there are $d$ constraints in~\refeq{lp:primal}, distribution $\D$ has support of size at most $d$. If it satisfies (a), then it also satisfies (c) (else it is not optimal), and we are done.

Suppose property (a) does not hold for $\D$. Then there exists a resource $i$ such that $c_i(\D,\mu) > B/T$. Since the $i$-th constraint in \refeq{lp:primal} can be restated as
    $\|\xi\|_1\, c_i(\D,\mu) \leq B$,
it follows that $\|\xi\|_1 <T$. Therefore the constraint in \refeq{lp:primal} that expresses the time horizon is not tight. Consequently, at most $d-1$ constraints in \refeq{lp:primal} are tight for $\xi$, so the support of $\D$ has size at most $d-1$.

Let us modify $\D$ to obtain another LP-optimal distribution $\D'$ which satisfies properties (a-c). W.l.o.g., pick $i$ to maximize $c_i(\D,\mu)$ and let $\alpha = \tfrac{B}{T}/c_i(\D,\mu)$. Define $\D'(x) = \alpha\, \D(x)$ for each non-null arm $x$ and place the remaining probability in $\D'$ on the null arm. This completes the definition of $\D'$.

Note that
    $c_j(\D',\mu) = \alpha\, c_j(\D,\mu) \leq B/T$
for each resource $j$, with equality for $j=i$. Hence, $\D'$ satisfies properties (a) and (c). Also,
    $r(\D',\mu) = \alpha\, r(\D,\mu)$,
and so
\begin{align*}
\LP(\D',\mu)
    = r(\D',\mu)\; \tfrac{B}{c_i(\D',\mu)}
    = r(\D,\mu) \; \tfrac{B}{c_i(\D,\mu)}
    = \LP(\D,\mu).
\end{align*}
Therefore $\D'$ is LP-optimal. It satisfies property (b) because it adds at most one to the support of $\D$.
\end{proof}

\section{Algorithm \kMAB}
\label{sec:Balance}

This section presents and analyzes \kMAB, one of the two main algorithms. The design principle behind \kMAB is to explore as much as possible while avoiding obviously suboptimal strategies.  On a high level, the algorithm is very simple. The goal is to converge on an LP-perfect distribution. The time is divided into phases of $|X|$ rounds each. In the beginning of each phase $p$, the algorithm prunes away all distributions $\D$ over arms that with high confidence are not LP-perfect given the observations so far. The remaining distributions over arms are called \emph{potentially perfect}. Throughout the phase, the algorithm chooses among the potentially perfect distributions. Specifically, for each arm $x$, the algorithm chooses a potentially perfect distribution $\Dpx$ which approximately maximizes $\Dpx(x)$, and ``pulls" an arm sampled independently from this distribution. This choice of $\Dpx$ is crucial; we call it the \emph{balancing step}. The algorithm halts as soon as the time horizon is met, or any of the constraints is exhausted. The pseudocode is given in Algorithm~\ref{alg:Balance}.

\newcommand{\TAB}{\hspace{4mm}}
\floatname{algorithm}{Algorithm}
\begin{algorithm}[h]
\caption{\kMAB}
\label{alg:Balance}
\begin{algorithmic}[1]
\STATE {\bf For} each phase $p=0,1,2,\, \ldots$ {\bf do}
\STATE \TAB Recompute the set $\F_p$ of potentially perfect distributions $\D$ over arms.
\STATE \TAB Over the next $|X|$ rounds, for each $x\in X$:
\STATE \TAB\TAB pick any distribution $\D=\Dpx\in \F_p$ such that
                    $\D(x) \geq \tfrac12\, \max_{\D'\in \F_p} \D'(x)$.%
\STATE \TAB\TAB choose an arm to ``pull" as an independent sample from $\D$.
\STATE \TAB\TAB {\bf halt} if time horizon is met or one of the resources is exhausted.
\end{algorithmic}
\end{algorithm}

We believe that \kMAB, like $\ucbone$ \citep{bandits-ucb1}, is a very general design principle and has the potential to be a meta-algorithm for solving stochastic online learning problems.

\begin{theorem}\label{thm:Balance}
Consider an instance of \BwK with $d$ resources, $m=|X|$ arms, and the smallest budget $B = \min_i B_i$. Algorithm \kMAB achieves regret
\begin{align}\label{eq:Balance-regret}
\LPOPT - \Rew \leq O(\log(\horizon) \log(\horizon / \narms))
\left(
    \sqrt{\nrsc \narms \LPOPT \,} + \LPOPT
        \sqrt{\frac{\nrsc \narms}{\budg}} \;
\right).
\end{align}
Moreover, \eqref{eq:regret-LPOPT-OPT} holds with $f(\LPOPT)$ equal to the right-hand side of \eqref{eq:Balance-regret}.
\end{theorem}

\begin{remark}\label{rem:Balance-implementation}
The specification of \kMAB involves a mathematically well-defined step --- approximate optimization over potentially perfect distributions --- for which we do not provide a specific implementation. Yet, \kMAB is a bandit algorithm in the sense that it is a well-defined mapping from histories to actions. We prove an ``information-theoretic" statement: there is an algorithm with the claimed regret. Such results are not uncommon in the literature, \eg \citep{LipschitzMAB-stoc08,DichotomyMAB-soda10,monster-icml14}, typically as first solutions for new, broad problem formulations, and are meaningful as proof-of-concept for the corresponding regret bounds and techniques.

\OMIT{The said optimization step involves optimization over all latent distributions that are feasible for a given ``BwK domain", such as dynamic pricing with limited supply. It is not even clear what model of computation (\eg which oracle) would be appropriate to characterize access to the domain knowledge. Efficient implementation of the said optimization step may be possible for some specific BwK domains (but we do not pursue this direction).}
\end{remark}

%


\xhdr{Remaining details of the specification.}
In the beginning of each phase $p$, the algorithm recomputes a ``confidence interval" $I_p$ for the latent structure $\mu$, so that (informally) $\mu\in I_p$ with high probability. Then the algorithm determines which distributions $\D$ over arms can potentially be LP-perfect given that $\mu\in I_p$. Specifically, let $\F_p$ be set of all distributions $\D$ that are LP-perfect for some latent structure $\mu'\in I_p$; such distributions are called \emph{potentially perfect} (for phase $p$).

\OMIT{
\begin{align}\label{eq:phase-F}
 \F_p
    = \bigcup_{\mu'\in I_p} \{ \text{distributions } \}
    \argmax_{\D\in \allD} \LP(\D,\mu').
\end{align}
}

It remains to define the confidence intervals $I_p$. For phase $p=0$, the confidence interval $I_0$ is simply $\domain$, the set of all feasible latent structures.
For each subsequent phase $p\geq 1$, the confidence interval $I_p$ is defined as follows. For each arm $x$, consider all rounds before phase $p$ in which this arm has been chosen. Let $N_p(x)$ be the number of such rounds, let $\empir{r}_p(x)$ be the time-averaged reward in these rounds, and let $\empir{c}_{p,i}(x)$ be the time-averaged consumption of resource $i$ in these rounds. We use these averages to estimate $r(x,\mu)$ and $c_i(x,\mu)$ as follows:
\begin{align}
|r(x,\mu) - \empir{r}_p(x)|
    &\leq \rad\left(\, \empir{r}_p(x), N_p(x) \,\right)
    \label{eq:algo-estimate-r}\\
|c_i(x,\mu) - \empir{c}_{p,i}(x)|
    &\leq \rad\left(\, \empir{c}_{p,i}(x), N_p(x) \,\right)
    \quad \text{for each resource $i$}
    \label{eq:algo-estimate-c}
\end{align}
The confidence interval $I_p$ is the set of all latent structures $\mu'\in I_{p-1}$ that are consistent with these estimates. This completes the specification of \kMAB.

For each phase of \kMAB, the round in which an arm is sampled from distribution $\Dpx$ will be called \emph{designated} to arm $x$. We need to use approximate maximization to choose $\Dpx$, rather than exact maximization, because an exact maximizer
    $\argmax_{\D\in \F_p} \D(x)$
is not guaranteed to exist.

\xhdr{Proof overview.}
\asedit{We start with some properties of the algorithm that follow immediately from the specification and hold deterministically (with probability 1). Then we identify several properties that the algorithm satisfies with very high probability. The rest of the analysis focuses on a ``clean execution" of the algorithm: an execution in which all these properties hold. We analyze the ``error terms" that arise due to the uncertainty on the latent structure, and use the resulting ``error bounds" to argue about the algorithm's performance.}

\subsection{Deterministic properties of \kMAB}

First, we show that any two latent structures in the confidence interval $I_p$ correspond to similar consumptions and rewards, for each arm $x$. This follows deterministically from the specification of $I_p$.

\begin{claim}
\label{cl:bdd-error-x}
Fix any phase $p$, any two latent structures $\mu',\mu''\in I_p$, an arm $x$, and a resource $i$. Then
\begin{align}
\ | c_i(x, \mu') - c_i(x, \mu'')|&\leq 6\; \rad\left( c_i(x,\mu'), N_p(x) \right) \label{eq:cl:bdd-error-x-c}\\
 | r(x, \mu') - r(x, \mu'')|&\leq 6\; \rad\left( r(x,\mu'), N_p(x) \right).\label{eq:cl:bdd-error-x-r}
\end{align}
\end{claim}
\begin{proof}
We prove \eqref{eq:cl:bdd-error-x-c}; \eqref{eq:cl:bdd-error-x-r} is proved similarly.

Let $N = N_p(x)$.
By specification of \kMAB, any $\mu'\in I_p$ is consistent with estimate~\refeq{eq:algo-estimate-c}:
$$|c_i(x,\mu') - \empir{c}_{p,i}(x)|
    \leq \rad\left(\, \empir{c}_{p,i}(x), N \,\right).
$$
It follows that
$$ | c_i(x, \mu') - c_i(x, \mu'')| \leq 2\; \rad\left( \empir{c}_{p,i}(x), N \right).
$$
Finally, we observe that by Claim~\ref{cl:rad-to-rad},
$$\rad\left( \empir{c}_{p,i}(x), N \right) \leq 3\,\rad\left( c_i(x,\mu'), N \right).
\qedhere
$$
\end{proof}

For each phase $p$ and arm $x$, let
    $\avedtx{p}{x}=\frac{1}{p}\sum_{q< p}\Dqx(x)$
be the average of probabilities for arm $x$ among the distributions in the preceding phases that are designated to arm $x$. Because of the balancing step in \kMAB, we can compare this quantity to $\D(x)$, for any $\D\in \F_p$. (Here we also use the fact that the confidence intervals $I_p$ are non-increasing from one phase to another.)

\begin{claim}
\label{cl:Dpx}
$\avedtx{p}{x}\geq \tfrac12\,\D(x)$
for each phase $p$, each arm $x$ and any distribution $\D\in \F_p$.
\end{claim}
\begin{proof}
Fix arm $x$. Recall that
    $\avedtx{p}{x}=\frac{1}{p}\sum_{q< p}\Dqx(x)$,
where $\Dqx$ is the distribution chosen in the round in phase $q$ that is designated to arm $x$. Fix any phase $q<p$. Because of the balancing step, $\Dqx(x) \geq \tfrac12\, \D'(x)$ for any distribution $\D'\in \F_q$. Since the confidence intervals $I_q$ are non-increasing from one phase to another, we have $I_p\subset I_q$ for any $q\leq p$, which implies that $\F_p \subset \F_q$. Consequently, $\Dqx(x) \geq \tfrac12\, \D(x)$ for each $q< p$, and the claim follows.
\end{proof}

\subsection{High-probability events}

We keep track of several quantities: the averages $\empir{r}_p(x)$ and $\empir{c}_{p,i}(x)$ defined above, as well as several other quantities that we define below.

Fix phase $p$ and arm $x$. Recall that $N_p(x)$ is the number of rounds before phase $p$ in which arm $x$ is chosen. Now, let us consider all rounds before phase $p$ that are \emph{designated} to arm $x$. Let $n_p(x)$ denote the number of times arm $x$ has been chosen in these rounds. Let $\empirdtx{p}{x}=n_t(x)/p$ be the corresponding empirical probability of choosing $x$. We compare this to $\avedtx{p}{x}$.

\OMIT{Let
    $\avedtx{p}{x}=\frac{1}{p}\sum_{q< p}\Dqx(x)$
be the average of probabilities for arm $x$ among the distributions chosen for these rounds.}

Further, consider all rounds in phases $q< p$. There are $N = p|X|$ such rounds. The average distribution chosen by the algorithm in these rounds is
    $\aveD_p = \tfrac{1}{N} \sum_{q< p,\, x\in X} \Dqx $.
We are interested in the corresponding quantities
    $r(\aveD_p, \mu)$ and $c_i(\aveD_p, \mu)$,
We compare these quantities to
    $\empir{r}_p = \tfrac{1}{N} \sum_{t=1}^N r_t$
and
    $\empir{c}_{p,i} = \tfrac{1}{N} \sum_{t=1}^N c_{t,i}$,
the average reward and the average resource-$i$ consumption in phases $q<p$.

We consider several high-probability events which follow from applying Theorem~\ref{thm:conf-rad} to the various quantities defined above. All these events have a common shape: some quantities $\nu,\empir{\nu}$ satisfy \eqref{eq:nu-conf-int} for some $N$. If this is the case, we that $\empir{\nu}$ is an \emph{$N$-strong estimator} for $\nu$.

\begin{lemma}\label{lm:high-prob-events}
For each phase $p$, arm $x$, and resource $i$, with probability $e^{-\Omega(\chernoffC)}$ it holds that:
\begin{itemize}
\item[(a)] $\empir{r}_p(x)$ is an $N_p(x)$-strong estimator for $r(x,\mu)$, and
$\empir{c}_{p,i}(x)$ is an $N_p(x)$-strong estimator for $c_i(x,\mu)$.

\item[(b)] $\avedtx{p}{x}$ is an $p$-strong estimator for $\empirdtx{p}{x}$.

\item[(c)] $r(\aveD_p,\, \mu)$  is an $(p|X|)$-strong estimator for $\empir{r}_p$,
and $c_i(\aveD_p,\, \mu)$ is an $(p|X|)$-strong estimator for
 $\empir{c}_{p,i}$.
\end{itemize}
\end{lemma}

We rely on several properties of the confidence radius $\rad()$, which we summarize below. (We omit the easy proofs.)

\begin{claim}
\label{cl:conf-rad-props}
The confidence radius $\rad(\nu,N)$, defined in \eqref{eq:conf-rad-defn}, satisfies the following properties:
\begin{OneLiners}
\item[(a)] monotonicity: $\rad(\nu,N)$ is non-decreasing in $\nu$ and non-increasing in $N$.
\item[(b)] concavity: $\rad(\nu,N)$ is concave in $\nu$, for any fixed $N$.
\item[(c)] $\max(0,\; \nu - \rad(\nu,N))$ is non-decreasing in $\nu$.
\item[(d)] $\nu- \rad(\nu,N) \geq \tfrac14\, \nu$ whenever $4\tfrac{\chernoffC}{N} \leq \nu \leq 1$.
\item[(e)] $\rad(\nu,N) \leq 3 \tfrac{\chernoffC}{N}$ whenever $\nu\leq 4\tfrac{\chernoffC}{N}$.
\item[(f)] $\rad(\nu,\alpha N) = \tfrac{1}{\alpha}\, \rad(\alpha\nu,\, N)$,
    for any $\alpha\in (0,1]$.
\item[(g)]
    $\frac{1}{N}\sum_{\ell=1}^{N} \rad(\nu, \ell)
        \leq O(\log N)\;\rad(\nu,N)$.
\end{OneLiners}
\end{claim}

\subsection{Clean execution analysis}

It is convenient to focus on a \emph{clean execution} of the algorithm: an execution in which all events in Lemma~\ref{lm:high-prob-events} hold. We assume a clean execution in what follows. Also, we fix an arbitrary phase $p$ in such execution.

\asedit{Clean execution analysis falls into two parts. First, we analyze the ``error terms": we look at the LP-value (resp., expected reward, or expected resource consumption) of a given distribution, and upper-bound the difference in this quantity between different latent structures $\mu,\mu'$ in the confidence interval $I_p$, or between different potentially perfect distributions $D',D''\in\F_p$. The culmination is Lemma~\ref{lm:balanceBwK-errors}, which upper-bounds the difference
    $|\LP(\D',\mu') - \LP(\D'',\mu'')|$
in terms of parameters $\tfrac{p}{d}$, $B$, $T$, and $\LPOPT$. Second, we apply these error bounds to reason about the algorithm itself. The key quantities of interest are LP-values of the chosen distributions, average reward/consumption, and the stopping time.}

\subsubsection{Bounding the error terms}

Since a clean execution satisfies the event in Claim~\ref{lm:high-prob-events}(a), it immediately follows that:

\begin{claim}
The confidence interval $I_p$ contains the (actual) latent structure $\mu$. Therefore, $\D^*\in \F_p$ for any distribution $\D^*$ that is LP-perfect for $\mu$.
\end{claim}

\begin{claim}
\label{claim:bounderror-Alex}
Fix any latent structures $\mu',\mu''\in I_p$ and any distribution $\D\in \F_p$. Then for each resource $i$,
\begin{align}
 | c_i(\D, \mu') - c_i(\D, \mu'')|&\leq O(1)\; \rad\left( c_i(\D,\mu'),\, p/d \right) \label{eq:claim:bounderror-Alex-c}\\
 | r(\D, \mu') - r(\D, \mu'')|&\leq O(1)\; \rad\left( r(\D,\mu'),\, p/d \right).\label{eq:claim:bounderror-Alex-r}
\end{align}
\end{claim}

\begin{proof}
We prove \eqref{eq:claim:bounderror-Alex-c}; \eqref{eq:claim:bounderror-Alex-r} is proved similarly.
Let us first prove the following:
\begin{align}\label{eq:inter:errorbound-Alex}
\forall x\in X,\quad
\D(x)\, | c_i(x,\mu') - c_i(x,\mu'') |\leq
    O(1)\; \rad(\D(x)\, c_i(x,\mu'),p).
\end{align}
Intuitively, in order to argue that we have good estimates on quantities related to arm $x$, it helps to prove that this arm has been chosen sufficiently often. Using the definition of clean execution and Claim~\ref{cl:Dpx}, we accomplish this as follows:
\begin{align*}
\tfrac{1}{p}\, N_p(x)
    &\geq \tfrac{1}{p}\, n_p(x)
    = \empirdtx{p}{x} \\
    &\geq  \avedtx{p}{x} - \rad(\avedtx{p}{x}, p) & \text{(by clean execution)} \\
    &\geq  \tfrac12\,\D(x) - \rad(\tfrac12\,\D(x), p)
        & \text{(by Claim~\ref{cl:Dpx} and Claim~\ref{cl:conf-rad-props}(c))}.
\end{align*}
Consider two cases depending on $\D(x)$. For the first case, assume
    $\D(x) \geq 8 \tfrac{\chernoffC}{p} $.
Using Claim~\ref{cl:conf-rad-props}(d) and the previous equation, it follows that
    $N_p(x) \geq \tfrac{1}{8}\, p\,\D(x)$.
Therefore:
\begin{align*}
\D(x)\; |c_i(x,\mu')-c_i(x,\mu'')|
&\leq 6\;\D(x)\, \rad(c_i(x,\mu'),\,N_p(x) )
    & \text{(by Claim~\ref{cl:bdd-error-x})}\\
&\leq 6\;\D(x)\, \rad(c_i(x,\mu'),\,\tfrac{1}{8}\, p\,\D(x) )
    & \text{(by monotonicity of $\rad$)}\\
&= 48\; \rad(\D(x)\, c_i(x,\mu_x'), p)
    & \text{(by Claim~\ref{cl:conf-rad-props}(f))}.
\end{align*}
The second case is that
    $\D(x) < 8 \tfrac{\chernoffC}{p} $.
Then \eqref{eq:inter:errorbound-Alex} follows simply because
    $\tfrac{\chernoffC}{p} \leq \rad(\cdot\,, p)$.

We have proved \eqref{eq:inter:errorbound-Alex}. We complete the proof of \eqref{eq:claim:bounderror-Alex-c} using concavity of $\rad(\cdot, p)$ and the fact that, by the specification of \kMAB, $\D$ has support of size at most $d$.
\begin{align*}
| c_i(\D, \mu') - c_i(\D, \mu'')|
    & \leq \textstyle \sum_{x\in X} \D(x)\, | c_i(x,\mu') - c_i(x,\mu'') | \\
    &\leq  \textstyle \sum_{x\in X,\D(x)>0} O(1)\; \rad(\D(x)\, c_i(x,\mu'),p) \\
    &\leq O(d)\; \rad\left( \tfrac{1}{d} \,\textstyle \sum_{x\in X} \D(x)\, c_i(x,\mu'),\; p \right)\\
    &= O(d)\; \rad\left( \tfrac{1}{d} \; c_i(\D,\mu'),\; p \right) \\
    &\leq O(1)\; \rad\left( c_i(\D,\mu'),\; \tfrac{p}{d} \right)
        \qquad \text{(by Claim~\ref{cl:conf-rad-props}(f))}. \qedhere
\end{align*}
\end{proof}

In what follows, we will denote
    $\ucblp_p    = \max_{\D\in \F_p,\;\mu\in I_p}\LP(\D,\mu)$.

\begin{claim}
\label{cl:lpconvergence-Alex}
Fix any latent structures $\mu',\mu''\in I_p$ and any distribution $\D\in \F_p$. Then
\begin{align}\label{eq:cl:lpconvergence-Alex}
|\LP(\D,\mu')-\LP(\D,\mu'')|
    &\leq O(T)\, \rad\left( \ucblp_p/T,\; \tfrac{p}{d} \right)
        + O(\ucblp_p\, \tfrac{T}{B})\,\rad\left(\tfbt,\tfrac{p}{d}\right).
\end{align}
\end{claim}

\begin{proof}
Since $\D\in \F_p$, it is LP-perfect for some latent structure $\mu$. Then
    $\LP(\D,\mu) = T\, r(\D,\mu)$.
Therefore:
\begin{align}
\LP(\D,\mu')-\LP(\D,\mu)
&\leq\;T\,\left( r(\D,\mu')- r(\D,\mu) \right) & \nonumber \\
&\leq\;O(T)\;\rad\left( r(\D,\mu), \tfrac{p}{d} \right)
    & \text{(by Claim~\ref{claim:bounderror-Alex})}.
        \label{eq:cl:lpconvergence-Alex-easy}
\end{align}

We need a little more work to bound the difference in the $\LP$ values in the other direction.

Consider
    $t_0 = \LP(\D,\mu')/  r(\D,\mu')$;
this is the value of the variable t in the optimal solution to the linear program~\refeq{eq:LP-D}. Let us obtain a lower bound on this quantity. Assume $t_0<T$. Then one of the budget constraints in~\refeq{eq:LP-D} must be tight, i.e. $t_0\, c_i(\D,\mu') = B$ for some resource $i$.
\begin{align*}
c_i(\D,\mu')
    &\leq c_i(\D,\mu)+O(1)\,\rad\left(c_i(\D,\mu),\tfrac{p}{d} \right)
        & \text{(by Claim~\ref{claim:bounderror-Alex})} \\
    &\leq \tfbt+O(1)\,\rad\left(\tfbt,\tfrac{p}{d}\right)
\end{align*}
Let $\Psi = \rad\left(\tfbt,\tfrac{p}{d}\right) $.
It follows that
    $t_0 = B/c_i(\D,\mu')  \geq T(1- O(\tfrac{T}{B}\,\Psi)).$
Therefore:
\begin{align*}
\LP(\D,\mu)-\LP(\D,\mu')
    &=\;T\, r(\D,\mu)-t_0\,r(\D,\mu') \\
    &\leq T\,r(\D,\mu)- \left[ T(1- O(\tfrac{T}{B}\,\Psi)) \right] \;  r(\D,\mu') \\
    &\leq T\, \left[ r(\D,\mu)- r(\D,\mu') \right]
        + O(\tfrac{T}{B}\,\Psi) \,T\, r(\D,\mu)  \\
    &\leq O(T)\, \rad\left( r(\D,\mu), \tfrac{p}{d} \right)
        + O(\tfrac{T}{B}\,\Psi)\,T\, r(\D,\mu)
    & \text{(by Claim~\ref{claim:bounderror-Alex})}.
\end{align*}
Using \eqref{eq:cl:lpconvergence-Alex-easy} and noting that
    $r(\D,\mu) = \LP(\D,\mu) / T \leq \ucblp_p/T$,
we conclude that
\begin{align*}
|\LP(\D,\mu)-\LP(\D,\mu')|
    &\leq O(T)\, \rad\left( \ucblp_p/T,\; \tfrac{p}{d} \right)
        + O(\ucblp_p\, \tfrac{T}{B})\,\rad\left(\tfbt,\tfrac{p}{d}\right).
\end{align*}
We obtain the same upper bound on
    $|\LP(\D,\mu)-\LP(\D,\mu'')|$,
and the claim follows.
\end{proof}

We will use $\Phi_p(\ucblp_p)$ to denote the right-hand side of \eqref{eq:cl:lpconvergence-Alex} as a function of $\ucblp_p$.

\begin{claim}
\label{cl:LP-of-Dt}
~
\begin{itemize}
\item[(a)]
Fix any latent structure $\mu^*\in I_p$, and any distributions $\D',\D''\in \F_p$.
Then
$$|\LP(\D',\mu^*) - \LP(\D'',\mu^*)|\leq 2 \Phi_p(\ucblp_p)$$
\item[(b)]
Fix any latent structure $\mu',\mu''\in I_p$, and any distributions $\D',\D''\in \F_p$.
Then
$$|\LP(\D',\mu') - \LP(\D'',\mu'')|\leq 3 \Phi_p(\ucblp_p)$$
\end{itemize}
\end{claim}

\begin{proof}
{\bf (a).} Since $\D',\D''\in \F_p$, it holds that $\D'$ and $\D''$ are LP-perfect for some latent structures $\mu'$ and $\mu''$. Further, pick a distribution $\D^*$ that is LP-perfect for $\mu^*$. Then:
\begin{align*}
\LP(\D',\mu^*)
    &\geq \LP(\D',\mu')- \Phi_p(\ucblp_p)
    &\textrm{(by Lemma~\ref{cl:lpconvergence-Alex} with $\D = \D'$)} &\\
    &\geq\LP(\D^*,\mu')- \Phi_p(\ucblp_p) \\
    &\geq \LP(\D^*,\mu^*)- 2\Phi_p(\ucblp_p)
    &\textrm{(by Lemma~\ref{cl:lpconvergence-Alex} with $\D =\D^*$)} &\\
    &\geq \LP(\D'',\mu^*)- 2\Phi_p(\ucblp_p). &&
\end{align*}

{\bf (b).} Follows easily from part (a) and Lemma~\ref{cl:lpconvergence-Alex}. \qedhere
\end{proof}

The following claim will allow us to replace $\Phi_p(\ucblp_p)$ by $\Phi_p(\LPOPT)$.%

\begin{claim}\label{cl:LP-phip-substitute}
$\Phi_p(\LPOPT)\geq \Omega(\asedit{\min}(\LPOPT,\Phi_p(\ucblp_p)))$.
\end{claim}
\begin{proof}
Consider the two summands in $\Phi_p(\ucblp_{p})$:
\begin{align*}
S_1(\ucblp_{p}) &= O(T)\, \rad\left( \ucblp_p/T,\; \tfrac{p}{d} \right), \\
S_2(\ucblp_{p}) &= O(\ucblp_p\, \tfrac{T}{B})\,\rad\left(\tfbt,\tfrac{p}{d}\right).
\end{align*}

We consider the following three cases. The first case is that
    $S_1(\ucblp_{p}) \geq \ucblp_p/12$.
Solving for $\ucblp_p$, we obtain $\ucblp_p\leq O(\frac{Td\chernoffC}{p})$, which implies that
$$ \Phi_p(\LPOPT)\geq \Omega(\ucblp_p)\geq \Omega(\LPOPT).
$$
The second case is that
    $S_2(\ucblp_{p}) \geq \ucblp_p/12$.
Then
    $$ \Phi_p(\LPOPT) \geq S_2(\LPOPT) \geq \LPOPT/12$$
In remaining case,
    $\Phi_p(\ucblp_p)\leq \frac{\ucblp_p}{6}$.
Then from Claim~\ref{cl:LP-of-Dt}(b) we get that
    $\ucblp_p\leq\, 2\,\LPOPT$.
Noting that $\Phi_p(M)$ is a non-decreasing function of $M$, we obtain
$$\Phi_p(\ucblp_p)\leq \Phi_p(2\,\LPOPT) \leq 2\,\Phi_p(\LPOPT). \qquad \qedhere
$$
\end{proof}

Claim~\ref{cl:LP-phip-substitute} and Claim~\ref{cl:LP-of-Dt} imply our main bound on the error terms:

\begin{lemma}
\label{lm:balanceBwK-errors}
Fix any latent structure $\mu',\mu''\in I_p$, and any distributions $\D',\D''\in \F_p$.
Then
$$|\LP(\D',\mu') - \LP(\D'',\mu'')|\leq O(\Phi_p(\LPOPT)).$$
\end{lemma}

\subsubsection{Performance of the algorithm}

The remainder of the analysis deals with rewards and resource consumption of the algorithm. We start with lower-bounding the LP-value for the chosen distributions.

\begin{claim}\label{cl:LP-phip-substitute-b}
For each distribution $\Dpx$ chosen by the algorithm in phase $p$,
$$\LP(\Dpx,\mu) \geq \LPOPT - O(\Phi_p(\LPOPT)).$$
\end{claim}

\begin{proof}
The claim follows easily from Lemma~\ref{lm:balanceBwK-errors}, noting that $\Dpx\in \F_p$.
\end{proof}

The following corollary lower-bounds the average reward; once we have it, it essentially remains to lower-bound the stopping time of the algorithm.

\begin{corollary}\label{cor:BasicBwK-Rew}
$\empir{r}_p \geq \tfrac{1}{T}\,(\LPOPT - O(\log p)\,\Phi_p(\LPOPT))$.
\end{corollary}

\begin{proof}
Throughout this proof, denote
    $\Phi_p\triangleq \Phi_p(\LPOPT)$.
By Claim~\ref{cl:LP-phip-substitute-b}, for each distribution $\Dqx$ chosen by the algorithm in phase $q<p$ it holds that
$$ r(\Dqx,\mu) \geq \tfrac{1}{T}\, \LP(\Dqx,\mu)
\geq \tfrac{1}{T}\, (\LPOPT - O(\Phi_q)).$$
Averaging the above equation over all rounds in phases $q<p$, we obtain
\begin{align*}
r(\aveD_p,\mu)
    &\geq \tfrac{1}{T}\, \left(
\LPOPT - \textstyle \tfrac{1}{p} \sum_{q<p} O(\Phi_q)
    \right) \\
    &\geq \tfrac{1}{T}\, \left(
        \LPOPT - O(\Phi_p \log p)
    \right).
\end{align*}
For the last inequality, we used Claim~\ref{cl:conf-rad-props}(fg) to average the confidence radii in $\Phi_q$.

Using the high-probability event in Claim~\ref{lm:high-prob-events}(c):
$$ \empir{r}_p
    \geq r(\aveD_p,\mu) - \rad( r(\aveD_p,\mu), p|X| ).
$$
Now using the monotonicity of $\nu - \rad(\nu,N)$ (Claim~\ref{cl:conf-rad-props}(c)) we obtain
\begin{align*}
\empir{r}_p
    &\geq \tfrac{1}{T}\,(\LPOPT - O(\Phi_p))
    - \rad\left(\;  \tfrac{1}{T}\,(\LPOPT - O(\Phi_p)),\; p|X| \;\right)\\
    &\geq \tfrac{1}{T}\,(\LPOPT - O(\Phi_p))
    - \rad\left(\;  \LPOPT/T,\; p|X| \;\right)\\
    &\geq \tfrac{1}{T}\,(\LPOPT - O(\Phi_p)).
\end{align*}
For the last equation, we use the fact that
$\Phi_p/T
    \geq \Omega(\rad(\LPOPT/T, \tfrac{p}{d} ))
    \geq \Omega(\rad(\;  \LPOPT/T,\; p|X| \;)).$
\end{proof}

The following two claims help us to lower-bound the stopping time of the algorithm.

\begin{claim}
\label{cl:BasicBwK-supply}
$ c_i(\Dpx,\mu) \leq \tfbt + O(1)\,\rad\left(\tfbt,\; \tfrac{p}{d} \right)$
for each resource $i$.
\end{claim}

\begin{proof}
By the algorithm's specification, $\Dpx\in \F_p$, and moreover there exists a latent structure $\mu'\in I_p$ such that $\Dpx$ is LP-perfect for $\mu'$. Apply Claim~\ref{claim:bounderror-Alex}, noting that $c_i(\Dpx,\mu') \leq \fbt$ by LP-perfectness.
\end{proof}

\begin{corollary} \label{cor:BasicBwK-supply}
    $\empir{c}_{p,i} \leq \tfbt + O(\log p)\; \rad(\tfbt, \tfrac{p}{d}) $
for each resource $i$.
\end{corollary}
\begin{proof}
Using a property of the clean execution, namely the event in Claim~\ref{lm:high-prob-events}(c), we have
\begin{align}\label{eq:cor:BasicBwK-supply-empir}
\empir{c}_{p,i}
    \leq c_i(\aveD, \mu) + \rad\left( c_i(\aveD, \mu),\; p\right).
\end{align}

Consider all rounds preceding phase $p$.
\begin{align}
 c_i(\aveD_p, \mu)
    &=  \frac{1}{p|X|}  \sum_{q<p,\, x\in X} c_i(\Dqx,\mu) & \nonumber\\
    &\leq \fbt + \frac{O(1)}{p|X|} \sum_{q<p,\, x\in X} \rad\left(\fbt, \frac{p}{d}\right)
        &\text{(by Claim~\ref{cl:BasicBwK-supply})} \nonumber\\
    &\leq \tfbt+O(\log p)\, \rad\left(\tfbt, \tfrac{p}{d}\right)
        &\text{(by Claim~\ref{cl:conf-rad-props}(fg))}.
        \label{eq:cor:BasicBwK-supply-ave}
\end{align}
For the last inequality, we used Claim~\ref{cl:conf-rad-props}(fg) to average the confidence radii.

Using the upper bound on $c_i(\aveD, \mu)$ that we derived above,
\begin{align*}
\rad\left( c_i(\aveD, \mu),\; \tfrac{p}{d}\right)
    \leq O(\log p)\; \rad\left( \tfbt+ \rad\left(\tfbt, \tfrac{p}{d}\right),\; \tfrac{p}{d} \right).
\end{align*}
Using a general property of the confidence radius that
$$ \rad(\nu+\rad(\nu,N),\,N) \leq O(\rad(\nu,N)),
$$
we conclude that
\begin{align}\label{eq:cor:BasicBwK-supply-rad}
 \rad\left( c_i(\aveD, \mu),\; \tfrac{p}{d}\right)
    \leq O(\log p)\; \rad(\tfbt, \tfrac{p}{d}).
\end{align}
We obtain the claim by plugging the upper bounds \refeq{eq:cor:BasicBwK-supply-ave} and \refeq{eq:cor:BasicBwK-supply-rad} into \refeq{eq:cor:BasicBwK-supply-empir}.
\end{proof}

We are ready to put the pieces together and derive the performance guarantee for a clean execution of \kMAB.

\begin{lemma}\label{lm:pf-Balance:almost-done}
Consider a clean execution of \kMAB. Then the total reward
\[ \Rew \geq \LPOPT - O\left( \log \tfrac{T}{|X|}\right)\; \Phi_{T/|X|}(\LPOPT).\]
\end{lemma}
\begin{proof}
Throughout this proof, denote
    $\Phi_p\triangleq \Phi_p(\LPOPT)$.
Let $p$ be the last phase in the execution of the algorithm, and let $T_0$ be the stopping time. Letting $m=|X|$, note that
    $pm <T_0 \leq (p+1)m$.

We can use Corollary \ref{cor:BasicBwK-Rew} to bound $\Rew$ from below:
\begin{align}\label{eq:pf-Balance-Rew}
\Rew
= T_0\,\empir{r}_{p+1}  > p\,m \,\empir{r}_{p+1} \geq
    \tfrac{p\,m}{T}\,(\LPOPT - O(\Phi_p \log p)).
\end{align}

Let us bound $\frac{p\,m}{T}$ from below.  The algorithm stops either when it runs out of time or if it runs out of resources during phase $p$. In the former case, $p=\flr{T/m}$. In the latter case,
    $B= T_0\,\empir{c}_{p+1,\,i} $
for some resource $i$, so
    $B\leq m(p+1)\;\empir{c}_{p+1,\,i}$.
Using Corollary \ref{cor:BasicBwK-supply}, we obtain the following lower bound on $p$:
$$ \frac{p\, m}{T}
    \geq 1 - O\left( \frac{p\, m\log p}{B} \right)\;
        \rad\left( \fbt,\, \frac{p}{d} \right).
$$
Plugging this into \eqref{eq:pf-Balance-Rew}, we conclude:
\begin{align*}
\Rew &\geq \LPOPT
    - O\left( \frac{p\, m\log p}{B} \right)\;
            \rad\left( \fbt,\, \frac{p}{d} \right) \LPOPT
    - O\left( \frac{p\,m\log p}{T} \right) \Phi_p \\
    &\geq \LPOPT - O\left( \frac{p\,m\log p}{T} \right)
        \left( \Phi_p+ \frac{T}{B}\, \rad\left( \fbt,\, \frac{p}{d} \right) \LPOPT\right)\\
  &\geq \LPOPT -\frac{p\,m\log p}{T}\; O(\Phi_p).
\end{align*}
To complete the proof, we observe that $(p\,\Phi_p \log p)$ is increasing in $p$ (by definition of $\Phi_p$), and plug in a trivial upper bound $p\leq T/m$.
\end{proof}

\asedit{To finish the proof of Theorem~\ref{thm:Balance}, we write down the definition of $\Phi_{T/m}(\LPOPT)$, $m=|X|$, and plug in the definition of the confidence radius \refeq{eq:conf-rad-defn}:
\begin{align*}
\Phi_{T/m}(\LPOPT)
    &\triangleq O(T)\, \rad\left( \LPOPT/T,\; \tfrac{T}{d|X|} \right)
        + O(\LPOPT\, \tfrac{T}{B})\,\rad\left(\tfbt,\tfrac{T}{dm}\right) \\
    &\leq O(\log(\horizon)
\left(
    \sqrt{\nrsc \narms \LPOPT \,} + \LPOPT
        \sqrt{\frac{\nrsc \narms}{\budg}} \;
\right).
\end{align*}
} 

\OMIT{ 
\ascomment{I left some Ashwin's stuff below just in case}

\ascomment{Ashwin's theorem statement}

\begin{theorem}\label{lm:basicBwK-totalrevenue}
The algorithm has total revenue
\begin{align}\label{eq-lm:basicBwK-totalrevenue}
\Rew \geq \LPOPT - O(d\, \chernoffC) \left(\frac{\maxLP}{B}\,\left(\sqrt{B |X|}+|X|\right)+\sqrt{\maxLP |X|}+|X|\right)
\end{align}
\end{theorem}

\ascomment{the last computation in Ashwin's proof}

\begin{align}\label{eq-lm:basicBwK-interface-1}
\Rew
&\geq \frac{p\, n}{T}\; \LPOPT - O(d\, \chernoffC\, n) \left(\frac{\maxLP}{B}\,\left(\sqrt{\frac{B\,p}{T}}+1\right)+\sqrt{\frac{\maxLP\,p}{T}}+1\right) \\
    &\geq \LPOPT - O(d\, \chernoffC\, n) \left(\frac{\maxLP}{B}\,\left(\sqrt{\frac{B\,p}{T}}+1\right)+\sqrt{\frac{\maxLP\,p}{T}}+1\right) \\
 &\geq \LPOPT - O(d\, \chernoffC) \left(\frac{\maxLP}{B}\,\left(\sqrt{B n}+n\right)+\sqrt{\maxLP n}+n\right)
\end{align}
} 

\section{Algorithm \pdbwk}
\label{sec:pdbwk}

%

This section develops an algorithm, called \pdbwk, that solves the \BwK problem using a very natural and intuitive idea: greedily select arms with the greatest estimated ``bang per buck,'' i.e.\ reward per unit of resource consumption. One of the main difficulties with this idea is that there is no such thing as a known ``unit of resource consumption'': there are $\nrsc$ different resources, and it is unclear how to trade off consumption of one resource versus another. The dual LP in Section~\ref{sec:lp} gives some insight into how to quantify this trade-off: an optimal dual solution $\eta^*$ can be interpreted as a vector of unit costs for resources, such that for every arm the expected reward is less than or equal to the expected cost of resources consumed. \asedit{Then the bang-per-buck ratio for a given arm $x$ can be defined as $r(x,\mu)/(\eta^* \cdot c(x,\mu))$, where the denominator represents the expected cost of pulling this arm. The arms in the support of the optimal distribution $\xi^*$ are precisely the arms with a maximal bang-per-buck ratio (by complimentary slackness), and pulling any other arm necessarily increases regret relative to $\LPOPT$ (by Remark~\ref{rem:LP-intuition}).}


To estimate the bang-per-buck ratios, our algorithm will try to learn an optimal dual vector $\eta^*$ in tandem with learning the latent structure $\mu$. Borrowing an idea from \citep{PlotkinST,GargK,AroraHK}, we use the multiplicative weights update method to learn the optimal dual vector. This method raises the cost of a resource exponentially as it is consumed, which ensures that heavily demanded resources become costly, and thereby promotes balanced resource consumption. Meanwhile, we still have to ensure (as with any multi-armed bandit problem) that our algorithm explores the different arms frequently enough to gain adequately accurate estimates of the latent structure. We do this by estimating rewards and resource consumption as optimistically as possible, i.e.\ using upper confidence bound (UCB) estimates for rewards and lower confidence bound (LCB) estimates for resource consumption. Although both of these techniques --- multiplicative weights and confidence bounds --- have been successfully applied in previous online learning algorithms, it is far from obvious that this particular hybrid of the two methods should be effective. In particular, the use of multiplicative updates on dual variables, rather than primal ones, distinguishes our algorithm from other bandit algorithms that use multiplicative weights (e.g.\ the $\expthree$ algorithm \citep{bandits-exp3}) and brings it closer in spirit to the literature on stochastic packing algorithms, especially \citep{DevanurJSW-ec11}.

The pseudocode is presented as Algorithm~\ref{alg:pdbwk}. When we refer to the UCB or LCB for a latent parameter
(the reward of an arm, or the amount of some resource that it utilizes),
these are computed as follows. Letting $\hat{\nu}$ denote the empirical
average of the observations of that random variable\footnote{Note
that we initialize the
algorithm by pulling each arm once, so empirical averages
are always well-defined.} and letting $N$ denote the number of times
the random variable has been observed, the lower confidence bound
(LCB) and upper confidence bound (UCB) are the left and right endpoints,
respectively, of the \emph{confidence interval} $[0,1] \cap
[\hat{\nu} - \rad(\hat{\nu},N), \; \hat{\nu} + \rad(\hat{\nu},N)]$.
The UCB or LCB for a vector or matrix are defined componentwise.

\newcommand{\EstCost}{\mathtt{EstCost}}

\begin{algorithm}[ht]
\caption{\pdbwk with parameter $\eps\in(0,1)$}
\label{alg:pdbwk}
\begin{algorithmic}[1]
\STATE {\bf Initialization}
\STATE \TAB In the first $\narms$ rounds, pull each arm once.
\STATE \TAB $v_1 = \ones \in [0,1]^\nrsc$.
\STATE \TAB\TAB
    \COMMENT{$v_t\in [0,1]^\nrsc$ is the round-$t$ estimate of the optimal solution $\eta^*$ to \refeq{lp:dual} in Section~\ref{sec:lp}.}
\STATE \TAB\TAB
    \COMMENT{We interpret $v_t(i)$ as an estimate of the (fictional) unit cost of  resource $i$, for each $i$.}
\STATE \TAB Set $\eps = \sqrt{\ln(\nrsc)/\budg}$.
\FOR{rounds $t = \narms+1,\ldots,\stime \;\; \mbox{\it (i.e., until
resource budget exhausted)}$}
\STATE For each arm $x\in X$,
\STATE \TAB Compute UCB estimate for the expected reward, $\rucb{t,x} \in [0,1]$.
\STATE \TAB Compute LCB estimate for the resource consumption vector,
                $\clcb{t,x} \in [0,1]^{\nrsc}$.
\STATE \TAB \emph{Expected cost} for one pull of arm $x$ is estimated by $\EstCost_x = \clcb{t,x}\cdot v_t$.
\STATE Pull arm $x=x_t\in X$ that maximizes  $\rucb{t,x}/\EstCost_x$,
            the optimistic \emph{bang-per-buck} ratio.
\label{algstep:argmin}
\STATE Update estimated unit cost for each resource $i$:
\begin{align*}
                v_{t+1}(i) = v_t(i)\, (1+\eps)^{\ell},\;  \ell = \clcb{t,x}(i).
\end{align*}
\label{algstep:mwu}
\ENDFOR
\end{algorithmic}
\end{algorithm}

The algorithm is fast: with machine word size of $\log T$ bits or more, the per-round running time is $O(\narms \nrsc)$. Moreover, if each arm $x$ consumes only $\nrsc_x$ resources that are known in advance, then $\clcb{t,x}$ can be implemented as a $\nrsc_x$-dimensional vector, and $\EstCost_x$ can be computed in $O(\nrsc_x)$ time. Then the per-round running time is
    $O(\narms + \nrsc+\sum_x \nrsc_x)$.

\begin{discussion}
\asedit{
The cost update in step~\ref{algstep:mwu} requires some explanation. Let us interpret this step as a separate algorithm which solves a particular problem. The problem is to optimize the total expected payoff when in each round $t>\narms$, one chooses a distribution $y_t = v_t/\|v_t\|_1$ over resources, and receives expected payoff $y_t\cdot \clcb{t,x_t}$. This is the well-known "best-expert" problem in which actions correspond to resources, and each action $i$ is assigned payoff $\clcb{t,x_t}(i)$.  Step~\ref{algstep:mwu} implements a multiplicative-weights algorithm for solving this problem. In fact, we could have used any other algorithm for this problem with a similar performance guarantee, as in Proposition~\ref{prop:hedge}.

But why does solving this particular best-experts problem make sense for \pdbwk? Particularly, why does it make sense to maximize this notion of  expected payoffs? Let us view distribution $y_t$ as a vector of \emph{normalized costs} of resources. Consider the total expected normalized cost consumed by the algorithm after round $\narms$, denote it $W$.
Then
    $W = \sum_{t=\narms+1}^\stime y_t\; c(x_t,\mu)$.
A lower confidence bound on this quantity is
    $W_\LCB = \sum_{t=\narms+1}^\stime y_t\cdot \clcb{t,x_t}$,
which is precisely the total expected payoff in the best-experts problem. In the analysis, we relate $W_\LCB$ and the upper confidence bound on the total expected reward in the same rounds,
    $\Rew_\UCB = \sum_{t=\narms+1}^\stime u_{t,x_t}$
Specifically, we prove that for any implementation of step~\ref{algstep:mwu}, we have
\begin{align}\label{eq:pdbwk-intuition-mw}
\Rew_\UCB \geq W_\LCB\; \LPOPT/B
    \quad \text{with high probability}.
\end{align}
(This follows from \eqref{eq:pdbwk-intuition-mw-proof}.)
Thus, maximizing $W_\LCB$ is a reasonable goal for the cost update rule.}

Step~\ref{algstep:mwu} can also be seen as a
variant of the Garg-K\"onemann width
reduction technique \citep{GargK}. The ratio
$\rucb{t,x}/\EstCost_x$
that we optimize in step~\ref{algstep:argmin} may
be unboundedly large, so in the multiplicative update
in step~\ref{algstep:mwu} we rescale this value to
$\clcb{t,x}(i)$,
which is guaranteed
to be at most 1;
this rescaling is mirrored in the
analysis of the algorithm.
Interestingly, unlike the Garg-K\"onemann algorithm which
applies multiplicative updates to the dual
vectors and weighted averaging to the primal ones,
in our algorithm the multiplicative updates
and weighted averaging are \emph{both} applied to
the dual vectors.
\end{discussion}

\begin{discussion}
\asedit{From the primal-dual point of view, we could distinguish a ``primal" problem in which one chooses among arms, and a ``dual" problem in which one updates the cost vector. In the primal problem, the choice of costs is deemed adversarial, and the goal is to ensure \eqref{eq:pdbwk-intuition-mw}. In the dual problem, the choice of arms is deemed adversarial, and the goal is to maximize $W_\LCB$ so as to obtain Proposition~\ref{prop:hedge}. In both problems, one is agnostic as to how the upper/lower confidence bounds $u_t$ and $L_{t,x}$ are updated over time. As mentioned above, the dual problem falls under a standard setting of the ``best-expert" problem, and is solved via a standard algorithm for this problem. Meanwhile the primal problem is solved via bang-per-buck ratios and an ad-hoc application of the ``optimism under uncertainty" principle.

When the rewards and consumptions are deterministic,%
\footnote{Then the dual problem maximizes $W$ rather than $W_\LCB$, and the primal problem ensures \refeq{eq:pdbwk-det-intuition} rather than \refeq{eq:pdbwk-intuition-mw}, see Section~\ref{sec:pdbwk-determ}.}
the analysis is completely modular: it works no matter which algorithm is used to solve the primal (resp., dual) problem. In the general case, the primal algorithm also needs to ensure that the ``error terms" come out suitably small.
}
\end{discussion}

The following theorem expresses the regret guarantee
for \pdbwk.
\begin{theorem} \label{thm:pdbwk}
Consider an instance of \BwK with $d$ resources, $m=|X|$ arms, and the smallest budget $B = \min_i B_i$.
The regret of algorithm \pdbwk with parameter $\eps=\sqrt{\ln(d)/B}$ satisfies
\begin{equation} \label{eq:pdbwk-regret}
\LPOPT -\algrwd
\leq
O \left( \sqrt{\log(\nrsc \horizon)}\right)
\left( \sqrt{\narms \, \LPOPT} + \LPOPT \, \sqrt{\frac{\narms}{\budg}} \right)\;
+ O(\narms)\, \log(\nrsc \horizon)\log(\horizon).
\end{equation}
Moreover, \eqref{eq:regret-LPOPT-OPT} holds with $f(\LPOPT)$ equal to the right-hand side of \eqref{eq:pdbwk-regret}.
\end{theorem}

The rest of the section proves this theorem. Throughout, it will be useful
to represent the latent values as matrices and vectors. For this purpose, we will number the arms as $\arms = \{ 1 \LDOTS \narms \}$ and let $r \in \R^{\narms}$
denote the vector whose $x$-th component is $r(x,\mu)$, the expected reward, for each arm $x\in X$. Similarly we will let $C \in \R^{\nrsc \times \narms}$ denote the matrix whose
$(i,x)$ entry is $c_i(x,\mu)$, the expected resource consumption, for each resource $i$ and each arm $x$. Let $\sbv_j^d\in \{0,1\}^d$ denote the $d$-dimensional $j$-th coordinate vector.

While \pdbwk uses multiplicative weights update as a general technique, we make use of a specific performance guarantee in our analysis. To this end, let us recall algorithm $\hedge$ ~\citep{FS97} from online learning theory, also known as the multiplicative weights algorithm.  It is an online algorithm for maintaining
a $d$-dimensional probability vector $y$ while
observing a sequence of
$d$-dimensional payoff vectors $\pi_1,\ldots,\pi_\stime$. The version presented below, along with the following performance guarantee, is adapted from \citet{cs683week2}; a self-contained proof appears in Appendix~\ref{app:hedge}.

\begin{algorithm}[H]
\caption{$\mathsf{Hedge}$ with parameter $\eps\in (0,1)$}
\label{alg:hedge}
\begin{algorithmic}[1]
\STATE $v_1 = \ones$
    \TAB\COMMENT{$v_t \in \R^d_+$ for each round $t$.}
\FOR{each round $t=1,2,3,\,\ldots $}
\STATE {\bf Output} distribution $y_t = v_t / \| v_t \|_1$.
\STATE {\bf Input} payoff vector $\pi_t \in [0,1]^d$.
\FOR{each resource $i$}
\STATE $v_{t+1}(i) = v_t(i)\, (1+\eps)^\ell$, $\ell = \pi_{t}(i)$.
\ENDFOR
\ENDFOR
\end{algorithmic}
\end{algorithm}


\begin{proposition} \label{prop:hedge}
Fix any parameter $\eps\in (0,1)$ and any stopping time $\tau$. For any sequence of payoff vectors
    $\pi_1,\ldots,\pi_\stime \in [0,1]^d$, we have
\[
\forall y \in \splx{d} \quad
\sum_{t=1}^\stime y_t^\trans \pi_t \geq (1-\eps) \sum_{t=1}^\stime y^\trans \pi_t
- \frac{\ln d}{\eps}.
\]
\end{proposition}
\noindent

\OMIT{ 
Before giving the full proof of Theorem~\ref{thm:pdbwk}, we present a simplified analysis of $\pdbwk$ in a toy model in which the outcome vectors are deterministic, so as to illuminate the key ideas. (This toy model of the \BwK problem is uninteresting as a learning problem since the latent structure does not need to be learned, and the problem reduces to solving a linear program.) Then we proceed with the actual analysis, which we split into ``analysis modulo error terms" and ``analysis of error terms".
} 


\subsection{Warm-up: The deterministic case}
\label{sec:pdbwk-determ}

\ascomment{this subsection has essentially been re-written, for the sake of clarity.}

To present the application of $\hedge$ to \BwK in its purest form, we first consider the ``deterministic case" in which the rewards of the various arms are deterministically equal to the components of a vector $r \in \R^{\narms}$, and the resource consumption vectors are deterministically equal to the columns of a matrix $C \in \R^{\nrsc \times \narms}$. Then there is no need to use upper/lower confidence bounds, so the algorithm can be simplified considerably, see Algorithm~\ref{alg:determ}. In the remainder of this subsection we discuss this algorithm and analyze its regret.

\begin{algorithm}[ht]
\caption{Algorithm \pdbwk for deterministic outcomes, with parameter $\eps\in(0,1)$}
\label{alg:determ}
\begin{algorithmic}[1]
\STATE {\bf Initialization}
\STATE \TAB In the first $\narms$ rounds, pull each arm once.
\STATE \TAB For each arm $x\in X$, let $r_x\in [0,1]$ and $C_x\in[0,1]^\nrsc$
\STATE \TAB\TAB denote the reward and the resource consumption vector revealed in Step 2.
\STATE \TAB $v_1 = \ones \in [0,1]^\nrsc$.
\STATE \TAB\TAB
    \COMMENT{$v_t\in [0,1]^\nrsc$ is the round-$t$ estimate of the optimal solution $\eta^*$ to \refeq{lp:dual} in Section~\ref{sec:lp}.}
\STATE \TAB\TAB
    \COMMENT{We interpret $v_t(i)$ as an estimate of the (fictional) unit cost of  resource $i$, for each $i$.}
\STATE \TAB Set $\eps = \sqrt{\ln(\nrsc)/\budg}$.
\FOR{rounds $t = \narms+1,\ldots,\stime \;\; \mbox{\it (i.e., until
resource budget exhausted)}$}
\STATE For each arm $x\in X$,
\STATE \TAB \emph{Expected cost} for one pull of arm $x$ is estimated by $\EstCost_x = C_x\cdot v_t$.
\STATE Pull arm $x=x_t\in X$ that maximizes  $r_x/\EstCost_x$, the \emph{bang-per-buck} ratio.
\label{algstep:argmin-det}
\STATE Update estimated unit cost for each resource $i$:
\begin{align*}
                v_{t+1}(i) = v_t(i)\, (1+\eps)^{\ell},\;  \ell = C_x(i).
\end{align*}
\label{algstep:mwu-det}
\ENDFOR
\end{algorithmic}
\end{algorithm}


Algorithm~\ref{alg:determ} is an instance of the multiplicative-weights
update method for solving packing linear programs. Interpreting
it through the lens of online learning, as in the survey by \citet{AroraHK},
it is updating a vector $y_t = v_t / \| v_t \|_1$ using
the $\hedge$ algorithm, where the payoff vector in any round $t>\narms$ is given by $\pi_t = C_{x_t}$ and the goal is to optimize the total (expected) payoff
    $W = \sum_{t=\narms+1}^\stime y_t\cdot C_{x_t}$.
Note that $W$ is also the total cost consumed by Algorithm~\ref{alg:determ}.

To see why $W$ is worth maximizing, let us relate it to the total reward collected by the algorithm in rounds $t>\narms$; denote this quantity by
    $\algrwd = \sum_{t=\narms+1}^{\stime-1} r_t$.
We will prove that
\begin{align} \label{eq:pdbwk-det-intuition}
    \algrwd \geq W\cdot \LPOPT/B \;
    \text{for any implementation of Step~\ref{algstep:mwu-det}}.
\end{align}
For this reason, maximizing $W$ also helps maximize $\algrwd$. Proving it is a major step in the analysis.

Let $\xi^*$ denote an optimal solution of the primal linear
program~\refeq{lp:primal} from Section~\ref{sec:lp},
and let $\optlp = r^\trans \xi^*$ denote the optimal
value of that LP.

For each round $t$, let $\x_t = \sbv_{x_t}^\narms$ denote the $x_t$-th coordinate vector.
We claim that
\begin{align}\label{eq:PDBwK-pointmass-det}
z_t \in \argmax_{\x \in \splx{\armset}} \left\{ \frac{r^\trans \x}{y_t^\trans C \x} \right\}.
\end{align}
In words: $z_t$ maximizes the ``bang-per-buck ratio" among all distributions $z$ over arms. Indeed, the  $\argmax$ in \eqref{eq:PDBwK-pointmass-det} is well-defined as that of a continuous function on a compact set. Say it is attained by some distribution $z$ over arms, and let $\rho\in\R$ be the corresponding $\max$. By maximality of $\rho$, the linear inequality
    $ \rho\, y_t^\trans C \x \geq r^\trans \x $
also holds at some extremal point of the probability simplex $\splx{\armset}$, i.e. at some point-mass distribution. For any such point-mass distribution, the corresponding arm maximizes the bang-per-buck ratio in the algorithm. Claim proved.

\begin{proof}[Proof of \eqref{eq:pdbwk-det-intuition}]
It follows that
\begin{align*}
y_t^\trans \, \pi_t
    &= y_t^\trans C z_t
    \leq r_t\, \left(y_t^\trans C \xi^*\right)/\LPOPT \\
W   &= \sum_t y_t^\trans \, \pi_t
    \leq \frac{1}{\LPOPT} \sum_t r_t\, \left(y_t^\trans C \xi^*\right)
    = \frac{1}{\LPOPT} \left(\sum_t r_t\, y_t^\trans\right) C \xi^*.
\end{align*}
Here the sums are over rounds $t$ with $\narms <t <\stime$. Now, letting
    $\bar{y} = \frac{1}{\algrwd} \sum_t r_t\, y_t \in [0,1]^\nrsc$
be the rewards-weighted average of distributions $y_{m+1} \LDOTS y_\tau$, it follows that
\begin{align*}
W   \leq \frac{\algrwd}{\LPOPT}\; \bar{y}^T C \xi^*
    \leq \frac{\algrwd}{\LPOPT}\; B.
\end{align*}
The last inequality follows because all components of $C \xi^*$ are at most $B$ by the primal feasibility of $\xi^*$.
\end{proof}

Now, combining \eqref{eq:pdbwk-det-intuition} and the regret bound for \hedge, we obtain
\begin{align} \label{eq:pdbwk-det-main}
\algrwd
    \geq W\cdot \LPOPT/B
    \geq \left[ (1-\eps) \sum_{\narms < t < \stime} y\, \pi_t
        - \frac{\ln \nrsc}{\eps} \right]
        \cdot \frac{\LPOPT}{B}
    \quad \forall y \in \splx{d}.
\end{align}
To continue this argument, we need to choose an appropriate vector $y$ to make the right-hand side large. Recall that $\pi_t = C z_t$, so
$\sum_{\narms < t < \stime} \pi_t$ is simply the total consumption vector in all rounds $\narms < t < \stime$. We know some resource $i$ must be exhausted by the time the algorithm stops, so the consumption of this resource is at least $B$. In a formula:
    $\textstyle \sum_{t=1}^\stime\; y\, \pi_t \geq B$,
where $y=\sbv_i^\nrsc$ is the identity vector for resource $i$. Plugging in this $y$ into \eqref{eq:pdbwk-det-main}, we obtain:
\begin{align*}
\algrwd
    &\geq \left[ (1-\eps) (B-m-1)
        - \frac{\ln \nrsc}{\eps} \right]
        \cdot \frac{\LPOPT}{B} \\
    &\geq \LPOPT - \left[ \eps B +m+1
        + \frac{\ln \nrsc}{\eps} \right]
        \cdot \frac{\LPOPT}{B} \\
    &= \LPOPT -
        O(\sqrt{B \ln d} +m) \cdot \frac{\LPOPT}{B}
        \qquad \text{if $\eps = \sqrt{\frac{\ln \nrsc}{\budg}}$}.
\end{align*}
This completes regret analysis for the deterministic case.

\subsection{Analysis modulo error terms}
\label{sec:modulo-errors}

We now commence the analysis of Algorithm \pdbwk.
In this subsection we show how to reduce the problem
of bounding the algorithm's regret to a problem
of estimating two error terms that reflect the
difference between the algorithm's confidence-bound estimates
of its own reward and resource consumption with the
empirical values of these random variables.
The error terms will be treated in Section~\ref{sec:error-analysis}.

Recall that the algorithm computes LCBs on expected resource consumption $\clcb{t,x}\in [0,1]^\nrsc$ and UCBs on expected rewards $\rucb{t,x}\in [0,1]$, for each round $t$ and each arm $x$. We also represent the LCBs as a matrix
    $\clcb{t} \in [0,1]^{\nrsc \times \narms}$
whose $x$-th column equals $\clcb{t,x}$, for each arm $x$. We also represent the UCBs as a vector $\rucb{t}\in [0,1]^\narms$ over arms whose $x$-th component equals $\rucb{t,x}$. Let $C_t$ be the resource-consumption matrix for round $t$. That is,
    $C_t \in [0,1]^{\nrsc \times \narms}$
denotes the matrix whose $(i,x)$ entry is the actual consumption of resource $i$ in round $t$ if arm $x$ were chosen in this round.

As in the previous subsection, let $\x_t = \sbv_{x_t}^\narms$ denote the $x_t$-th coordinate vector, and let
    $y_t = v_t / \| v_t \|_1$
be the vector of normalized costs.
Similar to \eqref{eq:PDBwK-pointmass-det}, $z_t$ maximizes the ``bang-per-buck ratio" among all distributions $z$ over arms:
\begin{align}\label{eq:PDBwK-pointmass}
z_t \in \argmax_{\x \in \splx{\armset}} \left\{ \frac{\rucb{t,x}^\trans \x}{y_t^\trans \clcb{t,x}\, \x} \right\}.
\end{align}

By Theorem~\ref{thm:conf-rad} and our choice of $\chernoffC$, it
holds with probability at least $1-T^{-1}$ that the confidence interval
for every latent parameter, in every round of execution, contains
the true value of that latent parameter. We call this high-probability
event a \emph{clean execution} of \pdbwk. Our regret guarantee will
hold deterministically assuming that a clean execution takes place.
The regret can be at most $T$ when a clean execution does not take
place, and since this event has probability at most $T^{-1}$ it
contributes only $O(1)$ to the regret. We will henceforth assume
a clean execution of \pdbwk.

\begin{claim}\label{cl:modulo-errors}
In a clean execution of Algorithm \pdbwk with parameter $\eps=\sqrt{\ln(d)/B}$, the algorithm's
total reward satisfies the bound
\begin{equation} \label{eq:modulo-errors}
\optlp - \algrwd \leq  \left[
  2 \optlp \left( \sqrt{\frac{\ln \nrsc}{\budg}} +
                \frac{\narms + 1}{\budg} \right) +
  \narms + 1 \right] +
  \frac{\optlp}{\budg} \left\| \sum_{\narms<t<\stime} \cdif{t} \x_t \right\|_\infty +
  \left| \sum_{\narms<t<\stime} \rdif{t}^\trans \x_t \right|,
\end{equation}
where
    $\cdif{t} = \crlz{t} - \clcb{t}$
and
    $\rdif{t} = \rucb{t} - \rrlz{t}$
for each round $t>\narms$.
\end{claim}


\begin{proof}
The claim is proven by mimicking the analysis of
Algorithm~\ref{alg:determ} in the preceding section,
incorporating error terms that reflect the differences
between observable values and latent ones.
As before, let $\xi^*$ denote an optimal solution of the primal linear
program~\refeq{lp:primal},
and let $\optlp = r^\trans \xi^*$ denote the optimal value of that LP.
Let $\optucb = \sum_{\narms < t < \stime} \rucb{t}^\trans \x_t$ denote the
total payoff the algorithm would have obtained, after its initialization
phase, if the actual
payoff at time $t$ were replaced with the upper confidence bound.
Let $y=\sbv_i^\nrsc$, where $i$ is a resource exhausted by the algorithm when it stops; then
    $y^\trans \left(\sum_{t=1}^\stime \crlz{t} \x_t \right) \geq \budg$.
As before,
\begin{equation} \label{eq:pdbwk.1}
y^\trans \left( \sum_{\narms < t < \stime} \crlz{t} \x_t \right) \geq \budg - \narms - 1.
\end{equation}
Finally let
\[
\bar{y} = \frac{1}{\optucb} \sum_{\narms < t < \stime} (\rucb{t}^\trans \x_t) y_t.
\]
Assuming a clean execution, we have
{\allowdisplaybreaks
\begin{align}
\budg & \geq \bar{y}^\trans C \xi^*
  && \mbox{\it ($\xi^*$ is primal feasible)}
\nonumber \\
  &= \frac1{\optucb} \sum_{\narms < t < \stime}
     (\rucb{t}^\trans \x_t)(y_t^\trans C \xi^*)
   && \mbox{\it (definition of $\bar{y}$)}
\nonumber \\
  & \geq \frac{1}{\optucb} \sum_{\narms<t<\stime} (\rucb{t}^\trans \x_t) (y_t^\trans \clcb{t} \xi^*)
 && \mbox{\it (clean execution)}
\nonumber \\
  & \geq \frac{1}{\optucb} \sum_{\narms<t<\stime} (\rucb{t}^\trans \xi^*) (y_t^\trans \clcb{t} \x_t)
 && \mbox{\it (by \eqref{eq:PDBwK-pointmass})}
\nonumber \\
    \label{eq:pdbwk-intuition-mw-proof}
  & \geq \frac{1}{\optucb} \sum_{\narms<t<\stime} (r^\trans \xi^*) (y_t^\trans \clcb{t} \x_t)
 && \text{\it (clean execution)} \\
  & \geq \frac{\optlp}{\optucb} \left[
         (1 - \eps) y^\trans \left( \sum_{\narms<t<\stime} \clcb{t} \x_t \right)
         - \frac{\ln \nrsc}{\eps} \right]
 && \mbox{\it (Hedge guarantee)}
\nonumber \\
  & > (1-\eps)\,\frac{\optlp}{\optucb} \left[
           y^\trans \left( \sum_{\narms<t<\stime} \crlz{t} \x_t \right)
         - y^\trans \left( \sum_{\narms<t<\stime} \cdif{t} \x_t \right)
         - \frac{\ln \nrsc}{\eps} \right]
\nonumber \\
  & \geq \frac{\optlp}{\optucb} \left[
           (1-\eps) \budg - \narms - 1
         - (1-\eps) y^\trans \left( \sum_{\narms<t<\stime} \cdif{t} \x_t \right)
         - \frac{\ln \nrsc}{\eps} \right]
 && \mbox{\it (definition of $y$; see eq.~\refeq{eq:pdbwk.1})}
\nonumber \\
\optucb
  & \geq \optlp \left[ 1 - \eps - \frac{\narms+1}{\budg} -
         \frac{1}{\budg} \left\| \sum_{\narms<t<\stime} \cdif{t} \x_t \right\|_{\infty}
         - \frac{\ln \nrsc}{\eps \budg} \right] .
\label{eq:optucb}
\end{align}
} 
The algorithm's actual payoff,
    $\algrwd = \sum_{t=1}^\stime \rrlz{t}^\trans \x_t$,
satisfies the inequality
    $$\algrwd \geq
            \optucb - \sum_{\narms<t<\stime} (\rucb{t} - \rrlz{t})^\trans \x_t
          = \optucb - \sum_{\narms<t<\stime} \rdif{t}^\trans \x_t.
    $$
Combining this with~\refeq{eq:optucb}, and plugging in $\eps=\sqrt{\ln(d)/B}$,
we obtain the bound~\refeq{eq:modulo-errors}, as claimed.
\end{proof}

\subsection{Error analysis}
\label{sec:error-analysis}

We complete the proof of Theorem~\ref{thm:pdbwk}
by proving upper bounds on the terms
$\left\| \sum_{\narms<t<\stime} \cdif{t} \x_t \right\|_\infty$
and $\left| \sum_{\narms<t<\stime} \rdif{t} \x_t \right|$
that appear on the right side of~\refeq{eq:modulo-errors}.
Both bounds follow from a more general lemma which we present below.

\asedit{
The general lemma considers a sequence of vectors $a_1,\ldots,a_\stime $ in $[0,1]^\narms$ and another vector $a_0\in [0,1]^\narms$. Here $a_{t,x}\in [0,1]$ represents a numerical outcome (\ie a reward or a consumption of a given resource) if arm $x$ is pulled in round $t$, and $a_{0,x}$ represents the corresponding expected outcome.  Further, for each round $t>\narms$ we have an estimate $b_t\in [0,1]^\narms$ for the outcome vector $a_t$ . We only assume a clean execution of the algorithm, and we derive an upper bound on
$ \left| \sum_{\narms < t < \stime} (b_t - a_t)^\trans \x_t \right| $.


\begin{lemma} \label{lem:vector-conf}
Consider two sequences of vectors
    $a_1,\ldots,a_\stime$ and $b_1,\ldots,b_\stime$,
in $[0,1]^\narms$, and a vector $a_0\in [0,1]^\narms$. For each arm $x$ and each round $t>\narms$, let $\aemp{t,x}\in [0,1]$ be the average observed outcome up to round $t$, \ie the average outcome $a_{s,x}$ over all rounds $s\leq t$ in which arm $x$ has been chosen by the algorithm; let $\npul_{t,x}$ be the number of such rounds. Assume that for each arm $x$ and all rounds $t$ with $\narms<t<\stime$ we have
\begin{align*}
|b_{t,x} - a_{0,x}|
    &\leq 2 \, \rad(\aemp{t,x},\npul_{t,x})
    \leq 6 \, \rad(a_{0,x}, \npul_{t,x}),\\
|\aemp{t,x} - a_{0,x}|
    &\leq \rad(\aemp{t,x},\npul_{t,x}).
\end{align*}
Let $\asum = \sum_{t=1}^{\stime-1} a_{t,x_t}$ be the total outcome collected by the algorithm. Then
\begin{equation} \label{eq:vector-conf}
\left| \sum_{\narms < t < \stime} (b_t - a_t)^\trans \x_t \right| \leq
O \left( \sqrt{\chernoffC \, \narms \asum} + \chernoffC \, \narms\, \log\horizon \right).
\end{equation}
\end{lemma}
} 


Before proving the lemma, we need to establish
a simple fact about confidence radii.

\begin{claim} \label{cl:confrad-sum}
For any two vectors $a,M \in \R_+^\narms$, we have
\begin{equation} \label{eq:confrad-sum}
\textstyle \sum_{x=1}^{\narms}\; \rad(a_x,M_x) \, M_x \leq
\sqrt{\chernoffC \, m (a^\trans M)} + \chernoffC \, m.
\end{equation}
\end{claim}
\begin{proof}
The definition of $\rad(\cdot,\cdot)$ implies that
$\rad(a_x,M_x) \, M_x \leq \sqrt{\chernoffC \, a_x M_x} + \chernoffC$.
Summing these inequalities and applying Cauchy-Schwarz,
\[
\textstyle  \sum_{x=1}^{\narms} \rad(a_x,M_x) \, M_x \leq
\sum_{x=1}^{\narms} \sqrt{\chernoffC \, a_x M_x} + \chernoffC\, m \leq
\sqrt{m} \cdot \sqrt{\sum_{x \in \arms} \chernoffC \,  a_x M_x} + \chernoffC\, m,
\]
and the lemma follows by rewriting the expression on the right side.
\end{proof}

\begin{proof}[Proof of Lemma~\ref{lem:vector-conf}]
For convenience, denote $\vnpul_t = (\npul_{t,1} \LDOTS \npul_{t,\narms})$, and observe that
\[\textstyle
    \vnpul_t = \sum_{s=1}^t \x_s
\quad\text{and}\quad
    A= \overline{a}_{\stime-1}^\trans \vnpul_{\stime-1} = \sum_{s=1}^{\stime-1} a_s^\trans\, \x_s.
\]

We decompose the left side of~\refeq{eq:vector-conf}
as a sum of three terms,
\begin{equation} \label{eq:vc.1}
\sum_{\narms < t < \stime} (b_t - a_t)^\trans \x_t =
\sum_{t=1}^{\narms} (a_t - b_t)^\trans \x_t +
\sum_{t=1}^{\stime-1} (b_t - a_0)^\trans \x_t +
\sum_{t=1}^{\stime-1} (a_0 - a_t)^\trans \x_t,
\end{equation}
then bound the three terms separately. The
first sum is clearly bounded above by $\narms$.
We next work on bounding the third sum.
Let $s=\stime-1$.
\begin{align}
\left| (a_0 - \aemp{s})^\trans\,\, \vnpul_s \right|
    &\leq \sum_{x \in \arms} \rad(\aemp{s,x},\npul_{s,x})\, \npul_{s,x}
        && \text{{\em (assuming clean execution)}}  \nonumber \\
    &\leq \sqrt{\chernoffC \, \narms \asum} + \chernoffC \, \narms.
   && \text{{\em (by Claim~\ref{cl:confrad-sum})}} \label{eq:vc.2} \\
\sum_{t=1}^{s} (a_0-a_t)^\trans\, \x_t
  &=
    a_0^\trans\, \vnpul_s - \sum_{t=1}^{s} a_t^\trans\, \x_t
   =
    (a_0 - \aemp{s})^\trans\, \vnpul_s.
 \nonumber \\
\left| \sum_{t=1}^{s} (a_0-a_t)^\trans\, \x_t  \right|
    &= \left| (a_0 - \aemp{s})^\trans\,\, \vnpul_s \right|
    \leq \sqrt{\chernoffC \, \narms \asum} + \chernoffC \, \narms. \nonumber
\end{align}

Finally we bound the middle sum in~\refeq{eq:vc.1}.
\begin{align}
      \left| \sum_{t=1}^{s} (b_t - a_0)^\trans \x_t \right|
  &\leq
      6 \sum_{t=1}^{s} \sum_{x \in \arms} \rad(a_{0,x}, \npul_{t,x}) \x_{t,x}
\nonumber \\
  &=
      6 \sum_{x \in \arms} \sum_{\ell=1}^{\npul_{s,x}} \rad(a_{0,x}, \ell)
\nonumber \\
  &=
      O \left( \sum_{x \in \arms}
            \sqrt{\chernoffC\, a_{0,x}\,\npul_{s,x}} + \chernoffC \log(\npul_{s,x})
        \right)
\nonumber \\
  & \leq
      O \left( \sqrt{\chernoffC \, \narms\, a_0^\trans\, \vnpul_s} + \chernoffC\, \narms\,            \log \horizon
        \right). \label{eq:vc.3}
\end{align}

We would like to replace the expression $a_0^\trans\, \vnpul_s$ on
the last line with the expression $\aemp{s}^\trans \,\vnpul_s = \asum$.
To do so, recall \eqref{eq:vc.2} and apply the following calculation:
\begin{align*}
a_0^\trans\, \vnpul_s
&\leq
    \aemp{s}^\trans \vnpul_s + \sqrt{\chernoffC \, \narms \asum}
         + \chernoffC \, \narms \\
&=
\asum + \sqrt{\chernoffC \, \narms \asum} + \chernoffC \, \narms \\
&\leq
\left( \sqrt{\asum} + \sqrt{\chernoffC \, \narms} \right)^2
\\
\sqrt{\chernoffC \, \narms a_0^\trans\, \vnpul_s}
   &\leq
\sqrt{\chernoffC \, \narms}
\left( \sqrt{\asum} + \sqrt{\chernoffC \, \narms} \right) =
\sqrt{\chernoffC \, \narms \asum} + \chernoffC \, \narms.
\end{align*}
Plugging this into \eqref{eq:vc.3}, we bound the middle sum in~\refeq{eq:vc.1} as
\begin{align}
      \left| \sum_{t=1}^{s} (b_t - a_0)^\trans \x_t \right|
 \leq
      O \left( \sqrt{\chernoffC \, \narms\asum} + \chernoffC\, \narms\,\log \horizon
        \right). \label{eq:vc.4}
\end{align}

Summing up the upper bounds for the three terms
on the right side of~\refeq{eq:vc.1}, we obtain~\refeq{eq:vector-conf}.
\end{proof}

\begin{corollary} \label{cor:error-terms}
In a clean execution of \pdbwk,
$$\left| \sum_{\narms<t<\stime} \rdif{t} \x_t \right| \leq
O \left( \sqrt{\chernoffC \, \narms  \algrwd} + \chernoffC \, \narms\, \log\horizon \right)$$
and
$$\left\| \sum_{\narms<t<\stime} \cdif{t} \x_t \right\|_{\infty} \leq
O \left( \sqrt{\chernoffC \, \narms  \budg} + \chernoffC \, \narms \, \log\horizon \right).$$
\end{corollary}
\begin{proof}
The first inequality is obtained by applying Lemma~\ref{lem:vector-conf}
with vector sequences $a_t = r_t$ and $b_t = u_t$, \asedit{and vector $a_0=r$. In other words, $a_0$ is the vector of expected rewards across all arms.}

The second inequality is obtained by applying the same lemma separately for each resource $i$, with vector sequences $a_t = (\sbv^\nrsc_i)^\trans\, C_t$ and
$b_t = \sbv^\nrsc_i\, L_t$, \asedit{and vector $a_0$ being the $i$-th row of matrix $C$. In other words, $a_0$ is the vector of expected consumption of resource $i$ across all arms.}
\end{proof}

\begin{proofof}{Theorem~\ref{thm:pdbwk}}
If $\narms \geq \budg/\log(\nrsc\horizon)$, then the regret bound in Theorem~\ref{thm:pdbwk} is trivial. Therefore we can assume without loss of generality that $\narms \leq \budg/\log(\nrsc\horizon)$.
Therefore, recalling \eqref{eq:modulo-errors}, we observe that
\[
2 \, \optlp \left( \sqrt{\frac{\ln \nrsc}{\budg}} + \frac{\narms+1}{\budg} \right)
=
O \left( \sqrt{\narms \log(\nrsc \narms \horizon)} \;
\frac{\optlp}{\sqrt{\budg}} \right).
\]
The term $m+1$ on the right side of~\eqref{eq:modulo-errors}
is bounded above by $m \log(\nrsc \narms \horizon)$.
Finally, using Corollary~\ref{cor:error-terms} we
see that the sum of the final two terms on the right side
of~\refeq{eq:modulo-errors} is bounded by
\begin{align*}
O \left( \sqrt{\chernoffC \, \narms} \left( \frac{\optlp}{\sqrt{\budg}} + \sqrt{\optlp}\right)
    + \chernoffC \, \narms \, \log\horizon   \right).
\end{align*}
The theorem follows by plugging in
    $\chernoffC = \Theta(\log(\nrsc \narms \horizon)) = O(\log(\nrsc \horizon)) $
(because $\narms\leq\budg\leq\horizon$).
\end{proofof}

%
%


\section{Lower Bound}
\label{sec:LB}

\newcommand{\Fam}{\mathcal{G}}  

\newcommand{\FamT}[1][]{\ensuremath{\Fam^{#1}_{(p,\eps)}}\xspace}
\newcommand{\FamInf}{\FamT[\infty]}

\newcommand{\KL}[3][]{\ensuremath{\mathtt{KL}_{#1}(#2\,\|\,#3)}}
\newcommand{\Wt}{\vec{w}_t}

\newcommand{\Id}[1]{{\bf 1}_{\{#1\}}}

We prove that regret~\refeq{eq:intro-regret} obtained by algorithm \pdbwk is optimal up to polylog factors. Specifically, we prove that any algorithm for \BwK must, in the worst case, incur regret
\begin{align}\label{eq:thm:lb}
\Omega\left( \min\left( \OPT,\;
    \OPT \sqrt{\frac{\narms}{\budg}}
    +\sqrt{\narms\,\OPT}\right)\right),
\end{align}
where $\narms = |X|$ is the number of arms and $\budg = \min_i \budg_i$ is the smallest budget.

\begin{theorem}\label{thm:lb}
Fix any $\narms\geq 2$, $d\geq 1$, $\OPT\geq \narms$, and
    $(\budg_1 \LDOTS \budg_d)\in [2,\infty)$.
Let $\Fam$ be the family of all \BwK problem instances with $\narms$ arms, $d$ resources, budgets $(\budg_1 \LDOTS \budg_d)$ and optimal reward $\OPT$. Then any algorithm for \BwK must incur regret~\refeq{eq:thm:lb} in the worst case over $\Fam$.
\end{theorem}

We treat the two summands in~\eqref{eq:thm:lb} separately:

\begin{claim}\label{cl:lb}
Consider the family $\Fam$ from Theorem~\ref{thm:lb}, and let \ALG be some algorithm for \BwK.
\begin{OneLiners}
\item[(a)] \ALG incurs regret
    $\Omega\left(\min\left( \OPT,\,\sqrt{\narms\,\OPT}\,\right)\right)$
in the worst case over $\Fam$.
\item[(b)] \ALG incurs regret
    $\Omega\left(\min\left( \OPT,\,\OPT \sqrt{\frac{\narms}{\budg}}\,\right)\right)$
in the worst case over $\Fam$.
\end{OneLiners}
\end{claim}

Theorem~\ref{thm:lb} follows from Claim~\ref{cl:lb}(ab). For part (a), we use a standard lower-bounding example for MAB. For part (b), we construct a new example, specific to \BwK, and analyze it using KL-divergence.

\begin{proof}[Proof of Claim~\ref{cl:lb}(a)]
Fix $\narms\geq 2$ and $\OPT\geq \narms$.
Let $\Fam_0$ be the family of all MAB problem instances with $\narms$ arms and time horizon $T = \flr{2\,\OPT}$, where the ``best arm'' has expected reward $\mu^* = \OPT/T$ and all other arms have reward $\mu^*-\eps$ with $\eps = \tfrac14 \sqrt{\narms/T}$.
Note that
    $\mu^*\in[\tfrac12, \tfrac34]$
and
    $\eps \leq \tfrac14$.
It is well-known \citep{bandits-exp3} that any MAB algorithm incurs regret
    $\Omega(\,\sqrt{\narms\,\OPT}\,)$
in the worst case over $\Fam_0$.

To ensure that $\Fam_0\subset \Fam$, let us treat each MAB instance in $\Fam_0$ as a \BwK instance with $d$ resources, budgets $(\budg_1 \LDOTS \budg_d)$, and no resource consumption.
\end{proof}

\subsection{The new lower-bounding example: proof of Claim~\ref{cl:lb}(b)}

Our lower-bounding example is very simple. There are $\narms$ arms. Each arm gives reward $1$ deterministically. There is a single resource with budget $B$.%
\footnote{More formally, other resources in the setting of Theorem~\ref{thm:lb} are not consumed. For simplicity, we leave them out.} The resource consumption, for each arm and each round, is either $0$ or $1$. The expected resource consumption is $p-\eps$ for the ``best arm" and $p$ for all other arms, where $0<\eps<p<1$. There is time horizon $T<\infty$. Let \FamT denote the family of all such problem instances, for fixed parameters $(p,\eps)$. We analyze this family in the rest of this section.

We rely on the following fact about stopping times of random sums. For the sake of completeness, we provide a proof in Section~\ref{app:facts}.

\begin{fact}\label{fact:stopping}
Let $S_t$ be the sum of $t$ i.i.d. 0-1 variables with expectation $q$. Let $\tau^*$ be the first time this sum reaches a given number $B\in \N$. Then $\E[\tau^*] = B/q$. Moreover, for each $T>\E[\tau^*]$ it holds that
$$ \textstyle \sum_{t>T} \; \Pr[\tau^*\geq t] \leq \E[\tau^*]^2/T.
$$
\end{fact}

\xhdr{Infinite time horizon.}
It is convenient to consider the family of problem instances which is the same as \FamT except that it has the \emph{infinite} time horizon; denote it \FamInf. We will first prove the desired lower bound for this family, then extend it to $\FamT$.

The two crucial quantities that describe algorithm's performance on an instance in \FamInf is the stopping time and the total number of plays of the best arm. (Note that the total reward is equal to the stopping time minus 1.) The following claim connects these two quantities.

\begin{claim}[Stopping time]\label{cl:LB-Fam-stopping}
Fix an algorithm \ALG for \BwK and a problem instance in \FamInf. Consider an execution of \ALG on this problem instance. Let $\tau$ be the stopping time of \ALG. For each round $t$, let $N_t$ be the number of rounds $s\leq t$ in which the best arm is selected. Then
\begin{align*}
   p \E[\tau] - \eps \E[N_\tau] = \flr{\budg+1}.
\end{align*}
\end{claim}
\begin{proof}
Let $C_t$ be the total resource consumption after round $t$. Note that
    $\E[C_t] = pt-\eps N_t$.
We claim that
\begin{align}\label{eq:pf:cl:LB-Fam-regret}
    \E[C_\tau] = \E[ p\tau -\eps N_\tau].
\end{align}
Indeed, let $Z_t = C_t-(pt-\eps N_t)$. It is easy to see that $Z_t$ is a martingale with bounded increments, and moreover that $\Pr[\tau<\infty]=1$. Therefore the Optional Stopping Theorem applies to $Z_t$ and $\tau$, so that $\E[Z_\tau] = E[Z_0]=0$. Therefore we obtain \eqref{eq:pf:cl:LB-Fam-regret}.

To complete the proof, it remains to show that $C_\tau = \flr{B+1}$. Recall that \ALG stops if and only if $C_t>B$. Since resource consumption in any round is either $0$ or $1$, it follows that $C_\tau = \flr{B+1}$.
 \end{proof}

\begin{corollary} \label{cor:LB-stopping}
Consider the setting in Claim~\ref{cl:LB-Fam-stopping}. Then:
\begin{OneLiners}
\item[(a)] If \ALG always chooses the best arm then
    $ \E[\tau] = \flr{B+1}/(p-\eps)$.
\item[(b)] $\OPT= \flr{B+1}/(p-\eps)-1$ for any problem instance in \FamInf.
\item[(c)] $p \E[\tau] - \eps \E[N_\tau] = (p-\eps)\, (1+\OPT)$.
\end{OneLiners}
\end{corollary}


\begin{proof}
For part(b), note that we have
    $\E[\tau] \leq \flr{B+1}/(p-\eps)$,
so
    $\OPT \leq \flr{B+1}/(p-\eps) -1$.
By part (a), the equality is achieved by the policy that always selects the best arm.
\end{proof}

The heart of the proof is a KL-divergence argument which bounds the number of plays of the best arm. This argument is encapsulated in the following claim, whose proof is deferred to Section~\ref{sec:KL-div}.

\begin{lemma}[best arm]\label{lm:best-arm}
Assume $p\leq \tfrac12$ and $\tfrac{\eps}{p} \leq \tfrac{1}{16} \sqrt{\tfrac{\narms}{B}}$.
Then for any \BwK algorithm there exists a problem instance in \FamInf such that the best arm is chosen at most $\tfrac34\, \OPT$ times in expectation.
\end{lemma}

Armed with this bound and Corollary~\ref{cor:LB-stopping}(c), it is easy to lower-bound regret over \FamInf.

\begin{claim}[regret]\label{cl:LB-Fam-regret-infinite}
If $p\leq \tfrac12$ and
$\tfrac{\eps}{p} \leq \tfrac{1}{16} \sqrt{\tfrac{\narms}{B}}$ then
any \BwK algorithm incurs regret
    $\tfrac{\eps}{4p}\, \OPT$
over \FamInf.
\end{claim}
\begin{proof}
Fix any algorithm \ALG for \BwK. Consider the problem instance whose existence is guaranteed by Lemma~\ref{lm:best-arm}. Let $\tau$ be the stopping time of \ALG, and let $N_t$ be the number of rounds $s\leq t$ in which the best arm is selected. By Lemma~\ref{lm:best-arm} we have
    $\E[N_\tau] \leq \tfrac34\, \OPT$.
Plugging this into Corollary~\ref{cor:LB-stopping}(c) and rearranging the terms, we obtain
    $ \E[\tau] \leq (1+\OPT)(1-\tfrac{\eps}{4p})$.
Therefore, regret of \ALG is
    $\OPT-(\E[\tau]-1) \geq \tfrac{\eps}{4p}\, \OPT$.
\end{proof}

Thus, we have proved the lower bound for the infinite time horizon.

\xhdr{Finite time horizon.}
Let us ``translate'' a regret bound for $\FamInf$ into a regret bound for $\FamT$.

We will need a more nuanced notation for $\OPT$. Consider the family of problem instances in $\FamT\cup\FamInf$ with a particular time horizon $T\leq \infty$. Let $\OPT_{(p,\eps,T)}$ be the optimal expected total reward for this family (by symmetry, this quantity does not depend on which arm is the best arm). We will write
    $\OPT_T = \OPT_{(p,\eps,T)}$
when parameters $(p,\eps)$ are clear from the context.

\begin{claim}\label{cl:KL-OPT-T}
For any fixed $(p,\eps)$ and any $T>\OPT_\infty$ it holds that
$\OPT_T \geq \OPT_\infty - \OPT^2_\infty/T$.
\end{claim}
\begin{proof}
Let $\tau^*$ be the stopping time of a policy that always plays the best arm on a problem instance in \FamInf.
\begin{align*}
\OPT_\infty - \OPT_T
    &= \E[\tau^*] - \E[\min(\tau^*,T)] \\
    &= \textstyle \sum_{t>T}\; (t-T)\; \Pr[\tau^*=t] \\
    &= \textstyle \sum_{t>T}\; \Pr[\tau^*\geq t] \\
    &\leq \E[\tau^*]/T^2 = \OPT_\infty^2/T.
\end{align*}
The inequality is due to Fact~\ref{fact:stopping}.
\end{proof}

\begin{claim}\label{cl:LB-Fam-regret}
Fix $(p,\eps)$ and fix algorithm \ALG. Let $\reg_T$ be the regret of \ALG over the problem instances in $\FamT\cup\FamInf$ with a given time horizon $T\leq \infty$. Then
$ \reg_T \geq \reg_\infty - \OPT_\infty^2/T.$
\end{claim}

\begin{proof}
For each problem instance $\adv{}\in\FamInf$, let $\Rew_T(\adv{})$ be the expected total reward of \ALG on $\adv{}$, if the time horizon is $T\leq \infty$. Clearly,
    $\Rew_\infty(\adv{}) \geq \Rew_T(\adv{})$.
Therefore, using Claim~\ref{cl:KL-OPT-T}, we have:
\begin{align*}
\reg_T
    &= \OPT_T - \inf_{\adv{}}  \Rew_T(\adv{}) \\
    &\geq \OPT_T - \inf_{\adv{}} \; \Rew_\infty(\adv{}) \\
    &= \reg_\infty +\OPT_T - \OPT_\infty \\
    &\geq \reg_\infty - \OPT_\infty^2/T. \qquad \qedhere
\end{align*}
\end{proof}

\begin{lemma}[regret: finite time horizon]\label{lm:LB-Fam-regret}
Fix
    $p\leq \tfrac12$
and
    $\eps = \tfrac{p}{16}\,\min(1,\sqrt{\narms/B})$.
Then for any time horizon $T>\tfrac{8p}{\eps}\, \OPT_\infty$ and any \BwK algorithm \ALG there exists a problem instance in \FamT with time horizon $T$ for which \ALG incurs regret
    $\Omega(\OPT_T)\,\min(1,\sqrt{\narms/B})$.
\end{lemma}

\begin{proof}
By Claim~\ref{cl:LB-Fam-regret-infinite}, \ALG incurs regret at least $\tfrac{\eps}{4p}\, \OPT_\infty$ for some problem instance in \FamInf. By Claim~\ref{cl:LB-Fam-regret}, \ALG incurs regret at least $\tfrac{\eps}{8p}\, \OPT_\infty$ for the same problem instance in \FamT with time horizon $T$. Since
    $\OPT_\infty\geq \OPT_T$,
this regret is at least
    $\tfrac{\eps}{8p}\, \OPT_T =  \Omega(\OPT_T)\,\min(1,\sqrt{\narms/B})$.
\end{proof}

\OMIT{Let $\OPT_T$ be the optimal expected total reward for a problem instance in \FamT with time horizon $T$.}

Let us complete the proof of Claim~\ref{cl:lb}(b). Recall that Claim~\ref{cl:lb}(b) specifies the values for $(\narms,\budg,\OPT)$ that our problem instance must have. Since we have already proved Claim~\ref{cl:lb}(a) and
    $\OPT \sqrt{\frac{\narms}{\budg}} \leq O(\sqrt{\narms\, \OPT})$
for $\OPT<3\budg$, it suffices to assume $\OPT\geq 3\budg$.

Let
    $\eps(p) = \tfrac{p}{16}\,\min(1,\sqrt{\narms/B})$,
as prescribed by Lemma~\ref{lm:LB-Fam-regret}. Then taking $\eps=\eps(p)$ we obtain regret
    $\Omega(\OPT_T)\,\min(1,\sqrt{\narms/B})$
for any parameter $p\leq \tfrac12$ and any time horizon
    $T> \tfrac{8p}{\eps}\,\OPT_{(p,\eps,\infty)}$.
It remains to pick such $p$ and $T$ so as to ensure that
    $f(p,T) = \OPT$,
where
    $f(p,T) = \OPT_{(p,\eps(p),T)}$.

Recall from Corollary~\ref{cor:LB-stopping}(b) that
    $\OPT_{(p,\eps,\infty)} = \tfrac{\Gamma}{p}-1$,
where
$$
 \Gamma = \flr{B+1}/\left( 1-\tfrac{1}{16}\, \min(1,\sqrt{\narms/B})\right)
$$
is a ``constant" for the purposes of this argument, in the sense that it does not depend on $p$ or $T$. So we can state the sufficient condition for proving  Claim~\ref{cl:lb}(b) as follows:
\begin{align}\label{eq:LB-proof-endgame-goal}
\text{Pick $p\leq \tfrac12$ and $T\geq \tfrac{8\Gamma}{\eps(p)}$ such that
    $f(p,T) = \OPT$.}
\end{align}

Recall that
    $\OPT_{(p,\eps,\infty)} \geq \OPT_{(p,\eps,T)}$ for any $T$,
and
    $\OPT_{(p,\eps,T)} \geq \tfrac12\,\OPT_{(p,\eps,\infty)}$
for any $T>2\,\OPT_{(p,\eps,\infty)}$ by Claim~\ref{cl:KL-OPT-T}. We summarize this as follows: for any $T>2 (\tfrac{\Gamma}{p}-1)$,
\begin{align}\label{eq:KL-bound-OPT}
    \tfrac{\Gamma}{p}-1 \geq  \OPT_{(p,\eps,T)} \geq  \tfrac12 (\tfrac{\Gamma}{p}-1).
\end{align}

Define $p_0 = \Gamma/\OPT$. Since $\OPT\geq 3\budg$,
    $\Gamma \leq \tfrac{16}{15}(B+1)$
and $B\geq 4$, it follows that $p_0 \geq \tfrac12$.
Let $T = \tfrac{8\Gamma}{\eps(p_0)}$. Then \eqref{eq:KL-bound-OPT} holds for all $p\in [p_0/4, \tfrac12]$. In particular,
$$ f(p_0, T) \leq \Gamma/p_0 = \OPT \leq f(p_0/4, T).
$$
Since $f(p,T)$ is continuous in $p$, there exists $p\in [p_0/4, p_0]$ such that
    $f(p,T) = \OPT$.
Since $p\leq p_0$, we have $T\geq \tfrac{8\Gamma}{\eps(p)}$, satisfying all requirements in \eqref{eq:LB-proof-endgame-goal}. This completes the proof of Claim~\ref{cl:lb}(b), and therefore the proof of Theorem~\ref{thm:lb}.

\OMIT{ 

Finally, let us use this regret bound to prove Claim~\ref{cl:lb}(b). We consider a minor generalization of \FamInf where each reward is deterministically equal to $r$, for some parameter $r\in(0,1]$. This generalization simply rescales all rewards (and hence, the value of $\OPT$) by the factor of $r$. Thus, the regret bound (Claim~\ref{lm:LB-Fam-regret}) holds as is, and the only change is that now
    $\OPT = r\, \flr{B+1}/(p-\eps)$.

Given $(\budg,\OPT,\narms)$, we need to pick the parameters $r\in (0,1]$, $p\in(0,1)$ and $\eps\in (0,p)$ to achieve the desired regret and match the desired value for $\OPT$. To achieve the desired regret, we set
    $\eps = \gamma p$,
where
    $\gamma = \tfrac{p}{8}\, \sqrt{\tfrac{\min(\narms,B)}{B}}$,
as prescribed by Lemma~\ref{lm:LB-Fam-regret}
Then
    $\OPT = \tfrac{r}{p}\, \flr{B+1}/(1-\gamma)$.
We can adjust $r$ and $p$ to achieve any desired positive value for the ratio $\tfrac{r}{p}$, so in particular we can achieve the desired value for $\OPT$.

This completes the proof of Claim~\ref{cl:lb}(b) and Theorem~\ref{thm:lb}.

} 

\subsection{Background on KL-divergence (for the proof of Lemma~\ref{lm:best-arm})}
\label{sec:KL-bg}

The proof of Lemma~\ref{lm:best-arm} relies on the concept of KL-divergence. Let us provide some background to make on KL-divergence to make this proof self-contained. We use a somewhat non-standard notation that is tailored to the needs of our analysis.

The \emph{KL-divergence} (a.k.a. \emph{relative entropy}) is defined as follows. Consider two distributions $\mu,\nu$ on the same finite universe $\Omega$.%
\footnote{We use $\mu,\nu$ to denote distributions throughout this section, whereas $\mu$ denotes the latent structure elsewhere in the paper.}
Assume $\mu \ll \nu$ (in words, $\mu$ is \emph{absolutely continuous} with respect to $\nu$), meaning that $\nu(w)=0 \Rightarrow \mu(w)=0$ for all $w\in \Omega$.
Then KL-divergence of $\mu$ given $\nu$ is
$$ \KL{\mu}{\nu}
    \triangleq \E_{w\sim (\Omega,\,\mu)} \log\left(\frac{\mu(w)}{\nu(w)}\right)
    = \sum_{w\in\Omega} \log\left(\frac{\mu(w)}{\nu(w)}\right) \mu(w).
$$
In this formula we adopt a convention that $\tfrac{0}{0}=1$. We will use the fact that
\begin{align}\label{eq:KL-norm}
    \KL{\mu}{\nu}\geq \tfrac12\; \|\mu-\nu\|_1^2.
\end{align}

Henceforth, let $\mu,\nu$ be distributions on the universe $\Omega^\infty$, where $\Omega$ is a finite set. For
    $\vec{w}=(w_1, w_2,\; \ldots) \in \Omega^\infty$ and $t\in \N$,
let us use the notation
    $\Wt = (w_1 \LDOTS w_t) \in \Omega^t$.
Let $\mu_t$ be a restriction of $\mu$ to $\Omega^t$: that is, a distribution on $\Omega^t$ given by
$$ \mu_t(\Wt) \triangleq \mu\left( \{ \vec{u}\in \Omega^\infty:\; \vec{u}_t = \Wt \} \right).
$$
The \emph{next-round conditional distribution} of $\mu$ given $\Wt$, $t<T$ is defined by
$$ \mu\left( w_{t+1} \,|\, \Wt \right)
    \triangleq \frac{ \mu_{t+1}(\vec{w}_{t+1}) }{  \mu_t(\Wt)  }.
$$
Note that $\mu( \cdot  \,|\Wt)$ is a distribution on $\Omega$ for every fixed $\Wt$.

The \emph{conditional KL-divergence} at round $t+1$ is defined as
$$ \KL[t+1]{\mu}{\nu}
    \triangleq \E_{\Wt\sim (\Omega^t,\;\mu_t)}
        \KL{ \mu(\cdot\,|\Wt) }{ \nu(\cdot\,|\Wt) }.
$$
In words, this is the KL-divergence between the next-round conditional distributions
    $\mu(\cdot\,|\Wt)$ and $\nu(\cdot\,|\Wt)$,
in expectation over the random choice of $\Wt$ according to distribution $\mu_t$.

We will use the following fact, known as the \emph{chain rule} for KL-divergence:
\begin{align}\label{eq:KL-chain-rule}
\KL{\mu_T}{\nu_T} = \sum_{t=1}^T \KL[t]{\mu}{\nu},
    \quad \text{for each $T\in\N$}.
\end{align}
Here for notational convenience we define
    $\KL[1]{\mu}{\nu} \triangleq \KL{\mu_1}{\nu_1}$.

\subsection{The KL-divergence argument: proof of Lemma~\ref{lm:best-arm}}
\label{sec:KL-div}

Fix some \BwK algorithm \ALG and fix parameters $(p,\eps)$. Let \adv{x} be the problem instance in $\FamInf$ in which the best arm is $x$. For the analysis, we also consider an instance \adv{0} which coincides with \adv{x} but has no best arm: that is, all arms have expected resource consumption $p$. Let $\tau(\adv{})$ be the stopping time of \ALG for a given problem instance \adv{}, and let $N_x(\adv{})$ be the expected number of times a given arm $x$ is chosen by \ALG on this problem instance.

Consider problem instance \adv{0}. Since all arms are the same, we can apply Corollary~\ref{cor:LB-stopping}(a) (suitably modified to the non-best arm) and obtain
    $\E[\tau(\adv{0})] = \flr{B+1}/p$.
We focus on an arm $x$ with the smallest $N_x(\adv{0})$. For this arm it holds that
\begin{align}\label{eq:KL-div-x}
 N_x(\adv{0}) \textstyle
    \leq \tfrac{1}{\narms} \sum_{x\in X} \; N_x(\adv{0})
    = \tfrac{1}{\narms} \E[\tau(\adv{0})]
    \leq \tfrac{\flr{B+1}}{p\,\narms}.
\end{align}
In what follows, we use this inequality to upper-bound $N_x(\adv{x})$. Informally, if arm $x$ is not played sufficiently often in \adv{0}, \ALG cannot tell apart \adv{0} and \adv{x}.

\newcommand{\Diff}{\mathtt{diff}}

The \emph{transcript} of \ALG on a given problem instance \adv{} is a sequence of pairs
    $\{(x_t,c_t)\}_{t\in \N}$,
where for each round $t\leq \tau(\adv{})$ it holds that $x_t$ is the arm chosen by \ALG and $c_t$ is the realized resource consumption in that round. For all $t> \tau(\adv{})$, we define
    $(x_t,c_t) = (\Null,0) $.
To map this to the setup in Section~\ref{sec:KL-bg}, denote
    $\Omega = (X\cup \{\Null\}) \times \{0,1\}$.
Then the set of all possible transcripts is a subset of $\Omega^{\infty}$.

Every given problem instance \adv{} induces a distribution over $\Omega^{\infty}$. Let $\mu,\nu$ be the distributions over $\Omega^{\infty}$ that are induced by \adv{0} and \adv{x}, respectively. We will use the following shorthand:
$$ \Diff[T_0,T_*]
    \triangleq \sum_{t=T_0}^{T_*} \nu(x_t=x)-\mu(x_t=x),
        \quad \text{ where } 1\leq T_0\leq T_* \leq \infty.
$$
For any $T\in \N$ (which we will fix later), we can write
\begin{align}\label{eq:KL-div-diff}
N_x(\adv{x})-N_x(\adv{0})
    =\Diff[1,\infty]
    = \Diff[1,T] + \Diff[T+1,\infty].
\end{align}
We will bound $\Diff[1,T]$ and $\Diff[T+1,\infty]$ separately.

\xhdr{Upper bound on $\Diff[1,T]$.}
This is where we use KL-divergence. Namely, by \eqref{eq:KL-norm} we have
\begin{align}\label{eq:KL-div-Diff-1-T}
\Diff[1,T]
    \leq \tfrac{T}{2} \; \|\mu_T-\nu_T\|_1
    \leq T \sqrt{\tfrac12 \; \KL{\mu_T}{\nu_T}}.
\end{align}
Now, by the chain rule (\eqref{eq:KL-chain-rule}), we can focus on upper-bounding the conditional KL-divergence
    $\KL[t]{\mu}{\nu}$
at each round $t\leq T$.

\begin{claim}\label{cl:KL-div-cond}
For each round $t\leq T$ it holds that
\begin{align}\label{eq:cl:KL-div-cond}
\KL[t]{\mu}{\nu} = \mu( x_t = x) \;
   \left(
        p\,\log (\tfrac{p}{p-\eps}) + (1-p)\,\log (\tfrac{1-p}{1-p+\eps})
   \right).
\end{align}
\end{claim}

\begin{proof}
The main difficulty here is to carefully ``unwrap" the definition of $\KL[t]{\mu}{\nu}$.

Fix $t\leq T$ and let $\Wt \in \Omega^t$ be the partial transcript up to and including round $t$. For each arm $y$, let
    $f(y| \Wt )$
be the probability that \ALG chooses arm $y$ in round $t$, given the partial transcript  $\Wt$. Let $c(y|\adv{})$ be the expected resource consumption for arm $y$ under a problem instance \adv{}. The transcript for round $t+1$ is a pair $w_{t+1} = (x_{t+1}, c_{t+1})$, where $x_{t+1}$ is the arm chosen by \ALG in round $t+1$, and $c_{t+1}\in \{0,1\}$ is the resource consumption in that round. Therefore if $c_{t+1}=1$ then
\begin{align*}
\mu(w_{t+1}\,|\Wt)
    &= f(x_{t+1}| \Wt )\;  c(x_{t+1}|\adv{0})
    = f(x_{t+1}| \Wt )\; p, \\
\nu(w_{t+1}\,|\Wt)
    &= f(x_{t+1}| \Wt )\;  c(x_{t+1}|\adv{x})
    = f(x_{t+1}| \Wt )\; \left( p-\eps\,\Id{x_{t+1}=x}\right).
\end{align*}
Similarly, if $c_{t+1}=0$ then
\begin{align*}
\mu(w_{t+1}\,|\Wt)
    &= f(x_{t+1}| \Wt )\;  (1-c(x_{t+1}|\adv{0}))
    = f(x_{t+1}| \Wt )\; (1-p), \\
\nu(w_{t+1}\,|\Wt)
    &= f(x_{t+1}| \Wt )\;  (1-c(x_{t+1}|\adv{x}))
    = f(x_{t+1}| \Wt )\; \left( 1-p+\eps\,\Id{x_{t+1}=x}\right).
\end{align*}

It follows that
\begin{align*}
\log \frac{ \mu(w_{t+1}\,|\Wt) }{ \nu(w_{t+1}\,|\Wt) }
    &= \Id{x_t=x}\; \left(
        \log (\tfrac{p}{p-\eps})\; \Id{c_{t+1}=1}
        +\log (\tfrac{1-p}{1-p+\eps})\; \Id{c_{t+1}=0}
    \right).
\end{align*}
Taking expectations over
    $w_{t+1}=(x_t,c_t)\sim \mu(\cdot\,|\Wt) $,
we obtain
\begin{align*}
\KL{ \mu(\cdot\,|\Wt) }{ \nu(\cdot\,|\Wt) }
   = f(x|\Wt)\;
    \left(
        p\,\log (\tfrac{p}{p-\eps}) + (1-p)\,\log (\tfrac{1-p}{1-p+\eps})
   \right).
\end{align*}
Taking expectations over
    $\Wt \sim \mu_t$,
we obtain the conditional KL-divergence $\KL[t]{\mu}{\nu}$. \eqref{eq:cl:KL-div-cond} follows because
$$ \E_{\Wt\sim \mu_t}\; f(x|\Wt) = \mu(x_t=x). \qquad\qedhere
 $$
\end{proof}

We will use the following fact about logarithms, which is proved using standard quadratic approximations for the logarithm. The proof is in Section~\ref{app:facts}.

\begin{fact}\label{fact:logs}
Assume
    $\tfrac{\eps}{p}\leq \tfrac12 $ and $p\leq \tfrac12$.
Then
$$ p\,\log (\tfrac{p}{p-\eps}) + (1-p)\,\log (\tfrac{1-p}{1-p+\eps})
    \leq \tfrac{2 \eps^2}{p}.
$$
\end{fact}

Now we can put everything together and derive an upper bound on $\Diff[1,T]$.

\begin{claim}
Assume
    $\tfrac{\eps}{p}\leq \tfrac12 $ and $p\leq \tfrac12$.
Then
    $\Diff[1,T] \leq T\, \frac{\eps}{p}\, \sqrt{\tfrac{B+1}{m}}$.
\end{claim}
\begin{proof}
By Claim~\ref{cl:KL-div-cond} and Fact~\ref{fact:logs}, for each round $t\leq T$ we have
$$\KL[t]{\mu}{\nu} \leq \frac{2 \eps^2}{p}\; \mu( x_t = x).$$
By the chain rule (\eqref{eq:KL-chain-rule}), we have
\begin{align*}
\KL{\mu_T}{\nu_T}
    &\leq \frac{2\eps^2}{p}\; \sum_{t=1}^T\, \mu(x_t=x)
    \leq \frac{2\eps^2}{p} \; N_x(\adv{0})
    \leq 2\; \frac{B+1}{m}\; \left(\frac{\eps}{p}\right)^2.
\end{align*}
The last inequality is the place where we use our choice of $x$, as expressed by
\eqref{eq:KL-div-x}.

Plugging this back into \eqref{eq:KL-div-Diff-1-T}, we obtain
    $\Diff[1,T] \leq T\, \frac{\eps}{p}\, \sqrt{\tfrac{B+1}{m}}$.
\end{proof}

\xhdr{Upper bound on $\Diff[T,\infty]$.}
Consider the problem instance \adv{x}, and consider the policy that always chooses the best arm. Let $\nu^*$ be the corresponding distribution over transcripts $\Omega^\infty$, and let $\tau$ be the corresponding stopping time. Note that
    $\nu^*(x_t=x)$ if and only if $\tau>t$.
Therefore:
\begin{align*}
\Diff[T,\infty]
    \leq  \sum_{t=T}^\infty \nu(x_t=x)
    \leq  \sum_{t=T}^\infty \nu^*(x_t=x)
    = \sum_{t=T}^\infty \nu^*(\tau>t)
    \leq \OPT^2/T.
\end{align*}
The second inequality can be proved using a simple ``coupling argument''. The last inequality follows from Fact~\ref{fact:stopping}, observing that $\E[\tau]=\OPT$.

\xhdr{Putting the pieces together.}
Assume
    $p\leq \tfrac12$ and $\tfrac{\eps}{p}\leq \tfrac12$.
Denote
    $\gamma = \tfrac{\eps}{p} \sqrt{\tfrac{B+1}{m}}$.
Using the upper bounds on
    $\Diff[1,T]$ and $\Diff[T+1,\infty]$
and plugging them into \eqref{eq:KL-div-diff}, we obtain
$$ N_x(\adv{x}) - N_x(\adv{0})
    \leq \gamma T + \OPT^2/T
    \leq \OPT \sqrt{\gamma}
$$
for $T = \OPT /\sqrt{\gamma}$. Recall that
    $N_x(\adv{0}) < \OPT/m$.
Thus, we obtain
$$ N_x(\adv{x}) \leq (\tfrac{1}{m} + \sqrt{\gamma})\, \OPT.
$$
Recall that we need to conclude that $N_x(\adv{x})\leq \tfrac34 \OPT$. For that, it suffices to have $\gamma\leq \tfrac{1}{16}$.

\section{\BwK with \preDiscr}
\label{sec:discretization}

\newcommand{\DError}{\mathtt{Err}}  
\newcommand{\Cov}{\mathtt{cov}}     

\asedit{In this section we develop a general technique for \preDiscr, and apply it to dynamic pricing with a single product and dynamic procurement with a single budget. For both applications, our regret bounds significantly improve over prior work. While the dynamic pricing application is fairly straightforward given the general result, the dynamic procurement application takes some work and uses a non-standard mesh of prices. We also obtain an initial result for dynamic pricing with multiple products. The main technical challenge is to upper-bound the discretization error; we can accomplish this whenever the expected resource to expected consumption ratio of each arm can be expressed in a particularly simple way.}

\subsection{\PreDiscr as a general technique}

\asedit{The high-level idea behind \preDiscr is to apply an existing \BwK algorithm with a restricted, finite action space $S\subset X$ that is chosen in advance. Typically $S$ is, in some sense, ``uniformly spaced" in $X$, and its ``granularity" is tuned in advance so as to minimize regret.}

Consider a problem instance with action space restricted to $S$. Let $\Rew(S)$ be the algorithm's reward on this problem instance, and let $\LPOPT(S)$ be the corresponding value of $\LPOPT$, as defined in Section~\ref{sec:lp}. $\OPT(X)$ and $\LPOPT(X)$ will refer to the corresponding quantities for the original action space $X$. The key two quantities in our analysis of \preDiscr are
\begin{align}
R(S)        &= \LPOPT(S) - \Rew(S) &\text{(\emph{$S$-regret})} \nonumber\\
\DError(S|X)&= \LPOPT(X)-\LPOPT(S) &\text{(\emph{discretization error of $S$ relative to $X$})}. \label{eq:preadjusted-DE}
\end{align}
Note that algorithm's regret can be expressed as
\begin{align*}
\OPT(X)-\Rew(S) \leq \LPOPT(X) - \Rew(S) = R(S) + \DError(S|X).
\end{align*}

Now, suppose $S$ is parameterized by $\eps>0$ which controls its ``granularity". Adjusting the $\eps$ involves balancing $R(S)$ and $\DError(S|X)$: indeed, decreasing $\eps$ tends to increase $R(S)$ but decrease $\DError(S|X)$. We upper-bound the $S$-regret via our main algorithmic result;%
\footnote{We need to use the regret bound in terms of the best known upper bound on $\OPT$, rather than $\OPT$ itself, because the latter is not known to the algorithm. For example, for dynamic pricing one can use $\OPT\leq \budg$.}
the challenge is to upper-bound $\DError(S|X)$.

A typical scenario where one would want to apply \preDiscr is when an algorithm chooses among prices. More formally, each arm includes a real-valued vector of prices in $[0,1]$ (and perhaps other things, such as the maximal number of items for sale). The restricted action set $S$ consists of all arms such that all prices belong to a suitably chosen mesh $M\subset [0,1]$ with granularity $\eps$. There are several types of meshes one could consider, depending on the particular BwK domain. The most natural ones are the \emph{\eps-additive mesh}, with prices that are integer multiples of $\eps$, and \emph{\eps-multiplicative mesh mesh}, with prices of the form $(1-\eps)^{\ell}$, $\ell\in\N$. Both have been used in the prior work on MAB in metric spaces \citep{Bobby-nips04,Hazan-colt07,LipschitzMAB-stoc08,Pal-Bandits-aistats10}) and dynamic pricing
(e.g., \citep{Bobby-focs03,Blum03,BZ09,DynPricing-ec12}). Somewhat surprisingly, for dynamic procurement we find it optimal to use a very different mesh, called \emph{\eps-hyperbolic mesh}, in which the prices are of the form $\tfrac{1}{1+\eps\ell}$, $\ell\in\N$.

\asedit{While in practice the action set $X$ is usually finite (although possibly very large), it is mathematically more elegant to consider infinite $X$. For example, we prefer to allow arbitrary fractional prices, even though in practice they may have to be rounded to whole cents. However, recall that $\LPOPT$ in Section~\ref{sec:lp} is only defined for a finite action space $X$. To handle infinite $X$, we define}
\begin{align}\label{eq:LPOPT-infinite}
\LPOPT(X)   = \sup_{\text{finite $X'\subset X$}} \LPOPT(X').
\end{align}

In line with Lemma~\ref{lem:lp-relax}, let us argue that
    $\LPOPT(X)\geq \OPT(X)$
even when $X$ is infinite. Specifically, we prove this for all versions of dynamic pricing and dynamic procurement, and more generally for any \BwK domain  such that for each arm there are only finitely many possible outcome vectors.

\begin{lemma}
Consider a \BwK domain with infinite action space $X$, such that for each arm there are only finitely many possible outcome vectors. Then
    $\LPOPT(X) \geq \OPT(X)$.
\end{lemma}

\begin{proof}
Fix a problem instance, and consider an \OptPolicy for this instance. W.l.o.g. this policy is deterministic.%
\footnote{A randomized policy can be seen as a distribution over deterministic policies, so one of these deterministic policies must have same or better expected total reward.} For each round, this policy defines a deterministic mapping from histories to arms to be played in this round. Since there are only finitely many possible histories, the policy can only use a finite subset of arms, call it $X'\subset X$. By Lemma~\ref{lem:lp-relax}, we have
\begin{align*}
    \LPOPT(X)\geq \LPOPT(X') \geq \OPT(X') = \OPT(X). \qquad \qedhere
\end{align*}
\end{proof}

\OMIT{ 
We note in passing that
$    \LPOPT(X) \geq \sup_{\text{finite $X'\subset X$}} \OPT(X')$,
so one could also use the $\sup$ on the right-hand side as a very reasonable alternative benchmark, instead of $\OPT(X)$.
} 

\subsection{A general bound on discretization error}
\label{sec:eps-discretization}

We develop a general bound on discretization error $\DError(S|X)$, as defined in \eqref{eq:preadjusted-DE}. To this end, we consider the expected reward to expected consumption ratios of arms (and the differences between them), whereas in the work on MAB in metric spaces it suffices to consider the difference in expected rewards.

To simplify notation, we suppress $\mu$, the (actual) latent structure: e.g., we will write $c_i(\D) = c_i(\D,\mu)$, $r(\D) = r(\D,\mu)$, and $\LP(\D,\mu) = \LP(\D)$
for distributions $\D$ and resources $i$.

\begin{definition}\label{def:eps-covers}
We say that arm $x$ \emph{\eps-covers} arm $y$ if the following two properties are satisfied for each resource $i$ such that $c_i(x)+c_i(y)>0$:
\begin{itemize}
\item[(i)] $r(x)/c_i(x)\geq r(y)/c_i(y)-\eps$.
\item[(ii)] $c_i(x)\geq c_i(y)$.
\end{itemize}
A subset $S\subset X$ of arms is called an \emph{\eps-discretization} of $X$ if each arm in $X$ is \eps-covered by some arm in $S$.
\end{definition}

\begin{theorem}[\preDiscr]\label{thm:discretization}
Fix a BwK domain with action space $X$. Let $S\subset X$ be an \eps-discretization of $X$, for some $\eps\geq 0$. Then the discretization error $\DError(S|X)$ is at most $\eps dB$. Consequently, for any algorithm with $S$-regret $R(S)$ we have
$\LPOPT(X) - \Rew(S) = R(S) + \eps dB$.
\end{theorem}

\begin{proof}
We need to prove that $\DError(S|X)\leq \eps dB$. If $X$ is infinite, then (by \eqref{eq:LPOPT-infinite}) it suffices to prove $\DError(S|X')\leq \eps dB$ for any finite subset of $X'\subset X$. Let $\D$ be the distribution over arms in $X'$ which maximizes $\LP(\D,\mu)$.  We use $\D$ to construct a distribution $\D_S$ over $S$ which is nearly as good.

We define $\D_S$ as follows. Since $S$ is an \eps-discretization of $X$, there exists a family of subsets $(\Cov(x)\subset X:\, x\in S)$ so that each arm $x\in S$ \eps-covers all arms in $\Cov(x)$, the subsets are disjoint, and their union is $X$. Fix one such family of subsets, and define
\begin{align*}
\D_S(x)&=\sum_{y\in \Cov(x)} \D(y)\, \min_{i:\, c_i(x)>0}\;\frac{c_i(y)}{c_i(x)}, \quad x\in S.
\end{align*}
Note that $\sum_{x\in S} \D_S(S)\leq 1$ by Definition~\ref{def:eps-covers}(ii). With the remaining probability, the null arm is chosen (i.e., the algorithm skips a given round).

To argue that $\LP(\D_S,\mu)$ is large, we upper-bound the resource consumption $c_i(\D_S)$, for each resource $i$, and lower-bound the reward $r(\D_S)$.
\begin{align}
c_i(\D_S)&=\textstyle \sum_{x\in S}\, c_i(x)\,\D_S(x) \nonumber \\
&\leq\sum_{x\in S}c_i(x)\quad \sum_{y\in \Cov(x):\; c_i(x)>0} \quad \D(y) \frac{c_i(y)}{c_i(x)}  \nonumber \\
&=\sum_{x\in S} \quad \sum_{y\in \Cov(x):\; c_i(x)>0} \quad\D(y)\, c_i(y) \nonumber \\
&= \sum_{y\in X} \D(y)\, c_i(y) \nonumber \\
&=c_i(\D) \label{eq:lm:BwK-discretization-c}
\end{align}
(Note that the above argument did not use the property (i) in Definition~\ref{def:eps-covers}.)

In what follows, for each arm $x$ define
    $I_x = \{i:\, c_i(x)>0\}$.
\begin{align}
r(\D_S)&= \textstyle \sum_{x\in S}\,r(x)\,\D_S(x)  &\nonumber \\
&=\sum_{x\in S} r(x)\sum_{y\in \Cov(x)}\D(y)\, \min_{i\in I_x}\;\frac{c_i(y)}{c_i(x)} & \nonumber \\
&=\sum_{x\in S}\; \sum_{y\in \Cov(x)}\D(y)\, \min_{i\in I_x}\;\frac{c_i(y)\,r(x)}{c_i(x)} & \nonumber \\
&\geq\sum_{x\in S}\; \sum_{y\in \Cov(x)}\D(y)\, \min_{i\in I_x}\; r(y)-\eps c_i(y)
    &\text{(by Definition~\ref{def:eps-covers}(i))}\nonumber \\
&=  \sum_{y\in X}\D(y)\, \min_i\; r(y)-\eps c_i(y) &\nonumber \\
&\geq \textstyle  \sum_{y\in X}\D(y)\,\left( r(y)-\eps\, \sum_i c_i(y) \right)&\nonumber \\
&= \textstyle r(\D)-\eps\sum_i c_i(\D).
\label{eq:lm:BwK-discretization-r}
\end{align}

Let
    $\tau(\D) = \min_i \tfrac{B}{c_i(\D)} $
be the stopping time in the linear relaxation, so that
    $\LP(\D) = \tau(\D)\,r(\D)$.
By \eqref{eq:lm:BwK-discretization-c} we have
    $\tau(\D_S) \geq \tau(\D)$.
We are ready for the final computation:
\begin{align*}
\LP(\D_S)
    &= \tau(\D_S)\; r(\D_S)  &\\
    &\geq \tau(\D)\; r(\D_S) &\\
    &\geq \textstyle \tau(\D)  \left(r(\D)-\eps\sum_i c_i(\D)\right)
        & \text{(by \eqref{eq:lm:BwK-discretization-r})}   \\
    &\geq \textstyle r(\D)\,\tau(\D)-\eps\, \tau(\D) \sum_i c_i(\D) \\
    &\geq \LP(\D) - \eps\, d\, B.\qedhere
\end{align*}
\end{proof}

\subsection{\PreDiscr for dynamic pricing}
\label{sec:discretization-pricing}

\asedit{We apply the machinery developed above to handle the basic version of dynamic pricing, as defined in Section~\ref{sec:problem-description}. In fact, our technique easily generalizes to multiple products, in a particular scenario which we call \emph{dynamic bundle-pricing}. We present the more general result directly. }

The dynamic bundle-pricing problem is defined as follows. There are $d$ products, with limited supply of each, and $T$ rounds. In each round, a new buyer arrives, an algorithm chooses a bundle of products and a price, and offers this bundle for this price. The offer is either accepted or rejected. The bundle is a vector $(b_1 \LDOTS b_d)$, so that $b_i\in \N$ units of each product $i$ are offered. We assume that the bundle must belong to a fixed collection $\mF$ of allowed bundles. Buyers' valuations over bundles can be arbitrary (in particular, not necessarily additive); they are drawn independently from a fixed distribution over valuations. For normalization, we assume that each buyer's valuation for any bundle of $\ell$ units lies in the interval $[0,\ell]$; accordingly, the offered price for such bundle can w.l.o.g. be restricted to the same interval.

\begin{theorem}\label{thm:DynPricing-discretization}
Consider the dynamic bundle-pricing problem such that there are $d$ products, each with supply $B$. Assume each allowed bundle consists of at most $\ell$ items, and prices are in $[0,\ell]$. Algorithm $\pdbwk$ with an \eps-additive mesh, for some $\eps=\eps(B,|\mF|,\ell)$, has regret
    $\Otilde(d\,B^{2/3} \, (|\mF|\ell)^{1/3})$.
\end{theorem}

\asedit{The basic version from Section~\ref{sec:problem-description} is a special case with a single product and a single allowed bundle which consists of one unit of this product. Taking
    $d=\ell=|\mF| = 1$
in Theorem~\ref{thm:DynPricing-discretization}, we obtain regret
    $\Otilde(B^{2/3})$.
This regret bound is optimal for any pair $(B,T)$, as proved in \citet{DynPricing-ec12}.}

\ascomment{Added explicit corollary statement.}


\begin{corollary}\label{cor:DynPricing-discretization}
Consider the dynamic pricing problem, as defined Section~\ref{sec:problem-description}. Algorithm $\pdbwk$ with an \eps-additive mesh, for a suitably chosen $\eps=\eps(B)$, has regret
    $\Otilde(B^{2/3})$.
\end{corollary}

\begin{proof}[Proof of Theorem~\ref{thm:DynPricing-discretization}]
First, let us cast this problem as a \BwK domain. To ensure that per-round rewards and per-round resource consumptions lie in $[0,1]$, we scale them down by the factor of $\ell$. Accordingly, the rescaled supply constraint is $B' = B/\ell$. In what follows, consider the scaled-down problem instance.

An arm is a pair $x=(b,p)$, where $b\in \mF$ is a bundle and $p\in [0,1]$ is the offered price. Let $\SalesRate(x)$ be the probability of a sale for this arm, divided by $\ell$; this probability is non-increasing in $p$ for a fixed bundle $b$. Then expected per-round reward is $r(x)=p\,\SalesRate(x)$, and expected per-round consumption of product $i$ is $c_i(x) = b_i\,\SalesRate(x)$. Therefore,
\begin{align}\label{eq:DynPricing-ratio}
\frac{r(x)}{c_i(x)} = \frac{p}{b_i}, \quad \text{for each arm $x=(b,p)$ and product $i$}.
\end{align}
This is a crucial domain-specific property that enables \preDiscr. It follows that for any arm $x=(b,p)$, this arm \eps-covers any arm $x'=(b,p')$ such that
    $p'-\eps \leq p\leq p'$.
Therefore an \eps-additive mesh $S$ is an \eps-discretization, for any $\eps>0$.

Consider algorithm $\pdbwk$ with action space $S$. Using Theorem~\ref{thm:pdbwk} and observing that $\LPOPT\leq dB'$, we obtain $S$-regret
    $$R(S) \leq \Otilde(d \sqrt{B'\,|S|}) = \Otilde(d\sqrt{B'\, |\mF|/\eps}). $$
By Theorem~\ref{thm:discretization} discretization error is $\DError(S|X)\leq \eps d B'$.
So, regret relative to $\LPOPT$ is
\begin{align*}
R(S) + \DError(S|X) \leq \Otilde(d\sqrt{B'\, |\mF|/\eps} + \eps dB')
    \leq \Otilde(d)\,(B/\ell)^{2/3} \, |\mF|^{1/3}
\end{align*}
for a suitably chosen $\eps = (B/\ell)^{-1/3}\; |\mF|^{1/3} $.
Recall that this is regret for the rescaled problem instance. For the original problem instance, rewards are scaled up by the factor of $\ell$, so regret is scaled up by $\ell$, too.
\end{proof}

One can easily extend Theorem~\ref{thm:DynPricing-discretization} to a setting where in each round an algorithm offers several copies of the same bundle for the same per-bundle price, and an agent can choose how many copies to buy (if any). More precisely, in each round an algorithm chooses two things: a bundle from $\mF$ and the number of copies of this bundle. The latter is restricted to be at most $\Lambda$, where $\Lambda$ is a known parameter. We call this setting \emph{dynamic bundle-pricing with multiplicity $\Lambda$}.
The algorithm and analysis is essentially the same.

\ascomment{Added explicit theorem statement.}

\begin{theorem}\label{thm:DynPricing-discretization-multiplicity}
Consider dynamic bundle-pricing with multiplicity $\Lambda$. Assume that each product has supply $B$, and each allowed bundle consists of at most $\ell$ items. Algorithm $\pdbwk$ with an \eps-additive mesh, for a suitably chosen $\eps=\eps(B,|\mF|,\ell,\Lambda)$, has regret
    $\Otilde(d\,(B\Lambda)^{2/3} \, (|\mF|\ell)^{1/3})$.
\end{theorem}

\OMIT{In particular, we recover \emph{dynamic pricing with non-unit demands}, as defined in Section~\ref{sec:apps-DynPricing}, as a special case with a single product and a single allowed bundle of unit size.}

\subsection{\PreDiscr for dynamic procurement}
\label{sec:discretization-procurement}

Application to dynamic procurement takes a little more work \asedit{and results in a weaker regret bound, compared to the application to dynamic pricing. The main reason is that the natural mesh for dynamic procurement is \eps-hyperbolic (rather than \eps-additive). One needs to bound this mesh from below to make it finite, which increases the mesh size and the discretization error.}

While our main goal here is to handle the basic version of dynamic procurement, as defined in Section~\ref{sec:problem-description}, the same technique easily extends to a generalization where the algorithm can buy multiple items in each round. The generalization is defined as follows. In each round $t$, the algorithm offers to buy up to $\Lambda$ units at price $p_t$ per unit, where $p_t\in [0,1]$ is chosen by the algorithm. The outcome is summarized by the number $k_t$ of items bought, where $k_t$ is an independent sample from some fixed (but unknown) distribution parameterized by $p_t$ and $\Lambda$. The algorithm is constrained by the time horizon $T$, budget $B$, and per-round supply constraint $\Lambda$. We prove the following:

\begin{theorem}\label{thm:DynProcurement-discretization}
Consider dynamic procurement with up to $\Lambda$ items bought per round. Algorithm $\pdbwk$ with a suitably chosen action space $S$ yields regret
   $\tilde{O}(\Lambda^{5/4} T/B^{1/4})$.
Specifically, $S = [p_0,1] \cap M$, where $M$ is the \eps-hyperbolic mesh, for some parameters
    $\eps,p_0\in (0,1)$
that depend only on $B$, $T$ and $\Lambda$.
\end{theorem}

Let us model this problem as a \BwK domain. The action space is $X = [0,1]$: the arms correspond to all possible prices. (The zero price corresponds to the ``null arm".) To ensure that rewards and consumptions lie in $[0,1]$, we scale them down by a factor of $\Lambda$, as in Section~\ref{sec:discretization-pricing}, so that reward in any round $t$ is $k_t/\Lambda$, and budget consumption is $p_t\,k_t/\Lambda$. Accordingly, budget is rescaled to $B' = B/\Lambda$. Henceforth, consider the scaled-down problem instance, unless specified otherwise.

Let $\SalesRate(p)$ be the expected per-round number of items sold for a given price $p$, divided by $\Lambda$; note that it is non-decreasing in $p$. Then expected budget consumption is $c(p)=p\, \SalesRate(p)$, and expected reward is simply
    $r(p) = \SalesRate(p)$.
It follows that
\begin{align*}
\frac{r(p)}{c(p)} = \frac{1}{p},
    \quad \text{for each price $p$}.
\end{align*}
Like \eqref{eq:DynPricing-ratio}, this is a crucial domain-specific property that enables \preDiscr.

By Definition~\ref{def:eps-covers} price $p$ \eps-covers price arm $q$ if and only if
    $q<p$ and $\tfrac{1}{p}\geq \tfrac{1}{q}-\eps$.
This makes the hyperbolic mesh a natural mesh for this problem, rather than additive or multiplicative ones. It is easy to see that the \eps-hyperbolic mesh $S$ on $X$ is an \eps-discretization of $X$: namely, each price $q$ is \eps-covered by the smallest price $p\geq q$ that lies in $S$.

Unfortunately, this mesh has infinitely many points. In fact, it is easy to see that any \eps-discretization on $X$ must be infinite, even for $\Lambda=1$. To obtain a finite \eps-discretization,
we only consider prices $p\geq p_0$, for some parameter $p_0$ to be tuned later. Below we argue that this restriction is not too damaging:

\begin{claim}\label{lm:non-unit-demand-truncate}
Consider dynamic procurement with non-unit supply. Then for any $p_0\in (0,1)$ it holds that
$$\LPOPT([p_0,1])\geq \LPOPT([0,1]) - p_0\, T^2/B' $$
\end{claim}
\begin{proof}
When $p_0>B'/T$ the bound is trivial and for the rest of the proof we assume that $p_0\leq B'/T$.

By \eqref{eq:LPOPT-infinite} it suffices to replace $\LPOPT([0,1])$ in the claim with $\LPOPT(X_0)$, for any given finite subset $X_0\subset [0,1]$. Let $\D$ be an LP-perfect distribution for the problem instance restricted to $X_0$; such $\D$ exists by Claim~\ref{cl:LP-properties}. Thus,
    $\LP(\D)=\LPOPT(X)$ and $c(\D)\leq \tfrac{B'}{T}$.
Furthermore, $\D$ has a support of size at most $2$; denote it as arms $p_1,p_2 \in[0,1]$, $p_1\leq p_2$, where the null arm would correspond to $p_1=0$. If $p_1\geq p_0$ then $\D$ has support in the interval $[p_0,1]$, and we are done; so from here on we assume $p_1<p_0$. Note that
\begin{align*}
    \LP(\D) = r(\D)\,\min \left( \tfrac{B'}{c(\D)},T \right) = T\, r(\D).
\end{align*}

To prove the desired lower bound on $\LPOPT([p_0,1])$, we construct a distribution $\D'$ with support in $\{0\}\cup [p_0,1]$ and a sufficiently large LP-value. (Here the zero price corresponds to the null arm.)

Suppose $p_2\leq \tfrac{B'}{T}$. Define $\D'$ by putting probability mass on price $\tfrac{B'}{T}$. Since $c(\tfrac{B'}{T}) \leq \tfrac{B'}{T}$, we have
$$ \LP(\D')
    = T\, r(\D')
    = T\, \SalesRate(\tfrac{B'}{T})
    \geq T\,\SalesRate(p_2)
    \geq T\, r(\D)
    = \LP(\D),
$$
and we are done. From here on, assume $p_2> \tfrac{B'}{T}$.

Now consider the main case:
        $p_1 \leq p_0 \leq \tfrac{B'}{T} < p_2$.
Define distribution $\D'$ as follows:
\begin{align*}
\D'(p_0) &=\D(p_1) \\
\D'(p_2) &=\max\left( 0,\D(p_2)-p_0/p_2 \right)\\
\D'(0) &=1-\D'(p_0)-\D'(p_2).
\end{align*}

We claim that $c(\D') \leq \tfrac{B'}{T}$. If $\D'(p_2)=0$ then
    $c(\D') = c(p_0) \leq p_0 \leq \tfrac{B'}{T}$.
If $\D'(p_2)>0$ then $\D'(p_2) = \D(p_2)-p_0/p_2$, and therefore,
\begin{align*}
c(\D')-c(\D)
    &= \D(p_1)\,\left( p_0 \SalesRate(p_0) - p_1 \SalesRate(p_1) \right) - p_2 \SalesRate(p_2) \tfrac{p_0}{p_2} \\
    &\leq p_0 \SalesRate(p_0) - p_0 \SalesRate(p_2) \leq 0.
\end{align*}
Then
    $c(\D') \leq c(\D) \leq \tfrac{B'}{T}$.
Claim proved.

Therefore, $\LP(\D') = T\, r(\D')$. To complete the proof:
\begin{align*}
r(\D')-r(\D)
    &\geq \D(p_1) \left(  \SalesRate(p_0) - \SalesRate(p_1) \right) - \SalesRate(p_2)\, p_0/p_2 \\
    &\geq - p_0 / p_2
     \geq - p_0 T/B' . \\
\LPOPT([p_0,1]) - \LPOPT(X_0)
    &= \LP(\D') - \LP(\D) \\
    & = T (r(\D')-r(\D)) \leq - p_0 T^2/B'. \qquad \qedhere
\end{align*}
\end{proof}

Suppose algorithm $\pdbwk$ is applied to a problem instance with a finite action space $S$. Then by Theorem~\ref{thm:pdbwk} the $S$-regret is
\begin{align*}
    R(S) = \Otilde(\sqrt{mT} + T\sqrt{m/B'}), \quad m=|S|.
\end{align*}

Let $S = [p_0,1] \cap M$, where $M$ is the \eps-hyperbolic mesh, for some $\eps,p_0\in (0,1)$. Then $m=|S| \leq \tfrac{1}{\eps p_0} $. Moreover, $S$ is an \eps-discretization for action space $X'=[p_0,1]$, for the same reason that $M$ is an \eps-discretization for the original action space $X=[0,1]$. Therefore:
\begin{align*}
\DError(S|X') &\leq \eps B'      & \text{(by Theorem~\ref{thm:discretization})} \\
\DError(X'|X) &\leq p_0\, T^2/B' & \text{(by Claim~\ref{lm:non-unit-demand-truncate})} \\
\LPOPT(X) -\Rew(S)
    &= R(S)+\DError(S|X')+\DError(X'|X) \\
    &\leq R(S) + \eps B' + p_0\, T^2/B'.
\end{align*}
Optimizing the choice of $\eps$ and $p_0$, we obtain the final regret bound of
    $\tilde{O}(T\, (B')^{-1/4})$.
Recall that this is the regret bound for the rescaled problem instance. Going back to the original problem instance, regret is multiplied by a factor of $\Lambda$. This completes the proof of Theorem~\ref{thm:DynProcurement-discretization}.

\section{Applications and corollaries}
\label{sec:apps}

\ascomment{re-wrote the "intro" and the first two subsections.}

We systematically overview various applications of \BwK and corresponding corollaries. This section can be read independently of the technical material in the rest of the paper.

\xhdr{Some technicalities.} In applications with very large or infinite action space $X$ we apply a \BwK algorithm with a restricted, finite action space $S\subset X$, where $S$ is chosen in advance. Immediately, we obtain a bound on the \emph{$S$-regret}: regret with respect to the value of $\LPOPT$ on the restricted action space (such bound depends on $|S|$). Instantiating such regret bounds is typically straightforward once one precisely defines the setting. In some applications we can choose $S$ using \preDiscr, as discussed in Section~\ref{sec:discretization}.

In some of the applications, per-round reward and resource consumption may be larger than $1$. Then one needs to scale them down to fit the definition of \BwK and apply our regret bounds, and scale them back up to obtain regret for the original (non-rescaled) version. We encapsulate this argument as follows:

\begin{lemma}\label{lm:BwK-rescaled}
Consider a version of \BwK with finite action set $S$, in which per-round rewards are upper-bounded by $r_0$, and per-round consumption of each resource is at most $c_0$. Then one can achieve regret
\begin{align}\label{eq:regret-rescaled}
\Otilde\left( \sqrt{r_0\, |S|\, \optrwd} + \optrwd \sqrt{c_0\,|S|/\budg}
\;\;\right)
\end{align}
by applying algorithm \pdbwk with suitably rescaled rewards, resource consumption, and budgets.
\end{lemma}

\begin{proof}
Denote
    $R(\OPT,B) = \sqrt{|S|\, \OPT} + \OPT \sqrt{|S|/B}$,
as in the main regret bound.

To cast this problem as an instance of \BwK, consider a rescaled problem instance in which all rewards are divided by $r_0$, and all consumptions and budgets are divided by $c_0$. Now we can apply regret bound \eqref{eq:intro-regret} for the scaled-down problem instance; we obtain regret
    $\Otilde(R(\OPT/r_0,B/c_0))$.
Multiply this regret bound by $r_0$ to obtain a regret bound for the original problem instance.
\end{proof}


\subsection{Dynamic pricing with limited supply}
\label{sec:apps-DynPricing}

In dynamic pricing, the algorithm is a monopolist seller that interacts with $T$ agents (potential buyers) arriving one by one. In each round, a new agent arrives, the algorithm makes an offer, the agent chooses among the offered alternatives, and leaves. The offer specifies which goods are offered for sale at which prices. The agent has valuations over the offered bundles of goods, and chooses an alternative which maximizes her utility: value of the bundle minus the price. An agent is characterized by her \emph{valuation function}: function from all possible bundles of goods that can be offered to their respective valuations. For each arriving agent, the valuation function is \emph{private}: not known to the algorithm. It is assumed to be drawn from a fixed (but unknown) distribution over the possible valuation functions, called the \emph{demand distribution}. Algorithm's objective is to maximize the total revenue; there is no bonus for left-over inventory.

\xhdr{Basic version.}
In the basic version from Section~\ref{sec:problem-description}, the algorithm has $B$ identical items for sale. In each round, the algorithm chooses a price $p_t$ and offers one item for sale at this price, and an agent either buys or leaves. The agent has a fixed private value $v_t\in [0,1]$ for an item, and buys if and only if $p_t\geq v_t$. Recall from Corollary~\ref{cor:DynPricing-discretization} that we obtain regret $\Otilde(B^{2/3})$, which is optimal according to \citep{DynPricing-ec12}.

\xhdr{Extension: non-unit demands.} Agents may be interested in buying more than one unit of the product, and may have valuations that are non-linear in the number of products bought. Accordingly, let us consider an extension where an algorithm can offer each agent multiple units. More specifically: in each round $t$, the algorithm offers up to $\lambda_t$ units at a fixed price $p_t$ per unit, where the pair $(p_t,\lambda_t)$ is chosen by the algorithm, and the agent then chooses how many units to buy, if any. We restrict $\lambda_t\leq \Lambda$, where $\Lambda$ is a fixed parameter. We obtain regret
        $\Otilde(B\, \Lambda)^{2/3}$
by Theorem~\ref{thm:DynPricing-discretization-multiplicity} (considering a special case when there is a single product and a single allowed bundle with one unit of this product). One can also consider a version with $\lambda_t=\Lambda$, so that the algorithm only chooses prices; then a very similar argument gives regret
    $\Otilde(B^{2/3} \Lambda^{1/3})$.

\xhdr{Extension: multiple products.}
When multiple products are offered for sale, it often makes sense to price them jointly. Formally, the algorithm has $d$ products for sale, with $B_i$ units of each product $i$. (To simplify regret bounds, let us assume $B_i=B$.) In each round $t$, the algorithm chooses a vector of prices
        $(p_{t,1} \LDOTS p_{t,d})\in [0,1]^d$
and offers at most one unit of each product $i$ at price $p_{t,i}$. The agent then chooses the subset of products to buy. We allow arbitrary demand distributions; we do not restrict correlations between valuations of different products and/or subsets of products.

Given a finite set $S$ of allowed price vectors, such as an \eps-additive mesh for some specific $\eps>0$, we obtain $S$-regret
    $\Otilde(d\sqrt{B\,|S|})$.
This follows from Lemma~\ref{lm:BwK-rescaled}, observing that per-round rewards are at most $r_0=d$, per-round consumption of each resource is at most $c_0=1$, and the optimal value is $\OPT\leq dB$.

There may also be a fixed collection of subsets that agents are allowed to buy, \eg agents may be restricted to buying at most three items in total. This does not affect our analysis and the regret bound.

Joint pricing is \emph{not} needed in the special case when each agent can buy an arbitrary subset $I$ of products, and her valuations are additive: $v(I) = \sum_{i\in I} v(i)$. Then she buys each product $i$ if and only if the offered price for this product exceeds $v(i)$. Therefore the problem is equivalent to a collection of $d$ separate per-product problems, and one can run a separate \BwK algorithm for each product. Using Corollary~\ref{cor:DynPricing-discretization} separately for each product, one obtains regret
        $\Otilde(d\,B^{2/3})$.

\xhdr{Extension: network revenue management.} More generally, an algorithm may have $d$ products for sale which may be produced on demand from limited \emph{primitive resources}, so that each unit of each product $i$ consumes a fixed and known amount $c_{ij}\in [0,1]$ of each primitive resource $j$. This generalization is known as network revenue management problem (see \citet{BesbesZeevi-or12} and references therein). All other details are the same as above; for simplicity, let us focus on a version in which each agent buys at most one item. Given a finite set $S$ of allowed price vectors, we obtain $S$-regret given by
    \refeq{eq:regret-rescaled}
with $r_0=c_0=d$.

In particular, if all resource constraints (including the time horizon) are scaled up by factor $\gamma$, regret scales as $\sqrt{\gamma}$. This improves over the main result in \citet{BesbesZeevi-or12}, where (essentially) regret is stated in terms of $\gamma$ and scales as $\gamma^{2/3}$.

\xhdr{Extension: bundling and volume pricing.}
When selling to agents with non-unit demands, an algorithm may use discounts and/or surcharges for buying multiple units of a product (the latter may make sense for high-valued products such as tickets to events at the Olympics). More generally, an algorithm can may use discounts and/or surcharges for some \emph{bundles} of products, where each bundle can include multiple units of multiple products, \eg two beers and one snack. In full generality, there is a collection $\mF$ of allowed bundles. In each round an algorithm offers a \emph{menu} of options which consists of a price for every allowed bundle in $\mF$ (and the ``none" option), and the agent chooses one option from this menu. Thus, in each round the algorithm needs to choose a price vector over the allowed bundles.

For a formal result, assume there is a finite set $S$ of allowed price vectors, each bundle in $\mF$ can contain at most $\ell$ units total, and the per-bundle prices are restricted to lie in the range $[0,\ell]$. Then we obtain $S$-regret
    $\Otilde(d\ell \sqrt{\ell\,B\,|S|})$.
This follows from Lemma~\ref{lm:BwK-rescaled}, observing that per-round rewards are at most $r_0=\ell$, per-round consumption of each resource is at most $c_0=\ell$, and the optimal value is $\OPT\leq d\ell B$.

The action space here is $|\mF|$-dimensional, which may result in a prohibitively large number of allowed price vectors. One can reduce the ``dimensionality" of the action space by restricting how the bundles may be priced. For example, each bundle may be priced at a volume discount $x\%$ compared to buying each unit separately, where $x$ depends only on the number of items in the bundle.

Moreover, we can analyze \preDiscr for a version where in each round the algorithm chooses only one bundle to offer. By Theorem~\ref{thm:DynPricing-discretization}, we obtain regret
    $\Otilde(B^{2/3} \, (|\mF|\ell)^{1/3})$.

\xhdr{Extension: buyer targeting.} Suppose there are $\ell$ different types of buyers (say {\em men} and {\em women}), and the demand distribution of a buyer depends on her type. The buyer type is modeled as a sample from a fixed but unknown distribution. In each round the seller observes the type  of the current buyer (e.g., using a \emph{cookie} or a user profile), and can choose the price depending on this type.

    This can be modeled as a BwK domain where arms correspond to functions from buyer types to prices. For example, with $\ell$ buyer types and a single product, the (full) action space is $X = [0,1]^\ell$. Assuming we are given a restricted action space $S\subset X$, we obtain $S$-regret $\Otilde(\sqrt{B\,|S|})$.

\OMIT{
Our regret guarantees for the above extensions are with respect to the \OptPolicy on the restricted action space. It is worth emphasizing that the benchmark equivalence result (the best fixed action is almost as good as the \OptPolicy, for regular demand distributions) applies only to the basic version of dynamic pricing. For example, no such result is known for multiple types of goods, even if the demand distribution for each individual type is regular. (See Appendix~\ref{ap:example} for a related discussion.)
} 

\subsection{Dynamic procurement and crowdsourcing markets}
\label{sec:apps-DynProcurement}

A ``dual" problem to dynamic pricing is \emph{dynamic procurement}, where the algorithm is buying rather than selling. In the basic version, the algorithm has a budget $B$ to spend, and is facing $T$ agents (potential sellers) that are arriving sequentially. In each round $t$, a new agent arrives, the algorithm chooses a price $p_t\in [1]$ and offers to buy one item at this price. The agent has private value $v_t\in [0,1]$ for an item (unknown to the algorithm), and sells if and only if $p_t \geq v_t$. The value is an independent sample from some fixed (but unknown) distribution. Algorithm's goal is to maximize the number of items bought. Recall from
    Theorem~\ref{thm:DynProcurement-discretization} that
we obtain regret $\Otilde\left( T/B^{1/4} \right)$ for this version.

\xhdr{Application to crowdsourcing markets.}
The problem is particularly relevant to the emerging domain of \emph{crowdsourcing}, where agents correspond to the (relatively inexpensive) workers on a crowdsourcing platform such as Amazon Mechanical Turk, and ``items" bought/sold correspond to simple jobs (``microtasks") that can be performed by these workers. The algorithm corresponds to the ``requester": an entity that submits jobs and benefits from them being completed. The (basic) dynamic procurement model captures an important issue in crowdsourcing that a requester interacts with multiple users with unknown values-per-item, and can adjust its behavior (such as the posted price) over time as it learns the distribution of users. While this basic model ignores some realistic features of crowdsourcing environments (see a survey \citet{Crowdsourcing-PositionPaper13} for background and discussion), some of these limitations are addressed by the generalizations which we present below.

\xhdr{Extension: non-unit supply.} We consider an extension where agents may be interested in more than one item, and their valuations may be non-linear. For example, a worker may be interested in performing several jobs. In each round $t$, the algorithm offers to buy up to $\Lambda$ units at a fixed price $p_t$ per unit, where the price $p_t$ is chosen by the algorithm and $\Lambda$ is a fixed parameter. The $t$-th agent then chooses how many units to sell. Recall from
    Theorem~\ref{thm:DynProcurement-discretization} that
we obtain regret
    $\tilde{O}(\Lambda^{5/4} T/B^{1/4})$
for this extension.

\xhdr{Extension: multiple types of jobs}. We can handle an extension in which there are $d$ types of jobs requested on the crowdsourcing platform, with a separate budget $B_i$ for each type. Each agent $t$ has a private cost $v_{t,i}\in [0,1]$ for each type $i$; the vector of private costs comes from a fixed but unknown distribution over $d$-dimensional vectors (note that  arbitrary correlations are allowed). The algorithm derives reward $u_i\in [0,1]$ from each job of type $i$. In each round $t$, the algorithm offers a vector of prices
    $(p_{t,1} \LDOTS p_{t,d})$,
where $p_{t,i}$ is the price for one job of type $i$. For each type $i$, the agent performs one job of this type if and only if $p_{t,i}\geq v_{t,i}$, and receives payment $p_{t,i}$ from the algorithm.

Here arms correspond to the $d$-dimensional vectors of prices, so that the action space is $X = [0,1]^d$. Given the restricted action space $S\subset X$, we obtain $S$-regret
        $\Otilde(d)(\sqrt{T\,|S|} + T\sqrt{d\,|S|/B})$,
where $B$ is the smallest budget. This follows from Lemma~\ref{lm:BwK-rescaled}, observing that per-round rewards are at most $r_0=d$, per-round consumption of each budget is at most $c_0=1$, and the optimal value is $\OPT\leq dT$.

\xhdr{Extension: additional features.} We can also model more complicated ``menus" so that each agent can perform several jobs of the same type. Then in each round, for each type $i$, the algorithm specifies the maximal offered number of jobs of this type and the price per one such job. We can also incorporate constraints on  the maximal number of jobs of each type that is needed by the requester, and/or the maximal amount of money spend on each type.

\xhdr{Extension: competitive environment.} There may be other requesters in the system, each offering its own vector of prices in each round. (This is a realistic scenario in crowdsourcing, for example.) Each seller / worker chooses the requester and the price that maximize her utility. One standard way to model such a competitive environment is to assume that the ``best offer" from the competitors is a vector of prices which comes from a fixed but unknown distribution. This can be modeled as a \BwK instance with a different distribution over outcomes which reflects the combined effects of the demand distribution of agents and the ``best offer" distribution of the environment.

\subsection{Other applications to Electronic Markets}
\label{sec:apps-other}

\xhdr{Ad allocation with unknown click probabilities.}
Consider \emph{pay-per-click} (PPC) advertising on the web (in particular, this is a prevalent model in sponsored search auctions). The central premise in PPC advertising is that an advertiser derives value from her ad only when the user clicks on this ad. The ad platform allocates ads to users that arrive over time.

Consider the following simple (albeit highly idealized) model for PPC ad allocation. Users arrive over time, and the ad platform needs to allocate an ad to each arriving user. There is a set $X$ of available ads. Each ad $x$ is characterized by the payment-per-click $\pi_x$ and click probability $\mu_x$; the former quantity is known to the algorithm, whereas the latter is not. If an ad $x$ is chosen, it is clicked on with probability $\mu_x$, in which case payment $\pi_x$ is received. The goal is to maximize the total payment. This setting and various extensions thereof that incorporate user/webpage context have received a considerable attention in the past several years (starting with \citep{yahoo-bandits07,yahoo-bandits-icml07,Langford-nips07}). In fact, the connection to PPC advertising has been one of the main motivations for the recent surge of interest in MAB.

We enrich the above setting by incorporating advertisers' \emph{budgets}. In the most basic version, for each ad $x$ there is a budget $B_x$ --- the maximal amount of money that can be spent on this ad. More generally, an advertiser can have an ad campaign which consists of a subset $S$ of ads, so that there is a per-campaign budget $S$. Even more generally, an advertiser can have a more complicated budget structure: a family of overlapping subsets $S\subset X$ and a separate budget $B_S$ for each $S$. For example, BestBuy can have a total budget for the ad campaign, and also separate budgets for ads about TVs and ads about computers. Finally, in addition to budgets (i.e., constraints on the number of times ads are clicked), an advertiser may wish to have similar constraints on the number of times ads are shown. \BwK allows us to express all these constraints.

\xhdr{Adjusting a repeated auction.}
An auction is held in every round, with a fresh set of participants. The number of participants and a vector of their types come from a fixed but unknown distribution. The auction is \emph{adjustable}: it has some parameter that the auctioneer adjust over time so as to optimize revenue. For example, \citet{RepeatedAuctions-soda13} studies a repeated second price auction with an adjustable reserve price, with unlimited inventory of a single product. \BwK framework allows to incorporate limited inventory of items to be sold at the auction, possibly with multiple products.

\xhdr{Repeated bidding.} A bidder participates in a repeated auction, such as a sponsored search auction. In each round $t$, the bidder can adjust her bid $b_t$ based on the past performance. The outcome for this bidder is a vector $(p_t, u_t)$, where $p_t$ is the payment and $u_t$ is the utility received. We assume that this vector comes from a fixed but unknown distribution. The bidder has a fixed budget. Similar setting have been studied in \citep{AminK-uai12,TranThanh-uai14}, for example.

We model this as a \BwK problem where arms correspond to the possible bids, and the single resource is money. Note that (the basic version of) dynamic procurement corresponds to this setting with two possible outcome vectors $(p_t,u_t)$: $(0,0)$ and $(b_t,1)$.

The \BwK setting also allows to incorporate more complicated constraints. For example, an action can result in several different types of outcomes that are useful for the bidder (e.g., an ad shown to a \emph{male} or an ad shown to a \emph{female}), but the bidder is only interested in a limited quantity of each outcome.

\subsection{Application to network routing and scheduling}
\label{sec:apps-scheduling}

In addition to applications to Electronic Markets, we describe two applications to network routing and scheduling. In both applications an algorithm chooses between different feasible policies to handle arriving ``service requests'', such as connection requests in network routing and jobs in scheduling.

\xhdr{Adjusting a routing protocol.} Consider the following stylized application to routing in a communication network. Connection requests arrive one by one. A connection request consists of a pair of terminals; assume the pair comes from a fixed but unknown distribution. The system needs to choose a \emph{routing protocol} for each connection, out of several possible routing protocols. The routing protocol defines a path that connects the terminals; abstractly, each protocol is simply a mapping from terminal pairs to paths. Once the path is chosen, a connection between the terminals is established. Connections persist for a significant amount of time. Each connection uses some amount of bandwidth. For simplicity, we can assume that this amount is fixed over time for every connection, and comes from a fixed but unknown distribution (although even a deterministic version is interesting). Each edge in the network (or perhaps each node) has a limited capacity: the total bandwidth of all connections that pass though this edge or node cannot exceed some value. A connection which violates any capacity constraint is terminated.  The goal is to satisfy a maximal number of connections.

We model this problem as \BwK as follows: arms correspond to the feasible routing protocols, each edge/node is a limited resource, each satisfied connection is a unit reward.

Further, if the time horizon is partitioned in epochs, we can model different bandwidth utilization in each phase; then a resource in \BwK is a pair (edge,epoch).
\OMIT{\footnote{This form of the problem generalizes a network routing problem that was communicated to us by Luyi Gui~\cite{Gui-comm}.}}

\xhdr{Adjusting a scheduling policy.} An application with a similar flavor arises in the domain of scheduling long-running jobs to machines. Suppose jobs arrive over time. Each job must be assigned to one of the machines (or dropped); once assigned, a job stays in the system forever (or for some number of ``epochs"), and consumes some resources.  Jobs have multiple ``types" that can be observed by the scheduler. For each type, the resource utilization comes from a fixed but unknown distribution. Note that there may be multiple resources being consumed on each machine: for example, jobs in a datacenter can consume CPU, RAM, disk space, and network bandwidth. Each satisfied job of type $i$ brings utility $u_i$.  The goal of the scheduler is to maximize utility given the constrained resources.

The mapping of this setting to \BwK is straightforward. The only slightly subtle point is how to define the arms: in \BwK terms, arms correspond to all possible mappings from job types to machines.

One can also consider an alternative formulation where there are several allowed scheduling policies (mappings from types and current resource utilizations to machines), and in every round the scheduler can choose to use one of these policies. Then the arms in \BwK correspond to the allowed policies.

\hspace*{1cm} \\[1em]
\noindent {\Large \bf Acknowledgements} \\[1em]
The authors wish to thank Moshe Babaioff,
Peter Frazier, Luyi Gui, Chien-Ju Ho and Jennifer Wortman Vaughan
for helpful discussions related to this work. In particular, the application to routing protocols generalizes a network routing problem that was communicated to us by Luyi Gui. The application of dynamic procurement to crowdsourcing have been suggested to us by Chien-Ju Ho and Jennifer Wortman Vaughan. We are grateful to anonymous JACM referees for their thorough and insightful feedback.

\addcontentsline{toc}{section}{\bibname}

\bibliographystyle{plainnat}
\bibliography{bib-abbrv,bib-slivkins,bib-bandits,bib-AGT}

\appendix
\section{The \OptPolicy beats the best fixed arm}
\label{ap:example}

Let us provide additional examples of \BwK problem instances in which the \OptPolicy (in fact, the best fixed distribution over arms) beats the best fixed arm.

\xhdr{Dynamic pricing.}
Consider the basic setting of ``dynamic pricing with limited supply'': in each round a potential buyer arrives, and the seller offers him one item at a price; there are $k$ items and $n>k$ potential buyers. One can easily construct distributions for which
offering a mixture of two prices is strictly superior to offering
any fixed price. In fact this situation arises whenever the ``revenue
curve'' (the mapping from prices to expected revenue) is non-concave and its value at the quantile $k/n$ lies below its concave hull.

Consider a simple example: fix $\eps=k^{\delta-1/2}$ with $\delta\in(0,\tfrac12)$, and assume that the buyer's value for an item is $v=1$ with probability $\eps\,\tfrac{k}{n}$ and $v=\eps$ with the remaining probability, for some fixed $\eps\in(0,1)$.

To analyze this example, let $\Rew(\D)$ be the expected total reward (i.e., the expected total revenue) from using a fixed distribution $\D$ over prices in each round; let $\Rew(p)$ be the same quantity when $\D$ deterministically picks a given price $p$.

\begin{itemize}
\item Clearly, if one offers a fixed price in all rounds, it only makes sense to offer prices $p=\eps$ and $p=1$.
It is easy to see that
 $\Rew(\eps)=\eps k$
and
$\Rew(1) \leq n\cdot \Pr[\text{sale at price $1$ }] = \eps k$.

\item Now consider a distribution $\D$ which picks price $\eps$ with probability $(1-\eps)\tfrac{k}{n}$,
and picks price $1$ with the remaining probability. It is easy to show that
        $\Rew(D) \geq \eps k(2-o(1))$.
\end{itemize}

So, $\Rew(\D)$ is essentially twice as large compared to the total expected revenue of the best fixed arm.

\xhdr{Dynamic procurement.}
A similar example can be constructed in the domain of dynamic procurement. Consider the basic setting thereof: in each round a potential seller arrives, and the buyer offers to buy one item at a price; there are $T$ sellers and the buyer is constrained to spend at most budget $B$. The buyer has no value for left-over budget and each sellers value for the item is drawn i.i.d from an unknown distribution. Then a mixture of two prices is strictly superior to offering
any fixed price whenever the ``sales curve'' (the mapping from prices to probability of selling) is non-concave and its value at the quantile $B/T$ lies below its concave hull.

Let us provide a specific example. Fix any constant $\delta>0$, and let $\eps=B^{1/2+\delta}$. Each seller has the following two-point demand distribution: the seller's value for item is $v=\asedit{0}$ with probability $\tfrac{B}{T}$, and $v=\asedit{1}$ with the remaining probability. We use the notation $\Rew(\D)$ and $\Rew(p)$ as defined above.

\begin{itemize}
\item Clearly, if one offers a fixed price in all rounds, it only makes sense to offer prices $p=0$ and $p=1$.
It is easy to see that
 $\Rew(0)\leq T\cdot \Pr[\text{selling at price $0$}]=B$
and
$\Rew(1)=B$.

\item Now consider a distribution $\D$ which picks price $0$ with probability $1-\frac{B-\eps}{T}$, and picks price $1$ with the remaining probability. It is easy to show that
        $\Rew(\D) \geq (2-o(1))\, B$.
\end{itemize}

Again, $\Rew(\D)$ is essentially twice as large compared to the total expected sales of the best fixed arm.

\section{\kMAB beats \pdbwk sometimes}
\label{app:Balance-wins}

We provide a simple example in which \kMAB achieves much better regret than (what we can prove for) \pdbwk. The reason is that \kMAB is aware of the BwK domain, whereas \pdbwk is not. More precisely, \kMAB is parameterized by $\domain$, the set of all latent structures that are feasible for the BwK domain.

The example is a version of the deterministic example from Section~\ref{sec:problem-description}. There is a time horizon $T$ and two other resources, both with budget $B<T/2$. There are $m$ arms, partitioned into two same-size subsets, $X_1$ and $X_2$. Per-round rewards and per-round resource consumptions are deterministic for all arms. All arms get per-round reward $1$. For each resource $i$, each arm in $X_i$ only consumes this resource. Letting $c_i(x)=c_i(x,\mu)$ denote the (expected) per-round consumption of resource $i$ by arm $x$, one of the following holds:
\begin{itemize}
\item[(i)] $c_1(x_1) =1$ and $c_2(x_2)=\tfrac12$ for all arms $x_1\in X_1, x_2\in X_2$, or
\item[(ii)] $c_1(x_1) =\tfrac12$ and $c_2(x_2)=1$ for all arms $x_1\in X_1, x_2\in X_2$.
\end{itemize}

\xhdr{Analysis.} Note that an \OptPolicy alternates the two arms in proportion, $1:2$ or $2:1$, depending on the case, and an LP-optimal distribution over arms samples them in the same proportion.

The key argument is that, informally, \kMAB can tell (i) from (ii) after the initial $O(m\log T)$ rounds. In the specification of \kMAB, consider the confidence radius for the per-round consumption of resource $i$, as defined in \eqref{eq:algo-estimate-c}. After any one arm $x$ is played at least $C\,\log T$ rounds, for a sufficiently large absolute constant $C$, this confidence radius goes below $1/4$. Therefore any two latent structures $\mu$, $\mu'$ in the confidence interval satisfy
    $|c_i(x,\mu) - c_i(x,\mu')| <\tfrac12$.
It follows that the confidence interval consists of a single latent structure, either the one corresponding to (i) or the one corresponding to (ii), which is the correct latent structure for this problem instance. Accordingly, the chosen distribution over arms, being ``potentially perfect'' by design, is LP-optimal. Thus, \kMAB uses the LP-optimal distribution over arms after the initial $O(m\log T)$ rounds.

The resulting regret is
    $\tilde{O}(m + \sqrt{B})$,
where the $\sqrt{B}$ term arises because the empirical frequencies of the two arms can deviate by $O(\sqrt{B})$ from the optimal values. Whereas with algorithm $\pdbwk$ we can only guarantee regret
    $\tilde{O}(\sqrt{mB})$.

\section{Analysis of the Hedge Algorithm}
\label{app:hedge}

We provide a self-contained proof of Proposition~\ref{prop:hedge}, the performance guarantee for the $\mathsf{Hedge}$ algorithm from \citet{FS97}. The presentation is adapted from \citet{cs683week2}.

For the sake of convenience, we restate the algorithm and the proposition. It is an online algorithm for maintaining
a $d$-dimensional probability vector $y$ while
observing a sequence of
$d$-dimensional payoff vectors $\pi_1,\ldots,\pi_\stime$.
The algorithm is initialized with a parameter
$\eps \in (0,1)$.
\begin{algorithm}[H]
\myCaption{$\mathsf{Hedge}(\eps)$}
\begin{algorithmic}[1]
\STATE $v_1 = \ones$
\FOR{$t=1,2,\ldots,\stime$}
\STATE $y_t = v_t / (\ones^\trans v_t)$.
\STATE $v_{t+1} = \diag{(1+\eps)^{\pi_{ti}}} v_t$.
\ENDFOR
\end{algorithmic}
\end{algorithm}

The performance guarantee of the algorithm is expressed
by the following proposition.
\begin{proposition*}[Proposition~\ref{prop:hedge}, restated]
\label{prop:hedge-app}
For any $0 < \eps < 1$ and any sequence of payoff
vectors $\pi_1,\ldots,\pi_\stime \in [0,1]^d$, we have
\[
\forall y \in \splx{d} \quad
\sum_{t=1}^\stime y_t^\trans \pi_t \geq (1-\eps) \sum_{t=1}^\stime y^\trans \pi_t
- \frac{\ln d}{\eps}.
\]
\end{proposition*}
\begin{proof}
The analysis uses the potential function $\Phi_t = \ones^\trans v_t$.
We have
\begin{align*}
\Phi_{t+1} &= \ones^\trans \diag{(1+\eps)^{\pi_{ti}}} v_t \\
&= \sum_{i=1}^d (1+\eps)^{\pi_{ti}} v_{t,i} \\
&\leq \sum_{i=1}^d (1 + \eps \pi_{ti}) v_{t,i} \\
&= \Phi_t \left( 1 + \eps y_t^\trans \pi_t \right) \\
\ln(\Phi_{t+1}) & \leq \ln(\Phi_t) + \ln(1 + \eps y_t^\trans \pi_t) \leq
\ln(\Phi_t) + \eps y_t^\trans \pi_t.
\end{align*}
On the third line, we have used the inequality $(1+\eps)^x \leq 1+\eps x$
which is valid for $0 \leq x \leq 1$. Now, summing over $t=1,\ldots,\stime$
we obtain
\[
\sum_{t=1}^\stime y_t^\trans \pi_t \geq
\tfrac{1}{\eps} \left( \ln \Phi_{\stime+1} - \ln \Phi_1 \right) =
\tfrac{1}{\eps} \ln \Phi_{\stime+1} - \frac{\ln d}{\eps}.
\]
The maximum of
$y^\trans \left( \sum_{t=1}^\stime \pi_t \right)$
over $y \in \splx{d}$ must be attained at one
of the extreme points of $\splx{d}$, which are simply
the standard basis vectors of $\R^d$. Say that the
maximum is attained at $\sbv_i$. Then we have
\begin{align*}
\Phi_{\stime+1} & = \ones_\trans v_{\stime+1} \geq v_{\stime+1,i}
= (1+\eps)^{\pi_{1i} + \cdots + \pi_{\stime i}} \\
\ln \Phi_{\stime+1} & \geq \ln(1+\eps) \sum_{t=1}^\stime \pi_{ti} \\
\sum_{t=1}^\stime y_t^\trans \pi_t &\geq
\frac{\ln(1+\eps)}{\eps} \sum_{t=1}^\stime \pi_{ti}
 - \frac{\ln d}{\eps} \\
& \geq
(1-\eps) \sum_{t=1}^\stime y^\trans \pi_t
 - \frac{\ln d}{\eps}.
\end{align*}
The last line follows from two observations. First, our
choice of $i$ ensures that $\sum_{t=1}^\stime \pi_{ti} \geq
\sum_{t=1}^\stime y^\trans \pi_t$ for every $y \in \splx{d}$.
Second, the inequality $\ln(1+\eps) > \eps-\eps^2$ holds
for every $\eps > 0$. In fact,
\begin{align*}
-\ln(1+\eps) &=
\ln \left( \frac{1}{1+\eps} \right) =
\ln \left( 1 - \frac{\eps}{1+\eps} \right) <
- \frac{\eps}{1+\eps} \\
\ln(1+\eps) &> \frac{\eps}{1+\eps} > \frac{\eps(1-\eps^2)}{1+\eps} =
\eps- \eps^2. \qedhere
\end{align*}
\end{proof}

\section{Facts for the proof of the lower bound}
\label{app:facts}

For the sake of completeness, we provide self-contained proofs for the two facts used in Section~\ref{sec:LB}.

\begin{fact*}[Fact~\ref{fact:stopping}, restated]
Let $S_t$ be the sum of $t$ i.i.d. 0-1 variables with expectation $q$. Let $\tau$ be the first time this sum reaches a given number $B\in \N$. Then $\E[\tau] = B/q$. Moreover, for each $T>\E[\tau]$ it holds that
\begin{align}\label{eq:app-fact-int}
    \textstyle \sum_{t>T} \; \Pr[\tau\geq t] \leq \E[\tau]^2/T.
\end{align}
\end{fact*}
\begin{proof}
$\E[\tau] = B/q$ follows from the martingale argument presented in the proof of Claim~\ref{cl:LB-Fam-stopping}. Formally, take $q=p-\eps$ and $N_\tau = \tau$.

\OMIT{Define $b_t$ as the random variable which denotes the amount of resource used at the end of time $t$.}

Assume $T>\E[\tau]$. The proof of \eqref{eq:app-fact-int} uses two properties, one being that a geometric random variable is memoryless and other being Markov's inequality.
Let us first bound the random variable $\tau-T$ conditional on the event that $\tau>T$.
\begin{align*}
\E[\tau-T|\tau>T]
    &= \textstyle \sum_{t=1}^B\; \Pr[S_T=t] \; \E[\tau-T|\tau>T,S_T=t] \\
    &=\textstyle  \sum_{t=1}^B\; \Pr[S_T=t] \; \E[\tau-T|S_T=t] \\
    &\leq \textstyle  \sum_{t=1}^B\; \Pr[S_T=t] \; \E[\tau-T|S_T=0] \\
    &\leq\E[\tau-T|S_T=0]=\E[\tau].
\end{align*}

By Markov's inequality we have
    $\Pr[\tau\geq T]\leq \frac{\E[\tau]}{T}$.
Combining the two inequalities, we have
\begin{align*}
\textstyle \sum_{t>T}\;\Pr[\tau\geq t]
    &= \textstyle \sum_{t>T}\; \Pr[\tau\geq T]\;\; \Pr[\tau\geq t| \tau>T] \\
    &=\Pr[\tau\geq T] \;\; \E[\tau-T|\tau> T] \\
    &\leq \E[\tau]^2/T. \qquad\qquad \qedhere
\end{align*}
\end{proof}

\begin{fact*}[Fact~\ref{fact:logs}, restated]
Assume
    $\tfrac{\eps}{p}\leq \tfrac12 $ and $p\leq \tfrac12$.
Then
$$ p\,\log \left( \frac{p}{p-\eps} \right) +
    (1-p)\,\log \left( \frac{1-p}{1-p+\eps} \right)
    \leq \frac{2 \eps^2}{p}.
$$
\end{fact*}
\begin{proof}
To prove the inequality we use the following standard inequalities:
\begin{align*}
\log(1+x) &\geq x-x^2/2    &\forall x\in[0,1]\\
\log(1-x) &\geq -x-x^2     &\forall x\in[0,\tfrac12].
\end{align*}
It follows that:
\begin{align*}
p\,\log \left( \frac{p}{p-\eps} \right) +
    (1-p)\,\log \left( \frac{1-p}{1-p+\eps} \right)
&= - p\, \log\left(1-\frac{\eps}{p}\right) - (1-p)\, \log\left(1+\frac{\eps}{1-p}\right) \\
&\leq p\, \left(\frac{\eps}{p}+\frac{\eps^2}{p^2}\right)+(1-p)\, \left(-\frac{\eps}{1-p}+\frac{\eps^2}{(1-p)^2}\right) \\
&= \frac{\eps^2}{p}+\frac{\eps^2}{1-p}
\leq \frac{2 \eps^2}{p} \qedhere
\end{align*}
\end{proof}

\end{document}